\documentclass[a4paper,UKenglish, cleveref, autoref, thm-restate, numberwithinsect]{lipics-v2021}
\hideLIPIcs  %uncomment to remove references to LIPIcs series (logo, DOI, ...), e.g. when preparing a pre-final version to be uploaded to arXiv or another public repository

\usepackage{etoolbox}   % only needed for \patchcmd
\patchcmd{\thebibliography}{#1}{A99}{}{}   % <- put your desired widest label here
\title{The Expiration Streaming Model: Diameter, $k$-Center, Counting, Sampling, and Friends}

\titlerunning{The Expiration Streaming Model}

\author{Lotte Blank}{University of Bonn, Germany.}{lblank@uni-bonn.de}{https://orcid.org/0000-0002-6410-8323}{Funded by the Deutsche Forschungsgemeinschaft (DFG, German Research Foundation) – 459420781 (FOR AlgoForGe).}
\author{Sergio Cabello}{University of Ljubljana, Slovenia \and Institute of Mathematics, Physics and Mechanics, Ljubljana, Slovenia}{sergio.cabello@fmf.uni-lj.si}{https://orcid.org/0000-0002-3183-4126}{Funded in part by the Slovenian Research and Innovation Agency (P1-0297, N1-0218, N1-0285, J1-70045). Funded in part by the European Union (ERC, KARST, project number 101071836). Views and opinions expressed are however those of the authors only and do not necessarily reflect those of the European Union or the European Research Council. Neither the European Union nor the granting authority can be held responsible for them.}
\author{{Mohammad Taghi} Hajiaghayi}{University of Maryland, College Park, MD, USA}{hajiagha@cs.umd.edu}{https://orcid.org/0000-0003-4842-0533}{The work is partially supported by DARPA expMath, ONR MURI 2024 award on Algorithms, Learning, and Game Theory, Army-Research Laboratory (ARL) grant W911NF2410052, NSF AF:Small grants 2218678, 2114269, 2347322.}
\author{Robert Krauthgamer}{The Harry Weinrebe Professorial Chair of Computer Science,\and Weizmann Institute of Science, Rehovot, Israel}{robert.krauthgamer@weizmann.ac.il}{https://orcid.org/0009-0003-8154-3735}{Work partially supported by the Israel Science Foundation grant \#1336/23, by the Israeli Council for Higher Education (CHE) via the Weizmann Data Science Research Center, and by a research grant from the Estate of Harry Schutzman.}
\author{Sepideh Mahabadi}{Microsoft Research, Redmond, WA, USA}{smahabadi@microsoft.com}{https://orcid.org/0000-0001-5004-8991}{}
\author{André Nusser}{Université Côte d'Azur, CNRS, Inria, Sophia Antipolis, France}{andre.nusser@cnrs.fr}{https://orcid.org/0000-0002-6349-869X}{Supported by the French government through the France 2030 investment plan managed by the National Research Agency (ANR), as part of the Initiative of Excellence of Université Côte d’Azur under reference number ANR-15-IDEX-01.}
\author{{Jeff M.} Phillips}{University of Utah, Salt Lake City, UT, USA}{jeffp@cs.utah.edu}{https://orcid.org/0000-0003-1169-2965}{Work supported by funding from NSF 2115677 and 2421782, and Simons Foundation MPS-AI-00010515.}
\author{Jonas Sauer}{Karlsruhe Institute of Technology, Germany.}{jonas.sauer@kit.edu}{https://orcid.org/0000-0002-7196-7468}{Funded by the Deutsche Forschungsgemeinschaft (DFG, German Research Foundation) – 459420781 (FOR AlgoForGe).}

\authorrunning{Blank, Cabello, Hajiaghayi, Krauthgamer, Mahabadi, Nusser, Phillips, Sauer} %TODO mandatory. First: Use abbreviated first/middle names. Second (only in severe cases): Use first author plus 'et al.'

\Copyright{TODO} %TODO mandatory, please use full first names. LIPIcs license is "CC-BY";  http://creativecommons.org/licenses/by/3.0/

 %TODO mandatory: Please choose ACM 2012 classifications from https://dl.acm.org/ccs/ccs_flat.cfm 

 %TODO mandatory; please add comma-separated list of keywords

\category{} %optional, e.g. invited paper

\relatedversion{} %optional, e.g. full version hosted on arXiv, HAL, or other respository/website
\funding{This research was initiated during the Workshop “Massive Data Models and Computational Geometry” held at the University of Bonn in September 2024 and funded by the DFG, German Research Foundation, through EXC 2047 Hausdorff Center for Mathematics and FOR 5361:  KI-FOR Algorithmic Data Analytics for Geodesy (AlgoForGe).}
\acknowledgements{We thank the organizers and participants of the Bonn workshop for the stimulating environment that inspired this research. We also thank the anonymous reviewers for useful references and comments that improve the exposition. }%optional

\nolinenumbers %uncomment to disable line numbering

\EventEditors{John Q. Open and Joan R. Access}
\EventNoEds{2}
\EventLongTitle{42nd Conference on Very Important Topics (CVIT 2016)}
\EventShortTitle{CVIT 2016}
\EventAcronym{CVIT}
\EventYear{2016}
\EventDate{December 24--27, 2016}
\EventLocation{Little Whinging, United Kingdom}
\EventLogo{}
\SeriesVolume{42}
\ArticleNo{23}
\usepackage{graphicx}
\graphicspath{{./}{../}}

\usepackage{xcolor}
\usepackage{xspace}
\usepackage{csquotes}
\usepackage{amsmath}
\usepackage{amssymb}
\usepackage{url}
\usepackage[ruled, linesnumbered, vlined]{algorithm2e}
\usepackage{appendix}
\usepackage{bbm}
\usepackage{lineno}
\SetKwProg{myproc}{Procedure}{}{}
\SetKwFunction{Insert}{Insert}
\SetKwFunction{AddAttractionPoint}{AddAttractionPoint}
\SetKwFunction{DiscardCoveredPoints}{DiscardCoveredPoints}
\DontPrintSemicolon

\usepackage{pgfplots}

\usepackage{amsthm}
\DeclareMathOperator{\poly}{poly}

\DeclareMathOperator{\diam}{diam}

\DeclareMathOperator{\sgn}{sgn}
\newcommand{\tO}{\tilde{\calO}}

\def\compactify{\itemsep=0pt \topsep=0pt \partopsep=0pt \parsep=0pt}

\newcommand{\RR}{\mathbb{R}}
\newcommand{\X}{\mathcal{X}}
\newcommand{\eps}{\varepsilon}
\newcommand{\etal}{et al.\xspace}

\newcommand{\maxK}{{\ensuremath{K}}}
\newcommand{\parK}{{\ensuremath{k}}}
\newcommand{\Abig}{\ensuremath{\hat{t}}}
\newcommand{\Rcal}{\ensuremath{\mathcal{R}}}
\newcommand{\calM}{\ensuremath{\mathcal{M}}}
\newcommand{\calO}{\ensuremath{\mathcal{O}}}
\newcommand{\E}{\ensuremath{\mathsf{E}}}

\newcommand{\kde}{\ensuremath{\textsc{kde}}}

\begin{document}

\maketitle

\begin{abstract}
An important thread in the study of data-stream algorithms
focuses on settings where stream items are active only for a limited time. 
We introduce a new 
\emph{expiration model}, where each item arrives with its own arbitrary expiration time. 
The special case where items expire in the order that they arrive,
which we call consistent expirations, 
contains the classical sliding-window model
of Datar, Gionis, Indyk, and Motwani [SICOMP 2002]
and its timestamp-based variant of Braverman and Ostrovsky [FOCS 2007].

Our first set of results explores the expiration streaming model
and presents algorithms for several fundamental problems,
including approximate counting, uniform sampling, and weighted sampling
by efficiently tracking active items without explicitly storing them all.
Naturally, these algorithms have many immediate applications, e.g., to range counting.

Our second and main set of results for the expiration model 
designs algorithms for the diameter and $k$-center problems,
where items are points in a metric space. 
Our results significantly extend those known for the special case of sliding-window streams 
by Cohen-Addad, Schwiegelshohn, and Sohler [ICALP 2016],
and obtain a strictly better approximation factor for the diameter
in the important special case of high-dimensional Euclidean metrics.  
We develop new decomposition and coordination techniques along with a geometric dominance framework 
to filter out redundant points based on both temporal and spatial proximity. 
\end{abstract}

\newpage
\setcounter{page}{1}

\section{Introduction}
%\andre{In the intro we sometimes use "expiration" and sometimes "general expiration" to refer to our new setting. In the rest of the paper we almost consistently use "general expiration". I think we should use that everywhere. The introduction is already very dense and so making the terms more self-explanatory helps.}
%\rnote{Let's call it "expiration model" and add the word general only when we contrast it with special cases like the consistent case. I suggest to also say it in our Intro.}
The sliding-window streaming model is widely used to represent a time-sensitive stream,
i.e., a sequence of data items that arrive over time and are active only for a limited time.
In the classical formulation of this model~\cite{BabcockMR02,DGIM02}, 
the $w$ most recent items, for a parameter $w$, form an \emph{active window},
and all queries are applied to this active window.
Hence, non-active items are also called \emph{expired}.
It is convenient to think of the items as arriving at successive time steps $t=1,2,\ldots$, 
and thus every item is active for exactly $w$ time steps.
Another variant of this model allows blank time steps where no item arrives,
which is essentially like having discrete items in a continuous time horizon,
and in this case the active window's size (number of items) might vary over time.
Most theoretical research has focused on the first variant above,
called sequence-based windows,
but, as explained, most results easily extend to the second variant,
called timestamp-based windows.
For instance, see~\cite{BO07}, which mentions both  but explicitly analyzes only the first variant.

We introduce a significantly richer \emph{expiration model}, where each item arrives with its own expiration time.
We stress that the order in which items expire (i.e., become non-active) is arbitrary.
This generality befits numerous scenarios where items are heterogeneous
in terms of data type, reliability, origin, policy settings, and so forth.
%\jonas{With the related work section now in the appendix, we should say a little more about potential applications.}
For example, think of credentials in a computer network or graphics objects on a screen.
%e.g., click or ad impression, security events, or sensor readings
%This aspect can potentially be much more difficult for algorithms
A notable special case is the \emph{consistent expiration model}, where items expire in the order in which they arrive.
We sometimes emphasize that we consider the general case, 
where expirations need not be consistent, 
by referring to it as the \emph{general expiration model}.
% To distinguish this from the general (i.e., non-consistent) case, we sometimes refer to the general case as the \emph{general expiration model}.
The consistent expiration model contains the classical sliding-window model discussed above, where all items are active for the same duration~$w$.
%Many sliding-window algorithms barely depend on $w$
% (e.g., their space bound has $\log w$ factors), or merely on the fact that items expire in the order they arrive (not requiring that exactly $w$ items are active).
% One may thus expect that sliding-window algorithms typically extend, with minimal effort, to the consistent-expiration model.
Many sliding-window algorithms barely depend on $w$ (e.g., their space bound has $\log w$ factors). In fact, some merely use the property that the items expire in the order in which they are inserted but not that exactly $w$ items are active at any point in time, and therefore these algorithms carry over immediately to the consistent expiration model. 
However, extending them to the general expiration model seems considerably more challenging, 
as new ideas seem necessary to handle items with non-consistent expirations.
It is an intriguing question which problems are indeed harder in this new model and how,
i.e., whether they truly require more storage or merely a more sophisticated algorithm. 
% For example, while counting active items in the sliding-window model needs only a single counter, 
% supporting the same counting query at arbitrary future times in the expiration model requires $\Omega(n)$  space, which we show via a reduction from INDEXING.
% \rnote{I don't see which SW counting requires a single counter (maybe insertion-only?),
%   we can probably delete this last sentence, as it's discussed/repeated in section 1.1 para 1. 
% }

Another related model is \emph{turnstile streaming},
where the input stream consists of item insertions and deletions.
%(hence each stream element is called an update instead of an item arrival).
The special case where an item can only be deleted after it was inserted is called \emph{strict turnstile} or, especially in geometric and graph settings~\cite{Indyk04, AGM12}, \emph{dynamic streams}.
Although this model bears similarity to our expiration model, it is actually incomparable. 
The crucial difference is that in turnstile streaming,
each deletion triggers the algorithm explicitly at the time of deletion,
whereas in our expiration model the deletions occur implicitly,
because the expiration information is provided only when the item arrives.
Performing the deletions explicitly would require the algorithm
to store the expirations of all active items, which is excessive.
For an in-depth discussion of related models and further related work, we refer to~\Cref{sec:related}.

Let us now formally define our expiration model.
The input stream is a sequence of items $\langle x_1,\dots,x_n\rangle$,
% where each item $x_i=(p_i,S_i,E_i)$ consists of
% an actual data point $p_i$ (e.g., in a metric space),
% an insertion time $S_i \in \NN$ that satisfies $S_1 \leq \ldots \leq S_n$,
% and an expiration time $E_i \in \NN \cup \{\infty\}$ with~$E_i > S_i$.
% To simplify notation, we identify each item $x=(p, S, E)$ with its data point~$p$
% and denote its insertion and expiration times by $S(p)$ and $E(p)$, respectively.
% (Strictly speaking, this notation assumes that the data points are distinct.)
where each item $x$ consists of an actual data point $p(x)$ from some universe $U$
(e.g., a metric space),
an insertion time $S(x)>0$, and an expiration time $E(x)>S(x)$,
where the insertion times must be non-decreasing, i.e., $S(x_1) \leq \ldots \leq S(x_n)$.
For convenience, we assume that the insertion and expiration times are integers,
and we allow items that never expire by having $E(x)=\infty$. 
We sometimes identify each item $x$ with its data point $p=p(x)$,
which slightly abuses notation because data points do not need to be distinct,
and denote its insertion and expiration times by $S(p)$ and $E(p)$, respectively. 
An item $x$ is \emph{active} at time $t$ if $S(x)\leq t < E(x)$,
i.e., from its insertion time up to (but not including) its expiration time.
A query at time $t$ is evaluated only on the set of items active at that time $t$.
For simplicity, we design our algorithms to handle a single query that arrives at an arbitrary time $t>0$ not known in advance to the algorithm.
These algorithms often extend to handle multiple queries.
For deterministic algorithms this is immediate.,
and for randomized algorithms there are standard arguments, such as
probability amplification via independent repetitions. % the median trick.

% This new expirations model contains several well-known models as a special case.
When items never expire, i.e., $E(x)=\infty$ for all~$x$, this is precisely the classical model of insertion-only streams.
We say that the stream has \emph{consistent expirations}
if the expiration times are non-decreasing, i.e., $E(x_1)\leq\ldots\leq E(x_n)$.
When $E(x)-S(x)=w$ for all~$x$,
particularly if insertion times are successive, i.e., every $S(x_i)=i$,
then this is precisely the classical sliding-window model.
In fact, it is convenient to focus on the special case where every $S(x_i)=i$,
which, as mentioned earlier, holds without loss of generality if we allow blank time steps where no item arrives.
  
Our space-complexity bounds count machine words (unless mentioned otherwise), where a word can store a data point (e.g., from a metric space),
a time instantiation (e.g., some $S(x)$),
or a counter in the range $[\poly(n)]$,
where throughout we define $[k] = \{1, \ldots, k\}$. 
This convention avoids bit-representation issues,
although in a simplified case where every $S(x_i)=i$
and data points lie in a universe $U=[\poly(n)]$, 
every word has $\Theta(\log n)$ bits.

\subsection{Results}

%We now present our results.

\subparagraph{Fundamental Problems.}
%\andre{I find that this whole subparagraph is quite dense. Maybe a different formatting or some highlighting works to make it easier to get a grasp of all the results we show.}
%\rnote{I made minor changes but don't know how to improve it significantly.}
% To gauge the power and intricacies of the expiration model,
% we first examine some fundamental streaming problems
% that are standard building blocks for solving many other problems. 
The first problems we consider in the general expiration model are some fundamental streaming problems that are standard building blocks for solving many other problems, thereby gaining a better understanding of the challenges when designing algorithms in our new model.
We start with \emph{counting}, which formally asks for the number of active $1$'s
in a stream of items from the universe $U=\{0,1\}$.
For exact counting, $\Omega(n)$ bits of space are required,
even in the consistent expiration model;
which follows from the known bound for the sliding-window model \cite{DGIM02},
in contrast with insertion-only streams, where $\calO(\log n)$ bits clearly suffice.%
\footnote{An alternative definition,
which asks to count the total number of active items, 
exhibits strict separation between the sliding-window and expiration models. 
In the sliding-window model (without blank time steps),
the answer is always $w$ and thus $\calO(1)$ space suffices,
whereas with expirations, even consistent ones,
the lower bound of $\Omega(n)$ bits still holds.
%\andre{we want to refer to the consistent expiration model here, no?
%Robi: done. }
}
Due to these lower bounds, we turn to \emph{approximate} counting.
We design two randomized streaming algorithms,
one achieves $\eps$-additive error using $\calO(\eps^{-1})$ space,
and the other $\eps$-relative error using $\tO(\eps^{-1} \log (\eps n))$ space,%
\footnote{Throughout, the notation $\tO(f)$ hides logarithmic factors in $f$. }
by employing powerful tools of insertion-only streams, like quantile sketches.
Moreover, we show that these space bounds are tight.
Approximate counting is a useful primitive when designing other algorithms.
We demonstrate this by designing a Count-Min sketch~\cite{cormode2005improved},
which solves frequency estimation (aka point queries) with $\eps$-additive error,
and thus also $\ell_1$-heavy-hitters, using $\calO(\eps^{-2})$ space.
These results appear in \cref{sec:GE-count}.

Another fundamental problem is to \emph{sample $k$ items} uniformly from the stream.
In insertion-only and sliding-window streams, this can be done via reservoir sampling~\cite{vitter1985random,BabcockMR02}.
%In insertion-only streams, this was solved using a famous technique called
%reservoir sampling~\cite{vitter1985random}, which was later extended to sliding-window streams~\cite{BabcockMR02}.
We build on this technique to present 
an algorithm for the expiration model that uses $\calO(k \log n)$ space,
and further present extensions to two more challenging tasks:
sampling \emph{without replacement},
and \emph{weighted sampling}, where each item is sampled with probability proportional to its weight.
Sampling tasks are useful primitives, and we indeed use them to design other algorithms,
e.g., for approximate quantiles with $\eps$-additive error using $\calO(\eps^{-2}\log n)$ space.
We also use them for geometric problems, 
such as range counting, logistic regression, and kernel density estimation (KDE);
here, items are points in $\RR^d$, and our algorithms use space $\calO(\eps^{-2}\log n)$,
assuming
for simplicity
that parameters such as $d$ and the VC-dimension are $\calO(1)$.
%\jonas{This sentence is hard to parse and needs to be broken up.}
%\rnote{done. }
Moreover, these sampling methods apply to matrix-approximation problems,
where each stream item is a row vector in $\RR^d$,
including rank-$k$ approximation and $\eps$-covariance error.
These results appear in \cref{sec:GE-Sampling}.

Technically, these results are less involved and build heavily on prior work. 
They also leave several questions for further investigation. 
For example, our algorithms often use more space than the analogous ones for insertion-only and turnstile streams, 
and thus may possibly be improved, e.g., better dependence on $\eps$.
%although for frequency estimation the term $\log n$ is needed
%even in sliding-window streams \cite{BGLWZ18}.
Note that for frequency estimation in sliding-window streams, near-optimal bounds are known \cite{BGLWZ18}.
Furthermore, our techniques do not yield Count-Sketch-type bounds \cite{charikar2004finding}
for $\ell_2$-point queries % frequency estimation 
and $\ell_2$-heavy-hitters, 
which are known for sliding-window streams \cite{BGLWZ18, FSW25}.  
%We also conjecture it is possible to sample proportional to a weight associated with each item, but do know how to prove it in small space.  
We are also not aware of any strict separation (in space complexity)
between consistent and general expirations.

\subparagraph*{Clustering Problems.}
Our main results (which are technically more challenging)
are for the diameter and $k$-center problems in a general metric space $\calM$,
where $d_\calM$ denotes its distance function
and the subscript may be omitted when clear from the context.
In the streaming setting,
the metric $\calM$ is fixed in advance %e.g., a $k$-dimensional Euclidean space,
and each item $x$ contains a data point $p(x)\in \calM$,
and recall that we may identify $p(x)$ with $x$ and write $x\in\calM$.
To avoid precision issues, we assume that  
%we assume that all distances lie in a range of aspect ratio $\Delta$, 
%and without loss of generality 
$d_\calM(x,y) \in [1,\Delta]$ for all distinct $x,y\in\calM$.%
\footnote{Our results hold even 
if we allow $d_\calM(x,y)=0$, i.e., $\calM$ is a pseudometric.
In fact, the same point may arrive multiple times, possibly with different expirations.
}

In the \emph{diameter problem}, the goal is to compute
$\diam(X) := \max\{ d(x,y)\mid x,y\in X \}$
for the set $X\subset \calM$ of active items.
We devise an algorithm for this problem in the expiration model;
it significantly extends the previously known algorithm,
that works in the more restricted sliding-window model~\cite{Cohen-AddadSS16},
while achieving the same approximation factor and space complexity. 
We prove the following in \cref{sec:diameter}.
\begin{restatable}{theorem}{ThmDiameterGenExp}
\label{thm:diameterGenExp_3apx}
There is a deterministic expiration-streaming algorithm  
that stores $\calO((1/\eps)\log \Delta)$ words 
and maintains a $(3+\eps)$-approximation of the diameter in a general metric $\calM$.
\end{restatable}
Our approximation factor is almost tight, as even in the sliding-window model,
achieving a $(3-\eps)$-approximation for the diameter for any fixed $\eps>0$ requires $\Omega(\sqrt[3]{n})$ space~\cite{Cohen-AddadSS16}.%
\footnote{It is easy to see that $\calO(1)$-approximation requires space complexity $\Omega(\log\Delta)$ already in sliding-window streams, e.g., consider the one-dimensional input where the $i$-th point is either $2^i$ or $0$~\cite{FeigenbaumKZ04}.}

We further improve the approximation factor to $1+\sqrt{3}+\epsilon \approx 2.73$ when the metric space~$\calM$ is Euclidean, regardless of the dimension 
(we only require that each point can be stored in a machine word).
Previously, an approximation factor below $3$ was not known for the Euclidean case,  
even in sliding-window streams.
\begin{restatable}{theorem}{ThmEuclideanGenExpiration}
\label{thm:EuclideanGenExpiration}
There is a deterministic expiration-streaming algorithm 
that stores $\calO((1/\eps)\log \Delta)$ words 
and maintains a $(1+\sqrt{3}+\eps)$-approximation of the diameter in Euclidean space. 
\end{restatable}

Finally, we turn to the \emph{$k$-center problem}, where the goal is to compute
\[
  OPT_k := \min_{C\subset\calM, |C|=k}\ \max_{x\in X} d(x,C), 
\]
where $d(x,C) := \min_{c\in C} d(x,c)$ is the distance from a point $x$ to its closest point in a set $C$. 
We present an algorithm for the expiration model that achieves a $(6k+2)(1+\eps)$-approximation using $\calO(\eps^{-1} k^2 \log\Delta)$ words of space. 
In comparison, previous work achieved 
$(6+\eps)$-approximation using $\calO(\eps^{-1} k \log\Delta)$ words
in the significantly more restricted sliding-window model \cite{Cohen-AddadSS16}. 
We prove the following in \cref{sec:kcenter}. 
%Observe that for $k=1$ we achieve even better approximation, namely, $2$ times our approximation factor for diameter from above, because the ratio between the optimal diameter and $1$-center values is always in the range $[1,2]$. 
%\rnote{I added brief mention that 1-center follows with extra factor 2.}

\begin{restatable}{theorem}{ThmKcenter}
\label{thm:kcenter}
There is a deterministic expiration-streaming algorithm 
that stores $\calO((\maxK^2/\eps)\log \Delta)$ words 
and maintains a $(6\parK+2)(1+\eps)$-approximate solution for $\parK$-center, for every $\parK\leq \maxK$, in a general metric $\calM$. 
\end{restatable}
For $k=1$, we use our diameter algorithm to achieve an improved approximation factor of $4+\eps$ in a general metric (\cref{thm:generalMetricMEB}) and $1+\sqrt{3}+\eps$ in Euclidean space (\cref{thm:EucideanMEB}).

All our algorithms report not only an objective value but also a feasible solution,
i.e., a pair of points (that is approximately the farthest)
or a set of $k$ center points.
\cref{tab:diameter} lists the known approximation factors
for diameter and $1$-center in high-dimensional Euclidean space
under different streaming models, including our new expiration model, 
and two more restricted ones of sliding-window and insertion-only streams.
(We restrict the table to $k=1$ and the Euclidean case to minimize clutter.) 

\begin{table}
%\internallinenumbers % <-- allow numbering inside tabular
\begin{center}
\label{tab:diameter}
\begin{tabular}{|lccl|}
\hline
\multicolumn{1}{|l}{model} &  diameter &  $1$-center (MEB) & references \& comments\\
\hline
\hline
\multicolumn{2}{|l}{\textbf{insertion only:}}  &  & \\
\multicolumn{1}{|l}{} &  $2$& $2$ & folklore; %\cite{Gonzalez85,HochbaumS85}; 
            general metrics\\
&  &  $1.5$ & \cite{Zarrabi-ZadehC06} \\
&  $\sqrt{2}+\eps$&  $\tfrac{1+\sqrt{3}}{2}+\eps 
         \approx 1.37$ & \cite{AS15} \\
& & $1.22+\eps$ & \cite{ChanP14} \\
& $\sqrt{2}+\eps$ & $1.22+\eps$ & \cite{HMV25} \\
\hline
\multicolumn{2}{|l}{\textbf{turnstile/dynamic:}} &  & \\
&  $\calO(\sqrt{\log n})$&  & \cite{Indyk03} \\ 
\hline
\multicolumn{2}{|l}{\textbf{sliding-window:}} & &\\
& $3+\eps$ & $4+\eps$ & \cite{Cohen-AddadSS16}; general metrics \\
& & $9.66+\eps$ & \cite{wangLT19} \\
\hline
\multicolumn{2}{|l}{\textbf{general expiration:}} & &  \\
& $1+\sqrt{3}+\eps \approx 2.73$ & $(1+\sqrt{3}+\eps) \approx 2.73$ & \Cref{thm:EuclideanGenExpiration}, \Cref{thm:EucideanMEB} \\
& $3+\eps$ & $4+\eps$ & \Cref{thm:diameterGenExp_3apx}, \Cref{thm:generalMetricMEB}; general metrics\\ 
\hline
\end{tabular}
\end{center}
\caption{Approximation factors of streaming algorithms 
for diameter and for $1$-center (also called Minimum Enclosing Ball) 
in high-dimensional Euclidean space $\RR^d$. 
We list only algorithms whose space complexity is 
$\poly(d\log (n\Delta))$. % WAS $\poly(d\log n)$. 
Some algorithms work also in general metric spaces.}
\end{table}

\subsection{Related Work and Models}
\label{sec:related}

The expiration model that we propose is not only a natural generalization 
of the sliding-window and other streaming models,
it also pertains to numerous other topics in computer science, as we outline next.
One example where general (i.e., non-consistent) expirations arise is online monitoring settings. 
Another example is computational economics and networking, 
where certificates or contracts are issued with known lengths/expirations,
and must be managed by computer systems. 
Our results are applicable when one is willing to sacrifice accuracy (bounded approximation) to attain dramatic space improvements.

\subparagraph*{Semi-Online Data Structures.}
In the study of dynamic data structures, the semi-online model allows for insertions and deletions where the time of deletion is provided at the time of insertion.  
Here, unlike the sketches we study, no approximation is generally allowed, 
and the focus is instead on reducing the running time of various operations.  
In particular, Dobkin and Suri~\cite{dobkin1991maintenance} 
employed the method of Bentely and Saxe~\cite{bentley1980decomposable} under a data structure that can be constructed in time $P(n)$ and can handle queries in time $Q(n)$, 
in order to handle insertions and deletions in this semi-online model in $\calO(\log n)$ time while increasing the query time to $\calO((P(n)/n + Q(n)) \log n)$.  
Several linear and near-linear space algorithmic improvements followed,
where the focus is primarily on exact methods and improved update times.
Notably, Chan~\cite{chan2003semi} provided improvements for several problems in computational geometry, including discrete $1$-center in $2$ dimensions with slightly sublinear updates.

\subparagraph*{Persistent Stream Queries.}
The database community studied persistent sketches (for streaming), 
where queries may be restricted to subsets of data in certain time windows~\cite{persistent-treaming-2015}.
Within this setting, %Shi~\etal
Shi, Zhao, Peng, Li, and Phillips~\cite{shi2021time}
considered at-the-time persistence (ATTP) and back-in-time persistence (BITP) models,
where the time window of a query must include the first or last time, respectively
(in other words, queries about any prefix or any suffix of the streams).
Notably, the BITP model can be interpreted as a sliding-window query
that specifies the window size $w$ at query time (rather than in advance).
This is closely related to consistent expirations,
especially if all expirations occur after all relevant insertions.
All methods we are aware of for the consistent expiration model should work for this BITP model.  

In this context, Shi \etal~\cite{shi2021time} studied
a variety of problems related to weighted counting and sampling,
where each item $x_i$ is associated with a weight $w_i$ (it could be uniform),
and the desired error bound is an additive $\eps W$
where $W = \sum_i w_i$ is the total weight.
This setting is very useful in standard sketching bounds for frequency estimation, quantiles, approximate range counting, kernel density estimates, and matrix-covariance sketching,
which we study as well.
In particular, they show that a random sample of size $k$
can be maintained in $\calO(k \log n)$ expected space
with $\calO(\log k)$ expected amortized update time.
If the items are selected at random proportionally to their weight
and the weights are in the range $[1,U]$,
then the expected space is $\calO(k\log (nU))$.  
Their methods for BITP use a sampling-in-reverse analysis, which was discovered earlier by Braverman \etal~\cite{braverman2020near} in the context of sliding-window linear-algebra problems, and also used later for sampling Lewis weights in sliding-window streams~\cite{woodruff2025online}. 
Moreover, they show that for mergeable sketches~\cite{agarwal2013mergeable} of size~$s(\eps)$,
a BITP sketch of size $\calO(\eps^{-1} \log n \cdot s(\eps))$ can be maintained.

Our results extend these ideas to the general expiration model, and formalize the proofs in \cref{sec:GE-count} for counting and in \cref{sec:GE-Sampling} for sampling.
Dealing with general expirations requires finding the relevant tools (in prior literature) and additional ideas,
to manage samples that expire in a completely different order than their arrivals. 

% Shi et al.~\cite{shi2021time} studied the generalized sliding-window model under the name back-in-time persistence (BITP).  In a minor distinction, elements are given a time stamp, and the query is made using a 1-sided interval query of the form $[t, \infty)$ from a specific time $t$ to now.  They studied sketching problems where items $a_i$ have a weight $w_i$ (it could be uniform), and the desired error bound is of the form $\eps W$ where $W = \sum_i w_i$ is the sum of weight of all times (so the number of items if it is uniform).  This models standard sketching bounds for frequency estimation, quantiles, approximate range counting, kernel density estimates, and matrix-covariance sketching.  
% Main results include:
% \begin{itemize}
%     \item A BITP sketch for a random sample of size $k$ can be maintained in $\calO(k \log n)$ expected space with $\calO(\log k)$ expected amortized update time.  If the items are selected proportional to their weight and the weights are in range $[1,U]$, the expected space is then $\calO(k(\log n + \log U))$.  
%     \item For a mergable~\cite{agarwal2013mergeable} sketch of size $s(1/\eps)$ a BITP sketch of size $\calO(s(1/\eps) \cdot (1/\eps) \log n)$ can be maintained.  
% \end{itemize}

\subparagraph*{Sliding-Windows and Smooth Histograms.}
The sliding-window model has been studied extensively,
including for frequency and counting problems,  maintaining aggregate statistics, and for geometric and graph problems. %\jonas{The previous sentence appears to be incomplete.}
An extremely popular technique for designing sliding-window streaming algorithms
is the smooth-histogram framework of Braverman and Ostrovsky~\cite{BO07}.
For monotone functions~$f$ that satisfy a certain smoothness property, they show how to convert an algorithm that estimates $f$ in insertion-only streams
into an algorithm that estimates $f$, using slightly more space, in sliding-window streams.  
This framework has been successfully employed for many different problems, 
from counting and frequency problems to graph problems,
but not for geometric problems, which are often not smooth,
e.g., $k$-median and $k$-means clustering \cite{BLLM16}.
Krauthgamer and Reitblat~\cite{KR22} defined a relaxation of this smoothness property, called almost-smoothness,
which is still sufficient to convert algorithms from insertion-only to sliding-window streams,
albeit with a bigger loss in the approximation factor.
It is not difficult to see that the diameter problem, in a general metric space, 
is $2$-almost-smooth (see \cref{lem:AlmostSmooth}),
and since it admits a folklore $2$-approximation in insertion-only streams,
the conversion of~\cite[Theorem 1.7]{KR22} also implies an $(8+\eps)$-approximation in sliding-window streams.
Although immediate, this bound is worse than the known $(3+\eps)$-approximation for diameter \cite{Cohen-AddadSS16}.
%\rnote{To be sure, I wrote down the proof for diameter in \cref{lem:AlmostSmooth}. I am not sure if $k$-center is almost-smooth as well, so I am not discussing it at all.}
Unfortunately, this entire framework seems inapplicable to general expirations.

\subparagraph*{Streaming Algorithms for Diameter and $k$-Center.}
Clustering problems have been studied extensively in the streaming model. 
The metric $k$-center problem was studied in insertion-only streams,
culminating in a $(2+\eps)$-approximation \cite{CharikarCFM04,McCutchenK08,Guha09},
and further extensions to $k$-center with outliers.
In Euclidean space of high dimension $d$,
i.e., when the space bound is restricted to be polynomial in $d$,
better approximation factors are known, particularly for $k=1$ and for diameter
\cite{Zarrabi-ZadehC06, AS15, ChanP14, JKS24, HMV25},
and for small $k$ \cite{KimA15}; 
see \cref{tab:diameter} for the precise constants.
These results were extended to sliding-window streams in \cite{Cohen-AddadSS16, wangLT19} as mentioned above. 
However, they usually do not extend to dynamic streams, which seem harder. 

In low Euclidean dimension, 
i.e., when allowing a space bound that grows exponentially with $d$, % like $(1/\epsilon)^d$
several results achieve a $(1+\epsilon)$-approximation for $k$-center
\cite{AgarwalP02,CeccarelloPP19}.
These results often extend to dynamic and sliding-window streams,
and to handle outliers \cite{dBMZ21,BergBM23}. 
Some of the above references prove near-matching lower bounds,
but usually for algorithms that must store input points \cite{Guha09,AS15,BergBM23},
i.e., these are not bit-complexity bounds for general algorithms.

%-----------------------------------------------------------------------------
\section{Technical Overview}
\label{sec:tech-over}

A key challenge in the expiration model is that the algorithm 
has to track the active items without explicitly storing them all.
In particular, it must be prepared for a scenario where 
no additional items arrive,
and at some future time it will be asked for an estimate.

\subparagraph*{Counting.}
Perhaps the simplest challenge is to just maintain a count of the active items, which we study in \cref{sec:GE-count}.  
In insertion-only streams, it suffices to maintain a single counter. 
In contrast, in the expiration model, 
the algorithm must provide an answer at all possible future times, 
which by a reduction to INDEXING requires $\Omega(n)$ bits of space.

We observe that approximate counting in the expiration model is equivalent to approximating a $1$-dimensional distribution on the expiration times.  
In particular, guaranteeing additive error~$\eps n$ on the count corresponds to $\eps$-error in the Kolmogorov-Smirnov distance. 
If the expirations are consistent, it can be handled 
by simply recording a check point every $\eps n$ insertions. 
The general-expiration case may seem much more complicated, but fortunately, 
it maps directly to the classic problem of quantiles summary in insertion-only streams: each item arrival in the expiration stream corresponds to 
inserting the expiration time of that item into the quantile summary.
Thus, we can employ both quantile sketch upper bounds and their structural lower bounds.  
For additive error $\eps n$ with a constant probability of failure to answer one query, this uses $\Theta(1/\eps)$ space.
For $\eps$-relative error, it uses $\tilde\Theta((1/\eps) \log(\eps n))$ space.  

\subparagraph*{Sampling.}
Maintaining a random sample, which we study in \cref{sec:GE-Sampling}, 
poses a similar challenge in that a sample is in some sense a subset approximation of the count of the items. 
However, the maintained sample might expire, 
and other items must be stored in advance to replace the expired one.
In the general-expiration case, the set of active items 
changes dynamically, and not in a controlled manner as in consistent expirations.

The main insight is that we can imagine running a reservoir sampler in the reverse order of the expirations -- assuming we know these expiration times all in advance.
For maintaining a single sample, we just keep track of when the reservoir gets updated;
our sketch only needs to maintain items that ever get placed in the reservoir.
The size of this set can be analyzed as a coupon collector problem, and is $\calO(\log n)$ in expectation for a single sample.  

Now we must maintain this sample without knowing the expiration times in advance.  
The trick is to assign each item a random value $u_i \in \mathsf{Unif}[0,1]$ and select the smallest active value as our sample.
Because this randomness is assigned once and does not change later,
we only need to maintain the smallest $u_i$ value among the active items.
Thus, an item with a larger $u_i$ value than another item that expires later does not need to be maintained,
and the remaining items are maintained, say in sorted order by expiration. 
The above description produces a single sample,
and we can just run~$k$ independent copies to get a sample of size~$k$. 

We can then adapt this analysis to sampling proportional to weights, as long as the total weight~$W$ of the stream is bounded and each weight is at least $1$.
By changing to an exponential distribution, namely each $u_i \sim \mathsf{Exp}(w_i)$,
we get that the item with minimal value is chosen proportional to its weight.
This follows from the min-stability property of the exponential distribution,
as previously exploited by Cohen~\cite{Cohen97}.  
This has the same complexity as an expiration stream of~$W$ items with uniform weights.

\subparagraph*{Implications of Counting and Sampling.}
Many statistically motivated summaries essentially require
only access to a counter of the data items or to a random sample from it.
Therefore, being able to maintain a counter or a sample over a stream
has numerous applications, 
and indeed fairly direct implications follow for problems involving quantiles, range counting, classification, regression, kernel density, and even matrix sketching. 

\subsection{Diameter}
The next technical challenge is maintaining the diameter of a point set. Unlike the counting and sampling problems, here items differ not only in their expiration time, but also in their geometric information, and the algorithms must preserve the \emph{geometry} of \emph{all} active items at all times.

\subparagraph*{Simple $\calO(1)$-Approximation in General Metrics.}

Let $p$ and $q$ be items of the stream. We say that $p$ is {\em dominated} by $q$ if $S(q)<S(p)$ and $E(p)\leq E(q)$, i.e., for the entire time that~$p$ is active, $q$ is active as well. An item that is not dominated by any other item is called \emph{long}. See~\Cref{fig:defLongDominated} for illustration.
Note that it is easy to decide with $\calO(1)$ words of space whether an arriving item is dominated or long, by keeping track of the latest expiration time among all the items that have arrived so far.

\begin{figure}
    \centering
    \includegraphics[page=3]{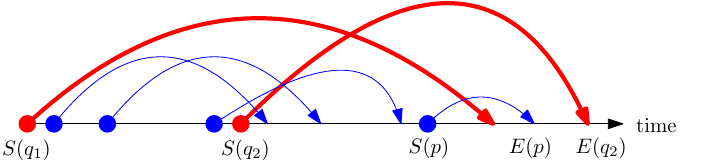}
    \caption{The long items of the stream are $q_1$ and $q_2$. Item $p$ is not long as it is dominated by $q_2$.}
    \label{fig:defLongDominated}
\end{figure}

Consider first the substream consisting of all long items, and notice that its expirations are consistent. It is easy to show that the sliding-window algorithm of~\cite{Cohen-AddadSS16} generalizes to the consistent expiration case, and so we can use it as a subroutine for the long items.
For an item~$q$ that is not long, 
suppose that it is geometrically close to some long item $p$ that dominates it.
Then we can simply ignore $q$ and use $p$ as a proxy for it: geometrically, they are close, and timewise, $p$ is active whenever $q$ is active. 
In the other case, i.e., if $q$ is far from $p$, 
then the pair $p,q$ may be a candidate for the diameter as long as $q$ is active, and thus we need to keep track of such pairs.

More precisely, for a dominated item $q$, 
let us assign one long item that dominates~$q$ as its {\em parent}. 
Again, this is easy to implement, as soon as $q$ arrives, using $\calO(1)$ words of space. 
Now there can be two cases: 
\begin{enumerate}[(a)] \compactify 
    \item Every dominated point is close to its parent (say up to a small constant fraction of the diameter). 
    In this case, the diameter of the long items is a constant-factor approximation to the diameter of all items 
    (by simple triangle-inequality arguments). 
    \item At least one dominated point is far from its parent. 
    To keep track of this scenario, 
    we store an array $A$, where each entry~$A[j]$ maintains a pair of points,
    namely, a dominated point $p$ and its parent $q$, selected as follows:
    from all such pairs where also $d(p,q)\in [(1+\eps)^j, (1+\eps)^{j+1})$,
    select the pair for which the expiration time of $p$ is the latest. 
    The size of this array is only $\calO((1/\eps) \log \Delta)$.
\end{enumerate}
The above reasoning suffices to obtain a constant-factor approximation of the diameter in the general expiration model, 
however it does not match the approximation factor $3+\eps$ that is known for the sliding-window setting.

\subparagraph*{An Improved Algorithm.}
To improve the constant, we replace the black-box subroutine for consistent expirations that is used to handle the long items.
Our redesigned subroutine is conceptually different from the sliding-window algorithm of~\cite{Cohen-AddadSS16} and allows for a tighter integration with the array~$A$, which seems not possible using the algorithm of~\cite{Cohen-AddadSS16}.
If space complexity is not a concern, then long items can be handled by simply having every long point~$p$ track the maximum distance to any point inserted after it and not dominated by it, which we call the {\em radius} of~$p$. 
From this information, we can retrieve the exact diameter of the long items. 
The main insight to achieve the claimed space bound is that if a long point~$p$ is ``sandwiched'' between two points with a similar radius, one inserted before~$p$ and the other after, then we can discard~$p$ while only losing a constant factor in the approximation guarantee. 

In addition, we again store an array~$A$ that captures the distances of dominated items to long items. In contrast to the above, we do not assign a particular parent to each dominated point~$p$.
Rather, when~$p$ arrives, we consider its distance to all stored long items that dominate~$p$. 
A more involved analysis shows that this improves the approximation factor to $3+\eps$, matching the best approximation factor known also for the sliding-window model, while using the same space bound. 

\subparagraph*{Further Euclidean Improvement.}
An additional advantage of our new subroutine for the consistent expiration case is that it admits further improvements in the Euclidean space. 
In general metric space, the algorithm stores a radius~$r$ for every stored long item~$p$ to bound the distance between~$p$ and the set~$S$ of points that are inserted after~$p$ and are not dominated by~$p$. 
In Euclidean space, the algorithm additionally stores for every~$p$ a carefully chosen point~$\widetilde{p}$,
whose distance to points in~$S$ is at most $r$ as well. 
Hence, $S$ lies in the intersection of two balls of radius~$r$, 
instead of only one ball of radius~$r$. 
Since the distance between $p$ and $\widetilde{p}$ is bounded, 
we can further improve the approximation factor for both consistent and general expirations.
In fact, this is strictly better than the best approximation known (and possible) 
for general metrics, even in the sliding-window model. 
See Section \ref{sec:diameter} for details.
\subsection{Approximating $\parK$-Center in Expiration Streams}
We focus on the decision version of $k$-center, where the goal is to decide whether the value of an optimal solution is roughly $\gamma$. We can then instantiate the decision algorithm for $\calO((\log \Delta)/\eps)$ different guesses of $\gamma$ to obtain our approximation.

As with the diameter, one can verify that the sliding-window algorithm of~\cite{Cohen-AddadSS16} in fact works for the consistent expiration model, and thus it provides an $\calO(1)$-approximation of~$k$-center on the set of all long items.
Once again, if a point $q$ that is dominated by $p$ is also geometrically {\em close} to $p$, then there is no need to keep the point~$q$, since any center that is close to~$p$ is also not far away from~$q$.
Thus, our strategy is to divide the stream into~$k$ substreams~$D_1,\dots,D_k$ such that for each substream~$D_i$, the expirations are consistent and therefore a modified instance of the algorithm of~\cite{Cohen-AddadSS16} can be run on it.
The challenge lies in ensuring that~$k$ substreams are sufficient.
To illustrate the argument, we first consider a simplified scenario in which each substream is allowed to use an unlimited amount of space.
Then, we explain how the algorithm can be modified to achieve the desired space bound.

\subparagraph*{A Simplified Scenario.} We define~$D_1$ as the substream containing all long points~$p$, as well as each point~$q$ that is geometrically close to a long point~$p$ in $D_1$ that dominates it. If a point $q$ is not geometrically close to any long point in $D_1$ that dominates it, then $q$ is fed to the second substream~$D_2$. All points passed on to~$D_2$ are then processed in the same manner as for~$D_1$, and the residual points not handled by $D_2$ are again passed to the next layer, etc.
If a point is rejected from the final substream~$D_k$, it is discarded.
If each substream is allowed to use unlimited space, then there is no need to discard any long items and we can show that~$k$ substreams are sufficient.
If a point~$q$ is rejected from all $k$ substreams $D_1,\dots,D_k$, then there is a sequence of points $p_1,\dots,p_k$ such that 
\begin{enumerate}[(i)] \compactify
    \item $p_i$ belongs to the $i$th substream,
    \item these $k$ points (as well as~$q$) are pairwise geometrically far away from each other,
    \item the points are nested timewise, i.e., $p_{i+1}$ is dominated by~$p_i$ and~$q$ is dominated by~$p_k$.
\end{enumerate}
As long as $q$ is not expired, then all points $p_1,\dots,p_k$ are also active, which provides a certificate that any solution for $k$-center must have a large value with respect to $\gamma$. 
Thus, there is no need to keep~$q$, and more generally, to maintain more than~$k$ substreams.

To bound the space within each substream, a naive approach would be to apply the algorithm of~\cite{Cohen-AddadSS16} to each substream independently.
However, because the pruning of long items is not coordinated between the substreams, it becomes impossible to maintain conditions~(ii) and~(iii) simultaneously.
Consider three points~$p_i$, $p_j$, $p_\ell$ with~$i<j<\ell$ such that~$p_i$ is discarded between the insertions of~$p_j$ and~$p_\ell$.
Then the algorithm cannot check whether~$p_\ell$ is geometrically close to~$p_i$, so condition~(ii) may be violated.
If~$p_i$ is replaced with another point~$p'_i$ from~$D_i$, then it cannot be guaranteed that~$p'_i$ dominates~$p_j$, so condition~(iii) may be violated.
Thus, a more coordinated approach is needed.

\subparagraph*{Our Approach.}
Our goal is to maintain a set of points $p_1, \dots, p_k$ that satisfy conditions (i)--(iii), while bounding the space usage of the algorithm within each substream and the total number of substreams.
If we define $p_i$ as the longest-living point in $D_i$, then this set of points satisfies conditions (i) and (iii).
To restore condition~(ii), we discard points from $D_i$ if they are geometrically close to points from substreams $D_j$ with $j < i$ that are inserted later.
This provides the coordination between the different substreams that was lacking in the naive approach, but it creates two new issues.

The first issue arises when our algorithm does not store a point $p$ in substream $D_i$ because it is close to another point $p'_i$ in $D_i$ (in which case we consider $p$ to be \emph{covered} by $p'_i$ in $D_i$).
Later on, we may have to discard $p'_i$ to ensure condition (ii) because it is close to another point $p'_j$ from a substream $D_j$ with $j < i$.
We then consider~$p$ to be covered by~$p'_j$ instead, but this increases the distance between $p$ and the point that covers it by a constant factor. Because $p'_j$ may itself be discarded later on to ensure condition (ii), this effect can cascade up to $k$ times, which increases the approximation guarantee to $\calO(k)$.

The second issue is that the longest-living point from a substream $D_i$ may need to be discarded to restore condition (ii), and this can in turn cause condition (iii) to be violated because the points $p_1,\dots,p_k$ are no longer nested.
We restore the condition by moving the longest-living point among all substreams $D_j$ with~$j > i$ to $D_i$.
Note that this may potentially violate condition (ii) again.
However, we show that a single sweep over all substreams suffices to restore both conditions.
For a more detailed overview of the algorithm, see~\Cref{sec:kcenter}.

%-----------------------------------------------------------------------------
\section{Approximating the Diameter in Expiration Streams}
\label{sec:diameter}

This section presents algorithms that approximate the diameter of a point set in the expiration streaming model.
%Unlike in counting and sampling problems, here items differ not only in their expiration times, but also in their geometric information, and algorithms must track the \emph{geometry} of \emph{all} active items at all times.
We first present an algorithm that matches the approximation factor known for sliding-window streams.
Afterwards, we modify this algorithm for Euclidean space to achieve a better approximation factor.
For simplicity, we assume in this section that the insertion times of the items are distinct, i.e., $S(p)\neq S(p')$ for all $p\neq p'$. If multiple items arrive simultaneously, we fix an arbitrary order to add them to the stream.

\subsection{$(3+\eps)$-Approximation in General Metrics}
\label{sec:diameterGeneralExpiration_new}
Our algorithm maintains a subset~$q_1,\dots,q_k$ of the long items that have been inserted so far. 
Along with every stored item~$q_i$, we store the radius~$r_i$ of the smallest ball centered at~$q_i$ that contains all elements that were inserted after~$q_i$ and are not dominated by $q_i$, i.e., expire later than $q_i$ (see \Cref{fig:diam_Gen_Exp}).
%Inspired by \cite{BO07}, our main insight to achieve the space bound is that if a long item~$q$ is ``sandwiched'' between two items with a similar radius, one inserted before~$q$ and the other after, then we can discard~$q$ while only losing a constant factor in the approximation guarantee. 
%We also adjust the definition of the array~$A$.
%Instead of assigning a particular parent to a dominated item~$q$, we consider the distance to all stored long items that dominate~$q$.
In addition, the algorithm stores an array~$A$.
For each~$j\in[\lceil \log_{1+\eps/3} \Delta\rceil]$, the entry~$A[j]$ stores the maximum expiration time~$E(p)$ among all pairs of points $q$ and $p$ such that (1) $q$ is long and stored, (2) $p$ is dominated by $q$, and (3) $d(q, p)\in [(1+\eps/3)^j,(1+\eps/3)^{j+1})$. %Note that $\Delta_{\max}< (1+\eps/3)^{j+1}\Delta_{\min}$ for $j=\log_{1+\eps/3}\Delta$.
Before the stream starts, we insert a dummy item~$p$ that 
expires immediately after it is inserted. Then, initially we store $q_1=p$, $r_1=0$, $k=1$, and $A[j]=0$ for all~$j$. 

\subparagraph*{Update Procedure.} Each time an item $(p, S, E)$ is inserted, we use \Cref{alg:diameter_general_exp} to update the data structure. If the item is long, we store it.
Line~\ref{li:updateRadius} ensures that for each stored (long) item~$q_i$, the radius~$r_i$ is correct: if~$q_i$ does not dominate the new item~$p$ and the distance from~$q_i$ to~$p$ is larger than~$r_i$, we update the radius to be this distance. Otherwise, if~$q_i$ dominates~$p$, then for the value of~$j$ such that their distance is in $[(1+\eps/3)^j, (1+\eps/3)^{j+1})$, the algorithm updates $A[j]$ to $\max\{A[j],E(p)\}$ (lines \ref{li:updateArray1}--\ref{li:updateArray2}).

To achieve our bound on the number~$k$ of the stored items, we further 
discard all stored items that are inserted between two stored items with a similar radius. 
In particular, if $q_i$ and $q_j$ with $i\leq j$ are two stored points such that $(1+\eps/3)\cdot r_j>r_i$, then we discard all items $q_t$ with $t\in \{i+1, \dots, j-1\}$. Finally, at the end of the update procedure of adding an item, we reindex the stored items starting from $1$, keeping their order.

\subparagraph*{Handling Expirations.} As long as at least two items are stored, the algorithm maintains the invariant that~$E(q_1)\leq T<E(q_2)$ for the current time step~$T$.
In the time step~$T=E(q_2)$, the algorithm discards~$q_1$ from the data structure along with its radius $r_1$.
The remaining items are reindexed starting from~$1$, keeping their order. As all stored items are long, at each time step at most one stored item expires. 

\begin{figure}
    \centering
    \includegraphics[width=\linewidth]{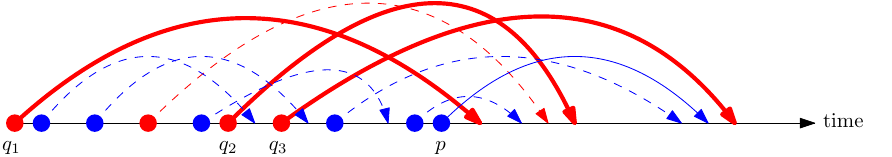}
    \caption{Long items drawn in red. Item $p$ is dominated by $q_3$. At time $S(p)$, when $p$ is inserted, items $q_1, q_2, q_3$ are all stored. 
    The algorithm then updates the radii $r_1$ and $r_2$ corresponding to $q_1$ and $q_2$, respectively, 
    and also the entry in $A$ corresponding to the distance $d(p, q_3)$.}
    \label{fig:diam_Gen_Exp}
\end{figure}

\subparagraph*{Answering Queries.} 
At time~$T$, we report $r=\max\{r_2\}\cup \{(1+\eps/3)^j\mid A[j]> T\}$ as an approximation to the diameter if $q_2$ exists. Otherwise, there are no active points in the stream at time~$T$.
We show that this is a $(3+\eps)$-approximation.
%\rnote{What is the input to Algorithm 1 (also 2)? I suppose it is $(p,S,E)$ but currently these letters are not defined.}

\begin{algorithm}
    \caption{Diameter update procedure for inserted item $(p, S, E)$.}
    \label{alg:diameter_general_exp}
    \If{$E>E(q_{k})$}
    {
        $q_{k+1}\gets p$, $r_{k+1}\gets 0$ \tcp*{$(p,S,E)$ is long}\label{li:onlyLongItemsStored}
    }
    \For{$i\in[k]$}
    {
        \If{$E>E(q_{i})$}
        {
            $r_i\gets \max\{r_i,d(q_i, p)\}$ \label{li:updateRadius}
        }
        \Else{
            let $j$ be such that $d(p, q_i)\in [(1+\eps/3)^j , (1+\eps/3)^{j+1})$\;\label{li:updateArray1}
            $A[j]\gets \max\{A[j], E\}$\; \label{li:updateArray2}
        }
    }
    $i\gets 1$\;
    
    \While{$i\leq k-1$} 
    {
        find the largest $j\geq i$ such that $(1+\eps/3)\cdot r_j>{r_i}$\label{li:diameter_compare_Radii}
        
        discard $q_t$ and $r_t$ for all $t\in\{i+1, \dotso, j-1\}$
        
        $i\gets\max\{j, i+1\}$ \label{li:diameter_end_while}
    }
    $k\gets$ number of remaining elements
    
    reindex the values $q_i$ and $r_i$ (keeping their order) to $q_1, \dotso, q_k$ and  $r_1, \dotso, r_k$
\end{algorithm}

\begin{lemma}\label{lem:diameter-consec-items}
    For every consecutively stored items $q_{i-1},q_i$, 
    and for every item $p$ satisfying $S(q_{i-1})< S(p)< S(q_i)$ and $E(q_{i-1}) < E(p)$, 
    it holds that~$d(q_{i-1},p)\leq (1+\eps/3)\cdot r_i$.
\end{lemma}
\begin{proof}
    Because~$q_{i-1}$ and~$q_i$ are consecutive, any long item that was inserted between them has been discarded.
    If no such long item exists, then the claim follows because there is no item~$p$ with $S(q_{i-1})< S(p)< S(q_i)$ and $E(q_{i-1}) < E(p)$.
    Otherwise, let~$T$ be the time at which the last long item between $q_{i-1}$ and $q_{i}$ is discarded.
    Denote with $r_{i-1}^T$ (resp. $r_i^T$) the radius stored at time~$T$ for the point which is now $q_{i-1}$ (resp. $q_{i}$).
    Then, $r_{i-1}^T<(1+\eps/3)\cdot r_i^T$ holds due to line~\ref{li:diameter_compare_Radii}.
    In particular, for every item $p$ with $S(q_{i-1}) < S(p) < S(q_i)$ and $E(q_{i-1}) < E(p)$, it holds that $d(q_{i-1},p)\leq r_{i-1}^T\leq (1+\eps/3)\cdot r_i^T$. The radius corresponding to point $q_i$ can only increase afterwards, so the claim follows.
\end{proof}

To show that the algorithm achieves $(3+\eps)$-approximation, 
we analyze different cases based on the insertion and 
expiration times of the items that realize the diameter. 

\ThmDiameterGenExp* % restatable theorem 

% \begin{theorem}\label{thm:diameterGenExp_3apx}
%      There exists an algorithm storing $\calO((1/\eps)\log \Delta)$ words 
%      that maintains a $(3+\eps)$-approximation of the diameter in the expiration streaming model in any metric space. 
% \end{theorem}

\begin{proof}
By lines~\ref{li:diameter_compare_Radii}--\ref{li:diameter_end_while} of \Cref{alg:diameter_general_exp}, it holds that ${{(1+\eps/3)}\cdot r_{i+2}\leq {r_i}}$ for each~$i \in [k-2]$.
Therefore, $k\in \calO(\log_{1+\eps/3} \Delta)$. %For each $i\in [k]$, storing the values $q_i$, $r_i$, and $E_i$ uses $\calO(\log n + \log \Delta)$ space. 
The size of the array $A$ is in $\calO(\log_{1+\eps/3} \Delta)$.
Since $\calO(\log_{1+\eps/3} \Delta)=\calO((1/\eps)\log \Delta)$, the space bound follows. 

Let $T$ be the current time. If there does not exist a stored point $q_2$, then all long items are expired at time~$T$. Therefore, all items are also expired and there are no active items at time~$T$. Otherwise, let $j$ be the maximum value such that $A[j]>T$. Then, the query algorithm at time~$T$ returns $r=\max\{r_2, (1+\eps/3)^j\}$. 
By our invariant, $q_2$ has not yet expired at time~$T$. 
As~$r_2$ is the radius corresponding to $q_2$ and $q_2$ is still active, there exists an item $p$ with $S(p)\geq S(q_2)$ and $E(p)> E(q_2)$ such that $d(p, q_2)=r_2$. Similarly, if $j$ exists, there exist two items $p$ and $p'$ with $E(p')\geq E(p)=A[j]>T$ such that $d(p, p')\geq (1+\eps/3)^j$. 
Hence, the diameter at time~$T$ is at least~$r$.
It remains to prove that the diameter is at most $(3+\eps)r$.
We first show that the following holds for every item~$p$ that is active at time~$T$:
\begin{enumerate}[(1)]
    \item If~$S(p) \leq S(q_2)$, then~$d(p,q_1) \leq (1 + \eps/3)r$ by~\Cref{lem:diameter-consec-items}.
    \item If~$S(p) \geq S(q_2)$, then~$d(p,q_2) \leq (1 + \eps/3)r$: If~$E(q_2) < E(p)$, then we have $d(q_2, p)\leq r_2\leq r$, since $r_2$ is the radius of the smallest ball centered at $q_2$ that contains all items that were inserted after $q_2$ and expire after $q_2$.
Otherwise, we have $T<E(p)\leq E(q_2)$ and $d(q_2, p)< (1+\eps/3)^{j+1}\leq (1+\eps/3) r$ by the definition of $A[j]$.
\end{enumerate}
%Claim~(1) follows by~\Cref{lem:diameter-consec-items}.
%To prove~(2), consider the case~$E(q_2) < E(p)$.
Now, let $p$ and $p'$ with $S(p)\leq S(p')$ be the items that realize the diameter at time~$T$. As all stored items are long and $E(q_{1})\leq T<E(p)$, it holds that $S(q_1)<S(p)$. We consider the following three cases to prove that $d(p, p')\leq (3+\eps)r$.
\begin{enumerate}[a)]
    \item If $S(p')<S(q_2)$ then
    $d(p, p')\leq d(p, q_{1})+ d(q_{1}, p')\leq (2+\eps)r$ by~(1). 
    \item If $S(p)< S(q_2) \leq S(p')$ then %it holds that $T<E<E_2$. Therefore, $d(p, p_1)\leq (1+\eps/3)r_2$ and $d(p_1, p_2)\leq (1+\eps/3)r_2$ due to~(2). Further, $d(p_2, p')\leq (1+\eps/3)r$ by the same arguments as in a). Hence, it holds that 
    $d(p, p')\leq d(p, q_{1})+d(q_{1}, q_2)+d(q_2, p')\leq (3+\eps)r$~by~(1),~(2).\label{case:3+eps}
    \item Otherwise $S(p)\geq S(q_2)$. Then $d(p, p')\leq d(p, q_2)+d(q_2, p')\leq (2+\eps) r$ by~(2). \qedhere %, since $r_2$ is the radius of the smallest ball centered at $q_2$ that contains all elements that got inserted after $q_2$ and expire after $q_2$ and by~(3). 
\end{enumerate}
%Therefore, the algorithm provides a $(3+\eps)$-approximation.
\end{proof}

The same algorithm approximates also the minimum enclosing ball.

\begin{restatable}{theorem}{GeneralMetricMEB}\label{thm:generalMetricMEB}
    There is a deterministic expiration-streaming algorithm storing $\calO((1/\eps)\log \Delta)$ words that maintains a $(4+\eps)$-approximation of the minimum enclosing ball in a metric $\calM$. 
\end{restatable}
\begin{proof}
    We use the same algorithm as in \Cref{thm:diameterGenExp_3apx}. Therefore, the space bound follows. At time~$T$, we return the ball centered at $q_2$ with radius $r=(2+\eps)\cdot \max\{r_2\}\cup \{(1+\eps/3)^j\mid A[j]>T\}$. By the proof of \Cref{thm:diameterGenExp_3apx}, there exist two active points that have distance at least $r/(2+\eps)$. Hence, the radius of the minimum enclosing ball of all points is at least~$r/(4+2\eps)$. Further, again by the proof of \Cref{thm:diameterGenExp_3apx}, all active points have distance at most~$2r$ to~$q_2$.
\end{proof}

\subsection{$(1+\sqrt{3}+\eps)$-Approximation in Euclidean Spaces}

In the proof of~\Cref{thm:diameterGenExp_3apx}, the case that prevents a better approximation factor is Case~\ref{case:3+eps}), in which the diameter is realized by two items~$p$ and~$p'$ such that~$p$ is inserted before~$q_2$ and~$p'$ after~$q_2$. Here, our best available bound for the distance between~$q_1$ to~$p'$ is via~$q_2$.
We show that in the Euclidean metric, this bound can be improved from $2+\eps$ to~$\sqrt{3}+\eps$ by storing an additional point $\widetilde{q_2}$ that has a sufficiently large distance to $q_2$.
This yields an approximation factor of~$1+\sqrt{3}+\eps$, which is strictly better than the current best approximation factor, %possible for general metrics
 even in the sliding-window model. 
The modified update procedure is depicted in~\Cref{alg:Euclidean_diameter_general_exp}.
The expiration handling and query answering remain unchanged.

In the beginning, for every newly stored item $q_i$ we define $\widetilde{q_i}=q_i$. Consider the case that long items inserted between $q_{i-1}$ and $q_i$ are discarded. Further, let $p$ be the item such that $q_{i-1}$ and $q_i$ become consecutive stored items at time $S(p)$. In this case, we update $\widetilde{q}_i$ to $p$. In contrast to \Cref{alg:diameter_general_exp}, the value $r_i$ is defined by also considering distances to $\widetilde{q_i}$. We update $r_i$ such that it holds that
%Let $p$ be the item that is currently being inserted. If $j\neq i+1$ in line~\ref{li:Eu_diameter_compare_Radii} of \Cref{alg:Euclidean_diameter_general_exp}, it holds that $r_j=d(q_j, p)$. In this case, we store along with the item with value $q_j$, radius $r_j$ and expiration time~$E(q_j)$, the point $\widetilde{q}_j=p$. 
 %We define~$\widetilde{S}_i = S(\widetilde{q}_i)$. 
%From now on $r_j$ is defined such that
\begin{enumerate}[(1)]
    \item $d(p, q_i)\leq r_i$ for all $p$ with $S(q_i)<S(p)<E(q_i)<E(p)$, 
    \item $d(p, \widetilde{q_i})\leq r_i$ for all $p$ with $S(\widetilde{q}_i)<S(p)<E(q_i)<E(p)$, and
    \item there exists an item $p$ with $S(\widetilde{q}_i)\leq S(p)<E(q_i)<E(p)$ such that $r_i=\max\{d(p, q_i), d(p, \widetilde{q}_i)\}$.
\end{enumerate}
% Further it holds that
% \begin{enumerate}[(4)]
%     \item $d(p, q_i)\leq  (1+\eps/3) d(\widetilde{q}_{i+1}, q_{i+1})$ for all $(p, S, E)$ with $S_{i}\leq S\leq \widetilde{S}_{i+1}$ and $E_i<E$.\label{enu:<S_i+1} \jonas{Turn this into a lemma, analogous to~\Cref{lem:diameter-consec-items}.}
% \end{enumerate}
%, $\widetilde{S}_j=S$ and the radius $\widetilde{r}_j=0$ of the item that gets inserted. 
%Further, let $\tilde{r}_j=r_j$ at this time. 
We store and reindex $\widetilde{q}_i$ and $r_i$ along with $q_i$.

%The algorithm handles deletions as well as queries for the diameter in the same way as in \Cref{sec:diameterGeneralExpiration_new}.
%The query algorithm at time~$T$ returns $r=\max\{r_2\}\cup\{(1+\eps/3)^j\mid A[j]>T\}$ if $r_2$ exists otherwise there are no active points in the stream at time~$T$.

\begin{algorithm}
    \caption{Euclidean diameter update procedure for inserted item $(p, S, E)$.}\label{alg:Euclidean_diameter_general_exp}
    \If{$E>E(q_{k})$}
    {
        $q_{k+1}\gets p$, $r_{k+1}\gets 0$, $\widetilde{q}_{k+1}\gets p$ \tcp*{$(p,S,E)$ is long}\label{li:Eu_onlyLongItemsStored}
    }
    \For{$i\in[k]$}
    {
        \If{$E>E(q_{i})$}
        {
            $r_i\gets \max\{r_i, d(q_i, p), d(\widetilde{q}_i, p)\}$ \label{li:Eu_updateRadius}
        }
        \Else{
            let $j$ be such that $d(p, q_i)\in [(1+\eps/3)^j , (1+\eps/3)^{j+1})$\;\label{li:Eu_updateArray1}
            $A[j]\gets \max\{A[j], E\}$\; \label{li:Eu_updateArray2}
        }
    }
    $i\gets 1$\;
    
    \While{$i\leq k-1$}
    {
        find the largest $j\geq i$ such that $(1+\eps/3)\cdot r_j>{r_i}$\label{li:Eu_diameter_compare_Radii}\label{li:removeLongItems_start}\;
        discard $q_t$ and $r_t$ for all $t\in\{i+1, \dotso, j-1\}$\;
        \lIf{$j>i+1$}{
            $\widetilde{q_j}\gets p$\label{li:Eu_diameter_set_extra}
        }
        
        $i\gets\max\{j, i+1\}$\label{li:Eu_diameter_end_while}
    }
    $k\gets$ number of remaining elements
    
    reindex $q_i$, $\widetilde{q_i}$, and $r_i$ (keeping their order) to $q_1, \dotso, q_k$,  $\widetilde{q_1}, \dotso, \widetilde{q_k}$, and $r_1, \dotso, r_k$
\end{algorithm}
The additional stored points~$\widetilde{q}_i$ give us a stronger version of~\Cref{lem:diameter-consec-items}.%The upper bound on the insertion time of the item $p$ is increased from $S(q_i)$ to $S(\widetilde{q}_i)$.
\begin{restatable}{lemma}{diameterConsecItemsEuclidean}
    \label{lem:diameter-consec-items-euclidean}
    For every consecutively stored items $q_{i-1},q_i$, 
    and for every item $p$ satisfying $S(q_{i-1})< S(p)< S(\widetilde{q}_i)$ and $E(q_{i-1}) < E(p)$, it holds that~$d(q_{i-1},p)\leq (1+\eps/3)\cdot d(\widetilde{q}_i,q_i)$.
\end{restatable}
\begin{proof}
    First, consider the case that~$\widetilde{q}_i=q_i$. Further, let $p$ be the item with smallest insertion time $S(p)>S(q_{i-1})$ such that $E(p)>E(q_{i-1})$. Then, it holds that $p$ is a long item. Since $\widetilde{q}_i=q_i$, it holds that no long items in between $q_{i-1}$ and $q_i$ got inserted. Therefore, there cannot exist an item $p$ with $S(q_{i-1})< S(p)< S(\widetilde{q}_i)$ and $E(q_{i-1}) < E(p)$ and the claim is trivially true.
    Otherwise, let~$T=S(\widetilde{q}_i)$ be the time at which~$\widetilde{q}_i$ is set in line~\ref{li:Eu_diameter_set_extra}.
    Denote with $r_{i-1}^T$ (resp. $r_i^T$) the radius stored at time~$T$ for the point which is now $q_{i-1}$ (resp. $q_{i}$).
    Then, it holds that~$r_i^T=d(\widetilde{q}_i,q_i)$ because otherwise the condition~$(1 + \eps/3)\cdot r_i > r_{i-1}$ in line~\ref{li:Eu_diameter_compare_Radii} would have already been fulfilled in time step~$T-1$.
    For every item $p$ with $S(q_{i-1}) < S(p) < S(\widetilde{q}_i)=T$ and $E(q_{i-1}) < E(p)$, it then holds that $d(q_{i-1},p)\leq r_{i-1}^T\leq (1+\eps/3)\cdot r_i^T=(1+\eps/3)\cdot d(\widetilde{q}_i,q_i)$.
\end{proof}

% With these values in hand, we maintain the following properties. Consider the case that items get discarded in between two stored items $q_{i}$ and $q_{i+1}$. %Define $r=\max\{r_{i+1}, \widetilde{r}_{i+1}\}$. 
% Then by line~\ref{li:Eu_diameter_compare_Radii} of \Cref{alg:Euclidean_diameter_general_exp}, it holds for every item $(p, S, E)$ with $E_i<E$
% that
% \begin{enumerate}[(1)]
%     \item if $S_{i}\leq S\leq \widetilde{S}_{i+1}$, then $d(p, q_i)\leq  (1+\eps/2) d(\widetilde{q}_{i+1}, q_{i+1})$, \label{enu:<S_i+1}
%     \item if $S\geq \widetilde{S}_{i+1}$ and $E>E_{i+1}$, then
%     $d(p, q_{i+1})\leq r_{i+1}$ and $d(p, \widetilde{q}_{i+1})\leq r_{i+1}$, and
%     \item there exists an item $(p, S, E)$ with $S\geq \widetilde{S}_{i+1}$
%     such that $d(p, q_{i+1})=r_{i+1}$ or $d(p, \widetilde{q}_{i+1})=r_{i+1}$.
% \end{enumerate}

In the following, we use~$B(p,r)$ to denote the ball of radius~$r \in \mathbb{R}$ around the point~$p \in \mathbb{R}^d$. We use the following two lemmas.

\begin{restatable}{lemma}{DistInLens}
\label{lem:distInLens}
    Let $a, b, p, q\in \mathbb{R}^d$ be points such that 
    $d(a, b)=r$ and $p, q\in B(a, r)\cap B(b,r)$. %$d(a, p)\leq r$, $d(a, q)\leq r$, $d(b, p)\leq r$, and $d(b,q)\leq r$.
    Then, it holds that $d(p, q)\leq \sqrt{3}r$. 
\end{restatable}
\begin{proof}
    We assume that $a=(0,\dots, 0)$, $b=(r, 0, \dots, 0)$, and $p=(p_1, \dots, p_d)\in B(a, r)\cap B(b,r)$ with $p_1\geq r/2$ without loss of generality. Define $m=(r/2, 0, \dots, 0)$. % and let $p=(p_1, \dots, p_d)$ be in $B(a, r)\cap B(b,r)$.
    %Without loss of generality, we assume that~$p_1\geq r/2$ (otherwise we switch the roles of~$a$ and~$b$).
    Then we obtain
    \[ d(p, m)^2=\left(p_1-\frac{r}{2}\right)^2+\sum_{i=2}^{d}{p_i^2}=\sum_{i=1}^{d}{p_i^2}+\frac{r^2}{4}-p_1r\leq r^2+\frac{r^2}{4}-\frac{r^2}{2}=\frac{3}{4}r^2.\]
    Hence, for all $p, q\in B(a, r)\cap B(b,r)$ it holds that $d(p, q)\leq d(p,m) + d(q,m) \leq \sqrt{3}r$.
\end{proof}

\begin{restatable}{lemma}{Angle}
\label{lem:angle5/6}
    Let $q_2$, $s$, $p'$, and $c$ be points and $\alpha$ the angle at $s$ in the triangle spanned by the points $q_2$, $p'$, and $s$ such that $d(q_2, c)=d(q_2, s)$, $c\in \overline{q_2 p'}$, and $\alpha\geq (5/6) \pi$ (see \Cref{fig:proof_angle}). Then, it holds that $d(s, p')\leq \sqrt{3} d(c, p')$. 
\end{restatable}
\begin{proof} 
    Let $\theta$ be the angle at $q_2$ in the triangle spanned by $q_2$, $p'$, and $s$. It holds that 
    \[\frac{d(q_2, s)}{\sin(\alpha+\theta)}=\frac{d(p', q_2)}{\sin(\alpha)}=\frac{d(s, p')}{\sin(\theta)}.\]
    Further, since $d(s, q_2)=d(q_2, c)$ and $c\in \overline{q_2p'}$ it holds that 
    \[d(q_2, s)+d(c, p')=d(p', q_2)=d(s, p')\cdot \frac{\sin(\alpha)}{\sin(\theta)}.\]
    Hence, 
    \begin{align*}
        \frac{d(c, p')}{d(s, p')}&=\frac{\sin(\alpha)}{\sin(\theta)}- \frac{d(q_2, s)}{d(s, p')}=\frac{\sin(\alpha)}{\sin(\theta)}- \frac{\sin(\alpha+\theta)}{\sin(\theta)}\\
        &= \frac{\sin(\alpha)}{\sin(\theta)}- \frac{\sin(\alpha)\cos(\theta) + \sin(\theta)\cos(\alpha)}{\sin(\theta)}\\
        &= \sin(\alpha) \cdot \frac{1 - \cos(\theta)}{\sin(\theta)} - \cos(\alpha)\\
        &= \sin(\alpha) \cdot \frac{2\sin^2(\theta/2)}{2\sin(\theta/2)\cos(\theta/2)} - \cos(\alpha)\\
        &= \sin(\alpha) \cdot \tan(\theta/2) - \cos(\alpha).
    \end{align*}
    We prove that $f(\alpha,\theta) = \sin(\alpha) \cdot \tan(\theta/2) - \cos(\alpha) \geq \frac{1}{\sqrt{3}}$ for $\alpha\in[(5/6) \pi, \pi)$ and $\theta\in (0, \pi/6]$. % with $\alpha+\theta\leq \pi$.
    Because~$\sin(\alpha)$ and~$\tan(\theta/2)$ are non-negative in the given intervals, it suffices to show that $-\cos(\alpha) \geq \frac{1}{\sqrt{3}}$ for~$\alpha\in[(5/6) \pi, \pi)$.
    The minimum is attained for~$\alpha=(5/6)\pi$ with~$-\cos((5/6)\pi)=\frac{\sqrt{3}}{2} > \frac{1}{\sqrt{3}}$.
\end{proof}

\ThmEuclideanGenExpiration* % restatable theorem 

% \begin{theorem}\label{thm:EuclideanGenExpiration}
%      There exists an algorithm storing $\calO((1/\eps)\log \Delta)$ words that maintains a $(1+\sqrt{3}+\eps)$-approximation of the diameter in the expiration streaming model in the Euclidean space. 
% \end{theorem}

\begin{figure}
    \centering
    \includegraphics[page=3,width=\linewidth]{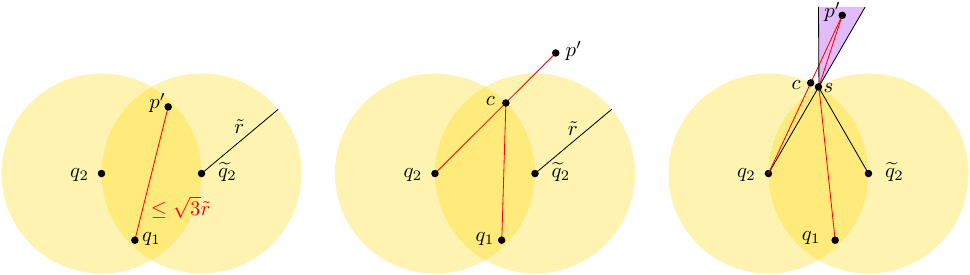}
    \caption{Visualization of Case i)-iii) (from left to right) of the proof of \Cref{thm:EuclideanGenExpiration}.}
    \label{fig:EuclideanGenExp}
\end{figure}
\begin{proof}
    %The space bound follows by \Cref{thm:diameterGenExp_3apx} and the proof of correctness follows similarly to it. 
    Let $r$ be the returned radius at time~$T$, i.e., %$r=\max\{r_2, \widetilde{r}_2\}\cup\{(1+\eps/3)^j\mid A[j]>T\}$ if $\widetilde{r}_2$ exists and otherwise 
    $r=\max\{r_2\}\cup\{(1+\eps/3)^j\mid A[j]>T\}$. 
    By construction of $r_2$ and $A$, it holds that there are two items that are active and have distance at least $r$. Hence, the diameter is at least $r$ at time~$T$. Let $p$ and $p'$ with $S(p)< S(p')$ be the items that realize the diameter at time~$T$. Similarly to the proof of \Cref{thm:diameterGenExp_3apx}, we consider different cases for the values of $S(p)$ and $S(p')$.
    \begin{enumerate}[a)]
    \item If $S(p)\geq S(q_2)$, then $d(p, p')\leq 2 r$, by the same arguments as in the proof of \Cref{thm:diameterGenExp_3apx}.
    \item If $S(p')<S(\widetilde{q}_2)$, then $d(p, p')\leq d(q_1, p)+d(q_1, p')\leq (2+\eps) d(\widetilde{q}_2,q_2) \leq (2+\eps)r$~by~\Cref{lem:diameter-consec-items-euclidean}. %(\ref{enu:<S_i+1})
    \item Otherwise, $S(p)< S(q_2)\leq S(\widetilde{q}_2)< S(p')$. Then, it holds that $d(p, p')\leq d(p, q_1)+d(q_1, p')$.
\end{enumerate}
    %Cases a) and b) are similar to \Cref{thm:diameterGenExp_3apx}. For Case c), we give a better bound.
    %In particular, we show that $d(q_1, p')\leq (\sqrt{3}+\eps)r$ (instead of $\leq (2+\eps)r$).
    For Case c), we show that $d(q_1, p')\leq (\sqrt{3}+\eps)r$.
    Note that since $S(q_1)< S(p)<S(q_2)$ and $E(q_1)<E(p)\leq E(q_2)$, a long item was inserted between $q_1$ and $q_2$ but is not stored at time~$T$. Hence, it holds that $q_2\neq \widetilde{q}_2$.
    Let $\tilde{r}=d(q_2, \widetilde{q}_2)$.
    By~\Cref{lem:diameter-consec-items-euclidean}, it holds that $d(q_1, q_2)\leq (1+\eps/3)\tilde{r}$ and $d(q_1, \widetilde{q}_2)\leq (1+\eps/3)\tilde{r}$ , i.e., $q_1 \in B(q_2, (1+\eps/3)\tilde{r}) \cap B(\widetilde{q}_2, (1+\eps/3)\tilde{r})$.%(\ref{enu:<S_i+1}).
    We assume that $d(\widetilde{q}_2, p')\leq d(q_2, p')$ as the other case follows analogously. 
    To prove $d(q_1, p')\leq (\sqrt{3}+\eps)r$, we distinguish between three cases (see \Cref{fig:EuclideanGenExp}).
    \begin{enumerate}
        \item[i)] If $d(q_2, p')\leq \tilde{r}$, then it follows that~$d(\widetilde{q}_2,p') \leq d(q_2,p') \leq \tilde{r}$.
        Hence, we have~$p' \in B(q_2,\tilde{r}) \cap B(\widetilde{q}_2,\tilde{r})$. By \Cref{lem:distInLens} it holds that $d(q_1, p')\leq (\sqrt{3}+\eps/\sqrt{3})\tilde{r}\leq (\sqrt{3}+\eps/\sqrt{3})r$.
    \end{enumerate}
    Otherwise, define $c$ to be the point on $\overline{q_2 p'}$ with $d(c, q_2)=\tilde{r}$.
    \begin{enumerate}
    \item[ii)] If $d(c, \widetilde{q}_2)\leq \tilde{r}$, then $d(q_1, c)\leq (\sqrt{3}+\eps/\sqrt{3})\tilde{r}$ by \Cref{lem:distInLens}.
    Hence, 
            \[d(q_1, p')\leq d(q_1, c)+d(c, p')= d(q_1, c)+d(q_2, p')-d(q_2, c) \leq (\sqrt{3}+\eps/\sqrt{3})\tilde{r}+ r-\tilde{r}\leq (\sqrt{3}+\eps/\sqrt{3}) r.\]
    \item[iii)] Otherwise, it holds that $d(c, \tilde{q}_2)>\tilde{r}$. Then,
    consider the plane containing $q_2$, $\widetilde{q}_2$ and~$p'$. Let~$s$ be the point in this plane with $d(q_2,s)=d(\widetilde{q}_2,s)=\tilde{r}$ that is closest to~$p'$ (see~\Cref{fig:proof_angle}). %See the right of \Cref{fig:EuclideanVarSlidingWindow} for an example.
    Then, $d(q_1, s)\leq (\sqrt{3}+\eps/\sqrt{3})\tilde{r}$ by \Cref{lem:distInLens}.
    We use \Cref{lem:angle5/6} to obtain a bound on $d(s, p')$.
    Let $\alpha$ be the angle at $s$ in the triangle spanned by the points~$q_2$,~$p'$,~and~$s$.
    Further, let~$L$ denote the perpendicular bisector of~$\overline{q_2\widetilde{q}_2}$ and let~$\beta$ be the smaller angle between the line~$L$ and the line containing $q_2$ and $s$ denoted by~$K$.
    We have~$\beta=\pi/6$ because the triangle spanned by~$q_2$, $\widetilde{q}_2$ and~$s$ is equilateral with angle~$2\beta$.
    Because~$d(c,\widetilde{q}_2)>\tilde{r} =d(c,q_2)$, it follows that~$c$ and~$q_2$ lie on the same side of~$L$, which implies that~$p'$ lies on the same side of the line~$K$ as $c$.
    On the other hand, $p'$ lies on the same side of~$L$ as~$\widetilde{q}_2$ because~$d(\widetilde{q}_2,p') \leq d(q_2,p')$.
    Hence, $p'$ lies in the sector spanned by~$K$ and~$L$ at~$s$ away from $q_2$ (colored in purple in \Cref{fig:proof_angle}).
    Therefore, it follows that~$\alpha \geq \pi - \beta = (5/6)\pi$.
    Hence, by \Cref{lem:angle5/6} we have $d(s, p')\leq \sqrt{3}d(c, p')$ and thus
    \[d(q_1, p')\leq d(q_1, s)+d(s, p')\leq (\sqrt{3}+\eps/\sqrt{3})\tilde{r}+\sqrt{3} d(c, p')\leq (\sqrt{3}+\eps/\sqrt{3}) r. \] 
    \end{enumerate}
    Therefore, in Case c) it holds that $d(p,p')\leq d(p, q_1)+d(q_1, p')\leq (1+\sqrt{3}+\eps) r$.
    %This concludes the proof.
\end{proof}
\begin{figure}
    \centering
    \includegraphics[page=5, height=4.5cm]{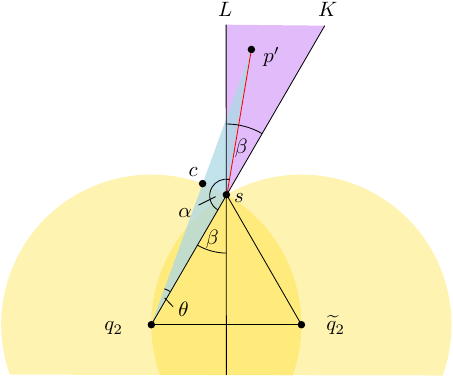}
    \caption{Visualization of Case iii) in~\Cref{thm:EuclideanGenExpiration}. The angle $\beta$ is $\pi/6$ because the triangle spanned by~$q_2$, $\widetilde{q}_2$ and~$s$ is equilateral. Because~$p'$ lies in the sector spanned by~$L$ and~$K$ at~$s$ (shown in purple), it follows that~$\alpha \geq \pi - \beta = (5/6)\pi$.}
    \label{fig:proof_angle}
\end{figure}

Again, the same algorithm can be used to approximate the minimum enclosing ball in Euclidean space $\RR^d$.

\begin{restatable}{theorem}{EuclideanMEB}\label{thm:EucideanMEB}
    There is a deterministic expiration-streaming algorithm storing $\calO((1/\eps)\log \Delta)$ words that maintains $(1+\sqrt{3}+\eps)$-approximation of Euclidean Minimum Enclosing Ball.
\end{restatable}
\begin{proof}
    We use the same algorithm as in \Cref{thm:EuclideanGenExpiration}. Therefore, the space bound follows. Define $\widetilde{r}=(1+\eps/3)d(q_2, \widetilde{q}_2)$ and $r=\max\{r_2\}\cup \{(1+\eps/3)^{j}\mid A[j]>T\}$. Then, it holds that $\widetilde{r}\leq (1+\eps/3)r$, since $d(q_2, \widetilde{q_2})\leq r_2$. %\jonas{by~\Cref{lem:diameter-consec-items-euclidean}?} \lotte{No, simply because $r_2\geq d(q_2, \widetilde{q_2})$}. 
    By construction, there are two active items that have distance at least~$r$. Thus, the radius of the minimum enclosing ball is at least~$r/2$.
    
    If $\widetilde{q}_2=q_2$, then we return the ball $B$ centered at $q_2$ with radius $(1+\eps) r$. Since $\widetilde{q}_2=q_2$, no long item was inserted between $q_1$ and $q_2$. Hence, because $q_1$ is expired, all active items were inserted not later than $q_2$. Therefore, by construction all active items are contained in~$B$ and $B$ is a $(2+\eps)$-approximation of the minimum enclosing ball.

    Otherwise, it holds that $\widetilde{q}_2\neq q_2$ and therefore $d(q_1, q_2)\leq \widetilde{r}$ and $d(q_1, \widetilde{q}_2)\leq \widetilde{r}$ by lines \ref{li:removeLongItems_start}--\ref{li:Eu_diameter_set_extra}.\linebreak %\jonas{by~\Cref{lem:diameter-consec-items-euclidean}?}. \lotte{No, since we have strict inequality in that lemma. It follows directly by the algorithm when items get removed in line 11-14}
    Define $m$ to be the midpoint between $q_2$ and $\widetilde{q}_2$ and let $c$ be the point on the straight-line segment between $q_1$ and $m$ such that $d(c, m)=\widetilde{r}/2$. Then, we return the ball $B$ centered at~$c$ with radius $r_c=((1+\sqrt{3}+\eps)/2)\cdot r$. We show that all active items are contained in $B$, which shows that~$B$ is a $(1+\sqrt{3}+\eps)$-approximation of the minimum enclosing ball.
    By the proof of \Cref{lem:distInLens}, it holds that $d(m, q_1)\leq (\sqrt{3}/2)\cdot \widetilde{r}$ and hence that $d(q_1, c)\leq ((\sqrt{3}-1)/2)\cdot \widetilde{r}$. %\jonas{This is a point on the straight-line segment between $q_1$ and $m$, right?} \lotte{yes for example there. Maybe there are also other points and any is okay for us. So I don't think we have to specify this.}
    Let $p$ be an active item. If $S(p)<S(\widetilde{q}_2)$, then $d(p, q_1)\leq \widetilde{r}$ by \Cref{lem:diameter-consec-items-euclidean}. Hence,
    $d(p, c)\leq d(p, q_1)+d(q_1, c)\leq \widetilde{r}+((\sqrt{3}-1)/2) \widetilde{r}\leq r_c$. Otherwise, $S(p)>S(\widetilde{q}_2)$ and $d(p, \widetilde{q}_2)\leq (1+\eps/3)r$ and $d(p, q_2)\leq (1+\eps/3)r$ by construction. Hence, $d(p, m)\leq (\sqrt{3}/2) (1+\eps/3)\cdot r$ by the proof of \Cref{lem:distInLens}. Therefore, it holds that $d(p, c)\leq d(p, m)+d(m, c)\leq (\sqrt{3}/2) (1+\eps/3)\cdot r+ \widetilde{r}/2\leq r_c$.
\end{proof}
%-----------------------------------------------------------------------------
\section{Approximating $\parK$-Center in Expiration Streams}
\label{sec:kcenter}

The work of~\cite{Cohen-AddadSS16} gives a $(6+\eps)$-approximation for metric $\parK$-center in the sliding-window streaming model storing $\calO((\parK/\eps)\log \Delta)$ words. It is not hard to verify that this algorithm also works under the consistent expiration streaming model with the same space bound and same approximation factor.
%a space bound of $\calO((\parK /\eps)\log \Delta)$ words and same approximation factor. %The additional $\calO(\log n)$ factor is due to having to store the expiration times of the points.
For general expirations, the same algorithm still yields a $(6+\eps)$-approximation. However, the required space can be arbitrarily large.
Below, we give an algorithm that computes a~$(6\parK+2+\eps)$-approximation in the general expiration model storing $\calO((1/\eps)\parK^2 \log \Delta)$ words.
Because the sliding-window algorithm by~\cite{Cohen-AddadSS16} is used as a subroutine, we briefly summarize its main ideas.

\subsection{Preliminaries on the Sliding-Window Algorithm}

% \begin{algorithm}
%     \caption{$\parK$-center algorithm for estimate~$\gamma$ in the consistent expiration model}\label{alg:k-center-consistent}
%     $A,R,O \gets \emptyset$\;
%     \For{item~$p_t$ in the stream $\langle p_1,\dots,p_n \rangle$}{
%         $O \gets \{ o \in O \mid E(o) > t \}$\;
%         \For{$a \in A$ with~$E(a) = t$}{
%             $\DeleteAttractionPoint(a)$\;
%         }
%         $\Insert(p_t)$\;
%     }
%     \myproc{$\DeleteAttractionPoint(a)$}{
%         $A \gets A \setminus \{a\}$\;
%         $R \gets R \setminus \{R(a)\}$\;
%         $O \gets O \cup \{R(a)\}$\;
%     }
%     \myproc{$\Insert(p)$}{
%         $D \gets \{ a \in A \mid d(p,a) \leq 2\gamma \}$\;
%         \For{$a \in D$ with~$E(p) > E(R(a))$}{
%             $R \gets \{p\} \cup R \setminus \{R(a)\}$\;
%             $R(a) \gets p$\;
%         }
%         \If{$D = \emptyset$}{
%             $\AddAttractionPoint(p)$\;
%         }
%     }
%     \myproc{$\AddAttractionPoint(a)$}{
%         $A \gets A \cup \{a\}$\;
%         $R(a) \gets a$\;
%         $R \gets R \cup \{R(a)\}$\;
%         $a_{\text{old}}\gets \arg\min_{a\in A} E(a)$\;
%         \If{ $|A|>\parK+1$}
%         {
%             $\DeleteAttractionPoint(a_{\text{old}})$\;
%         }
%         \If{$|A|>\parK$}
%         {
%             $O\gets \{o\in O\mid E(o) > E(a_\text{old})\}$\;
%         }
%     }
% \end{algorithm}
For a given estimate~$\gamma$ of the solution value, the~$\parK$-center algorithm by~\cite{Cohen-AddadSS16} operates as follows.
%It either returns a solution with value~$\leq 6\gamma$ or a certificate that~$\text{OPT} > \gamma$.
The algorithm maintains a set~$A$ of at most~$\parK+1$ \emph{attraction points}.
Each attraction point~$a$ is associated with a ball of radius~$2\gamma$ centered at~$a$.
The algorithm guarantees that no attraction point lies inside the ball of another attraction point.
The \emph{representative}~$R(a)$ of~$a$ is the longest-living point inside this ball inserted while $a$ is active.
The set of all representatives is denoted by~$R$.
When~$a$ expires, its representative~$R(a)$ is kept in memory until it itself expires.
The representatives of expired attraction points are called \emph{orphans} and are stored in a set~$O$.

%For each time step~$t$, the algorithm first discards all attraction points and orphans that expire at~$t$.
%If an attraction point is discarded, its representative is moved from~$R$ to~$O$.
%Note that it is not necessary to check for expiring representatives: because a representative~$R(a)$ is by definition the longest-living point in the ball around its attraction point~$a$, it cannot expire before becoming an orphan unless~$R(a)=a$.
%After all expired points have been discarded, the new point~$p_t$ is inserted.

When a new point $p_t$ is inserted and it
%If~$p_t$ 
lies in the ball of at least one attraction point, then~$p_t$ is considered \emph{covered} %.
and the representative~$R(a)$ might be updated.
%If it outlives the representative~$R(a)$ of any attraction point~$a$ that covers~$p_t$, then~$R(a)$ is replaced with~$p_t$.
If~$p_t$ is not covered, it is inserted as a new attraction point, with itself as the representative.
If adding~$p_t$ increases the number of stored attraction points above~$\parK+1$, then the algorithm identifies the shortest-living attraction point~$a_\text{old}$ and discards it (and the representative~$R(a_\text{old})$ becomes an orphan).
Finally, if the number of stored attraction points is greater than~$\parK$, then all orphans that do not outlive~$a_\text{old}$ are discarded.

For every estimate~$\gamma \in \{ (1+\eps)^i \mid i \in [\lceil\log_{1+\eps} \Delta\rceil] \}$, an instance of the algorithm is run in parallel.
To answer a query, the algorithm iterates over the estimates in the ascending order.
For each estimate~$\gamma$, the algorithm attempts to construct a solution~$C$ by picking an arbitrary point in~$A \cup R \cup O$ and then greedily adding any point~$p \in A \cup R \cup O$ with~$d(p,C) > 2\gamma$.
If~$|C|>k$, then it is a certificate for~$\text{OPT}>\gamma$. Otherwise, $C$ is returned. It can be shown that $C$ is a~$6\gamma$-coreset for the stream because~$R\cup O$ is a $4\gamma$-coreset.
%If~$|C| \leq k$ upon termination, then~$C$ is returned.
%Otherwise, it is a certificate for~$\text{OPT}>\gamma$ and the algorithm continues with the next estimate.

%We give brief arguments for the approximation factor of~$6+\eps$ and the space bound of~$\mathcal \calO((\parK/\eps)\log \Delta)$ words. (we refer to~\cite{Cohen-AddadSS16} for a full proof).
%If~$|A|>k$, then~$A$ forms a certificate that~$\text{OPT} > \gamma$.
%Otherwise, it can be shown that the greedily constructed set~$C$ is a~$6\gamma$-coreset for the stream because~$R \cup O$ is a~$4\gamma$-coreset.
%Every point~$p$ is covered by an attraction point $a$ when it is inserted.
%Its distance to the representative $R(a)$ is at most $4\gamma$, so the claim holds until $R(a)$ is deleted.
%If $R(a)$ expires, then $p$ has expired as well because~$R(a)$ is the longest-living point inside the ball around~$a$.
%If~$R(a)$ is deleted because $|A|>k$ holds, then it will continue to hold at least until $E(R(a))>E(p)$.\jonas{Maybe we don't need to sketch the proof at this point. We do need to mention that~$R \cup O$ is a~$4\gamma$-coreset.}

The main insight regarding the space bound is that the orphans that are still in memory belong to a subset of the~$\parK+1$ most recently discarded attraction points.
Any orphan~$o$ that belongs to an even older attraction point must have been inserted before the most recently discard attraction point~$a$, because the algorithm stores at most~$\parK+1$ attraction points at any time.
Due to the consistent expiration property of the stream, $o$ expires before~$a$, so~$o$ is discarded at the latest when~$a$ is discarded.
In the general expiration model, consistent expirations are not guaranteed.
This is the main reason why the size of~$O$ cannot be bounded, because then the orphans belonging to arbitrarily old attraction points may still be in memory.

\subsection{Overview of Our Algorithm}
We present an algorithm that, given a parameter $\maxK\in \mathbb{N}$, solves $\parK$-center under the (general) expiration streaming model for every $\parK\leq \maxK$.
As with~\cite{Cohen-AddadSS16}, our algorithm solves the decision version of the problem for a given estimate~$\gamma$, and this decision algorithm is then run in parallel for different estimates.
We observe that in the algorithm by~\cite{Cohen-AddadSS16}, the bound on the number of stored orphans does not require that the entire stream has consistent expirations, only that the attraction points are long.
This ensures that the attraction points always outlive the orphans of already discarded attraction points.
In the general expiration model, we cannot guarantee that all attraction points are long.
Therefore, we split the stream into~$\maxK$ substreams.
In each substream~$i$, we run a modified version of the algorithm by~\cite{Cohen-AddadSS16}, storing the sets~$A_i$, $R_i$ and~$O_i$ of attraction points, representatives and orphans.
When a point is inserted, it is placed into the first substream in which it is long or covered by the $2\gamma$-ball of an existing attraction point.
This ensures that within each substream, all attraction points are long and the number of orphans is bounded.

The main challenge lies in ensuring that~$\parK$ substreams are sufficient to approximate~$\parK$-center.
If there is a point~$p$ that does not fit into any of the first~$\parK$ substreams, we maintain a certificate that~$\text{OPT} > \gamma$ until~$p$ expires.
This certificate consists of the longest-living point from each substream, plus~$p$ itself.
Hence, we need to ensure that these points are pairwise more than~$2\gamma$ apart.
This is done by discarding points in different substreams if they are too close to each other: a point~$p$ from substream~$i$ is now also considered covered (and is therefore discarded) if it expires not later than a stored point~$p'$ from a lower substream~$j<i$ with~$d(p,p')\leq 2\gamma$.
To ensure that it is sufficient to consider lower substreams, we maintain the invariant that the substreams are ordered in descending order of their longest-living stored point.
If the longest-living point in substream~$i$ becomes covered, the invariant is violated.
It is restored by moving the longest-lived point among all substreams~$j>i$ down to substream~$i$ (which may lead to cascading effects).

A side effect of the expanded covering rule is that it causes the approximation factor to be dependent on~$\parK$.
When a representative or orphan~$r$ is discarded from substream~$i$ because it is covered by a point~$p$ from a lower substream, then the~$4\gamma$-coreset for substream~$i$ may be destroyed.
If we replace~$r$ with~$p$ in the coreset, its approximation guarantee increases to~$6\gamma$.
Moreover, the effect may be cascading because~$r$ may later be discarded when it is covered by another point that is inserted into an even lower substream.

\subsection{Detailed Algorithm Description}\label{sec:algorithm}
In this section, we give a detailed description of the $\parK$-center algorithm under the expiration model.
In addition to the sets~$A_i$, $R_i$ and~$O_i$, each substream~$i$ also maintains two values~$e_i$ and~$t_i$.
The value~$e_i$ is the expiration time of the longest-living point stored in substream~$i$, which is used to maintain the ordering of the substream.
The value~$t_i$ is the earliest time such that every item~$p$ that is still active was originally inserted into a substream~$j \leq i$ (note that it may have since been moved to a lower substream).
In other words, until time~$t_i$ there is at least one active item that was not originally inserted into any of the substreams~$1$ to~$i$, and~$t_i$ is the last expiration time among all such items.
We show that as a consequence of this definition, until time~$t_i$ every solution to $i$-center has radius greater than~$\gamma$.

Finally, the algorithm maintains a time~$\Abig$, which indicates that until time~$\Abig$, there is a set of at least~$\maxK+1$ active points that were at some time all included in the set~$A_i$ for some substream~$i$.
Note that if~$|A_i|>\maxK$ holds, then~$A_i$ is a certificate that all~$\parK$-center solutions have radius greater than~$\gamma$ until the first point from the set expires.
However, points from~$A_i$ may be discarded before they expire if they are covered by a newly inserted point from a lower substream, which destroys the certificate.
Hence, we store~$\Abig$ to indicate that the certificate still exists, even if some the points are no longer stored.

\subparagraph*{Invariants.} The algorithm maintains the following invariants at all times.
\begin{enumerate}[(1)]
    \item For all $i\in [\maxK]$, it holds that~$e_i=0$ if~$A_i \cup R_i \cup O_i = \emptyset$, and~$e_i = \max\{E(p) \mid p \in A_i \cup R_i \cup O_i \}$ otherwise.\label{in:Ei_longest_time}
    \item It holds that $e_1\geq e_2\geq \cdots \geq e_{\maxK}$. \label{in:orderEi}
    \item For all $a\in \bigcup_{i\in [\maxK]}A_i$, it holds that $d(a, R(a))\leq 2\gamma$. \label{in:distAR(a)}
    \item For all distinct $a, a'\in \bigcup_{i\in [\maxK]}A_i$, it holds that $d(a, a')>2\gamma$. \label{in:distAA}
    \item For all distinct $q_i\in R_i\cup O_i$ and $q_j\in R_j\cup O_j$ with $i\leq j$ and $E(q_i)\geq E(q_j)$, it holds that $d(q_i, q_j)>2\gamma$. \label{in:distLongRO}
\end{enumerate}

\subparagraph*{Initialization.}
The data structures are initialized as follows. 
\begin{itemize}
    \item $A_i, R_i, O_i\gets \emptyset$ for all $i\in [\maxK]$,
    \item $e_i, t_i\gets 0$ for all $i\in [\maxK]$, and
    \item $\Abig \gets 0$.
\end{itemize}

\begin{algorithm}
    \caption{$k$-center update procedure for inserted item~$p$.}\label{alg:updatek-center}
    %\myproc{$\Insert(p)$}{
        \For{$i = 1,\dots,K$}
        {
            \If{$\exists q\in R_i\cup O_i$ with $E(q)\geq E(p)$ and $d(p, q)\leq 2\gamma$\label{li:distToRepr}}
            {
                 \Return \tcp*{$p$ is added to substream $i$} 
            }
            \If{$\exists a\in A_i$ with $d(p,a)\leq 2\gamma$\label{li:distToAttr}}
            {
                \If{$E(p)>E(R(a))$}
                {  
                    $R_i \gets \{ p \}\cup R_i \setminus \{ R(a) \} $ \;
                    $R(a)\gets p$\label{li:updateRepr}\;
                    $\DiscardCoveredPoints(p, i)$\;
                }
                \Return \tcp*{$p$ is added to substream $i$}
            }
            \If{$E(p) > e_i$\label{li:E>ei}}
            {
                $\AddAttractionPoint(p, i)$\;\label{li:addToA}
                \Return \tcp*{$p$ is added to substream $i$}
            }
            $t_i\gets \max\{t_i,E\}$ \tcp*{$p$ is not added to substream $i$}\label{li:setT}
        }
        \tcp{reorder substreams}
        \For{$i=1,\dots,K$\label{li:beginEnsureOrder}}
        {
            \If{$\bigcup_{j=i}^{\maxK} R_j\cup O_j= \emptyset$}{
                $e_i\gets 0$ \tcp*{substream $i$ is empty}
                \textbf{continue}
            }
            $e_i\gets \max\{E(q)\mid q\in \bigcup_{j=i}^{\maxK} R_j\cup O_j\}$\label{li:orderEi}\;
            Let $j\geq i$ be minimal such that $q\in R_j\cup O_j$ with $E(q)=e_i$\label{li:minj}\;
            Let $q\in R_j\cup O_j$ such that $E(q)=e_i$\label{li:longestItem}\;
            \If{$i\neq j$}{
                \If{$q\in R_j$}{
                    Let $a$ be such that $R(a)=q$\;
                    $R_j\gets R_j\setminus q$, $A_j\gets A_j\setminus a$ \label{li:deleteFromRA}\tcp*{discard $a$ and $q$ from substream $j$}
                    $\AddAttractionPoint(q, i)$\label{li:AddToRA}%\tcp*{add $a$ and $q$ to substream $i$}
                }
                \Else{
                    $O_j\gets O_j\setminus q$\label{li:deleteFromO}\;
                    $\AddAttractionPoint(q, i)$\label{li:AddOtoA}\;
                }
            }
        }
    %}
\end{algorithm}

\begin{algorithm}
    \caption{Subroutines of the $k$-center update procedure for inserted item~$p$.}
    %Deleting points covered by a point~$p$ in substream~$i$ and Adding an attraction point~$a$ to substream~$i$}
    \myproc{$\DiscardCoveredPoints(p,i)$}{
        \For{$j=i+1, \dots, \maxK$}
        {
            \tcp{$F$ is moved to substream $i$}
            $F\gets \{a\in A_j\mid E(R(a))<E(p) \text{ and } (d(a, p)\leq 2\gamma \text{ or } d(R(a), p)\leq 2\gamma) \}$\; 
            $A_j\gets A_j\setminus F$\;
            $R_j\gets R_j\setminus \{R(a)\mid a\in F\}$\;
            \BlankLine
            \tcp{$G$ is moved to substream $i$}
            $G\gets \{o\in O_j\mid E(o)<E \text{ and } d(o,p)\leq 2\gamma\}$\;
            $O_j\gets O_j\setminus G$\;
        }
    }
    \myproc{$\AddAttractionPoint(a, i)$}{
            $A_i\gets A_i\cup \{a\}$\;
            $R(a)\gets a$\;
            $R_i\gets R_i\cup \{R(a)\}$\;
            $\DiscardCoveredPoints(a, i)$\;
            $e\gets \min_{a\in A_i} E(a)$\;
            Let $a_{\text{old}}\in A_i$ be such that $E(a)=e$ \label{li:oldA}\;
            \If{ $|A_i|>\maxK+1$\label{li:sizeofA}}
            {
                %$\DeleteAttractionPoint(a_\text{old})$\label{li:delALeadToOrphans}\;
                $A_i \gets A_i \setminus \{a_\text{old}\}$ \label{li:delALeadToOrphans}\;
                $R_i \gets R_i \setminus \{R(a_\text{old})\}$\;
                $O_i \gets O_i \cup \{R(a_\text{old})\}$\;
            }
            \If{$|A_i|>\maxK$}
            {
                $O_i\gets \{o\in O_i\mid E(o) \geq E(a_\text{old})\}$\label{li:boundSizeO}\;
                $\Abig\gets \max\{\Abig, E(a_\text{old})\}$\label{li:Abig}
            }
        }
\end{algorithm}

\subparagraph*{Item Insertion.}
When a point~$p$ is added to the stream, \Cref{alg:updatek-center} is called.
We iterate over the different substreams with index $i=1$ to $K$ until we find a substream to which $p$ can be added.
If a point~$q=R(a)$ in $R_i\cup O_i$ exists that expires later than~$p$ or at the same time as~$p$ and has distance at most~$2\gamma$ to~$p$, then $p$ is added to substream with index $i$ and $p$ is covered by~$q$. 
Note that as $p$ expires not later than $q$, 
we do not have to update~$R(a)$ and we do not store~$p$.
%the point~$q$ is still the longest surviving item of this group and stays the representative point of this group. 
%As we only store the attraction point and the representative point of a group, $p$ is not stored.
Further, if an attraction point $a\in A_i$ exists with $d(a, p)\leq 2\gamma$, we add $p$ to substream $i$ and~$p$ gets covered by~$a$ and check whether we have to update the representative point~$R(a)$. To maintain Invariant~\ref{in:distLongRO} when the representative point~$R(a)$ is updated, we run $\DiscardCoveredPoints$.

If~$p$ lives longer than all points currently in substream~$i$, then~$p$ would be a long item in substream~$i$ with the definition of \Cref{sec:diameter}.
In the case that $p$ has distance greater than $2\gamma$ to all points in $A_i$, it becomes an attraction point in~$A_i$ and we set $R(p)=p$.
%The group corresponding to $p$ consists at this time only of $p$. Hence, $p$ is the longest surviving point of this group and $R(p)$ is set to $p$.
Then, $\AddAttractionPoint$ is called to maintain Invariant~\ref{in:distAA} and to bound the storage size.
If~$p$ was not added to substream~$i$ at the end of iteration~$i$, we update the earliest time~$t_{i}$ at which all active points are originally inserted into substreams with index at most~$i$ to $\max\{t_{i}, E\}$.

% In the beginning, we compute an index $\ell$ such that the item $(p, S, E)$ can only be added to a substream with index at most $\ell$. 
% %which is the maximum index of the substream item $(p, S, E)$ can be added to. 
% The highest index of a substream, to which $(p, S, E)$ is added, is the lowest index of the substream in which $(p, S, E)$ would be a long item with the definition of \Cref{sec:diameterGeneralExpiration_new}. If no such index exists, the parameter $\ell$ is set to $\maxK+1$, which means that $(p, S, E)$ cannot become an attraction point in any substream. 

% Then, we iterate over the different substreams with index $i=1$ to $\ell$. For each substream, we check whether $(p, S, E)$ can be added to this substream until it is added.  If $(p, S, E)$ is not added after iterating over all substreams with index from $1$ to $\ell$, we add $(p, S, E)$ to substream $\ell$ and it becomes an attraction point in $A_\ell$.

The rest of \Cref{alg:updatek-center} (lines~\ref{li:beginEnsureOrder}--\ref{li:AddOtoA}) ensures that Invariant~\ref{in:orderEi} remains true. So, the longest surviving point in substream $i$ lives at least as long as the longest surviving point in substream $i+1$ for all $i\in [\maxK-1]$, i.e., $e_1\geq e_2\geq \cdots\geq e_\maxK$. We iterate over all substreams $i=1$ to $\maxK$ and find in iteration $i$ the longest surviving point~$q$ in the substreams with index $j=i$ to $\maxK$. If $q$ is not already contained in substream $i$, we remove it %s corresponding group 
from its substream and add it to substream $i$ as a new attraction point. In this case, its corresponding attraction point is discarded if it was still stored.
%the new attraction point of this group is set to $q$.

\subparagraph*{Item Expiration.}
If a point in $O_i$ for any $i\in [\maxK]$ expires, it is removed from~$O_i$. If a point $a\in A_i$ for any $i\in [\maxK]$ expires, it is removed from $A_i$ and $R(a)$ is moved from $R_i$ to $O_i$. Note that an item $R(a)$ is contained in $R_i$ only if $a\in A_i$. As $E(R(a))\geq E(a)$, an item $R(a)$ in $R_i$ can only expire if $E(a)=E(R(a))$. Then, $a$ is removed from $A_i$, $R(a)$ is added to $O_i$, and as $R(a)\in O_i$ expires it is removed from $O_i$.

\subparagraph*{Query Algorithm.}
As in~\cite{Cohen-AddadSS16}, 
%with the algorithm for the consistent expiration model,
we run the algorithm in parallel for every estimate~$\gamma\in \{(1+\eps)^i \mid i \in [\lceil\log_{1+\eps}\Delta\rceil]\}$.
A query for~$\parK \leq \maxK$ and time~$T$ is answered by iterating through the estimates in ascending order.
If~$\max\{t_{\parK}, \Abig\} \geq T$ for the current estimate~$\gamma$, then it is skipped.
Otherwise,  we first choose an arbitrary point $q\in \bigcup_{i\in[\parK]} A_i\cup R_i\cup O_i$ and add it to the center set $C$. Then, we greedily add any point $q\in \bigcup_{i\in[\parK]} A_i\cup R_i\cup O_i$ to $C$ with distance greater than $2\gamma$ to all points in $C$. If $|C|>\parK$ at termination, we have a certificate that no solution with radius $\gamma$ can exist. For the smallest value of $\gamma$ with $|C|\leq \parK$, we return $C$ and show that this gives a $(6\parK+2+\eps)$-approximation to metric $\parK$-center.

\subparagraph*{Notation.}
We denote with $A_i^T$, $R_i^T$, and $O_i^T$ the sets $A_i$, $R_i$, and $O_i$, with $R^T(a)$ the point $R(a)$, and with~$e_i^T$ and $t_i^T$ the values $e_i$ and $t_i$ at time~$T$. If an attraction point $a$ is not stored anymore at time~$T$, we denote with $R^T(a)$ the representative of $a$ at the time $a$ is removed from $\bigcup_{i\in [\maxK]}A_i$. In addition, we define $R^{\infty}(a)$ to be the last representative point of~$a$. 

\subsection{Storage Size}
In this section, we prove a space bound for the stored sets $A_i$, $R_i$, and $O_i$.

\begin{observation}\label{obs:long}
    At the time when $a$ is added to $A_i$, it holds that $e_i=E(a)$.
\end{observation}
\begin{proof}
    If $a$ is added to $A_i$ in line~\ref{li:AddToRA} or~\ref{li:AddOtoA} of \Cref{alg:updatek-center}, the claim holds by line~\ref{li:longestItem}.
    Otherwise, $a$ is added to~$A_i$ in line~\ref{li:addToA}. Let this time be $T$. Then, $e_i^{T-1}<E(a)$. $\AddAttractionPoint$ does not change the points of the substreams $j$ for $j\in [i-1]$ and it holds that $a\in A_i^T\cup R_i^T\cup O_i^T\subset A_i^{T-1}\cup R_i^{T-1}\cup O_i^{T-1}\cup \{a\}$. Hence, it follows that $e_j^T=e_j^{T-1}$ for all $j\in [i-1]$ and $e_i=E(a)$.
\end{proof}
    
\begin{lemma}\label{lem:spaceBoundARO}
    The size of each set $A_i$, $R_i$, and $O_i$ is at most $\maxK+1$ for any $i\in [\maxK]$.  
\end{lemma}
\begin{proof}
    Let $i\in [\maxK]$. Points are only added to~$A_i$ in $\AddAttractionPoint$. Hence, the size of $A_i$ is at most $\maxK+1$ for all $i\in[\maxK]$. Further, the set $R_i$ contains exactly one point for each point in $A_i$.
    
    To show that $O_i$ has size at most $\maxK+1$, we use similar arguments as in the proof of Lemma~7 in~\cite{Cohen-AddadSS16}. Let $T$ be the current time and let $\widehat{a}_w, \widehat{a}_{w-1}, \dots, \widehat{a}_1$ be the points that were removed from~$A_i$ in line~\ref{li:delALeadToOrphans} of~$\AddAttractionPoint$ or because they expired, sorted in ascending order of their expiration times. %At the time an item $a$ gets added to $A_i$, it holds that $e_i=E(a)$. Hence, it holds that $E(\widehat{a}_w)\leq E(\widehat{a}_{w-1})\leq \cdots\leq  E(\widehat{a}_1)$. 
    Note that due to the choice of~$a_\text{old}$ in line~\ref{li:oldA} in $\AddAttractionPoint$, this is also in ascending order of their expiration times.
    Furthermore, due to~\Cref{obs:long}, this is the same order in which they were inserted into~$A_i$. 
    Then, the points that have been in $O_i$ until time~$T$ are $R^\infty(\widehat{a}_w), R^\infty(\widehat{a}_{w-1}), \dots, R^\infty(\widehat{a}_1)$.
    We show that all points~$R^\infty(\widehat{a}_{\maxK+j})$ with~$j\in\{2, 3\dots, w\}$ have been removed from~$O_i$ by time~$T$, which implies~$|O^T_i|\leq \maxK+1$.
    Since line~\ref{li:sizeofA} of~$\AddAttractionPoint$ enforces~$|A_i|\leq \maxK+1$, the point~$\widehat{a}_{\maxK+j}$ must have already been removed from~$A_i$ by the time~$\widehat{a}_1$ was inserted.
    Hence, $R^\infty(\widehat{a}_{\maxK+j})$ was added to $R_i$ before $\widehat{a}_1$ was added to $A_i$. Therefore, it holds that $E(R^\infty(\widehat{a}_{\maxK+j}))<E(\widehat{a}_1)$ by~\Cref{obs:long}. 
    If~$\widehat{a}_1$ was removed from $A_i$ because it expired, then $R^\infty(\widehat{a}_{\maxK+j})$ has also expired at time~$T$ and hence is not contained in $O^T_i$. Otherwise, it was discarded in line~\ref{li:delALeadToOrphans} of $\AddAttractionPoint$ at some time~$t \leq T$.
    Then, line~\ref{li:boundSizeO} of $\AddAttractionPoint$ was called with $a_{\text{old}}=\widehat{a}_1$.
    Hence, $R^\infty(\widehat{a}_{\maxK+j})$ was removed from $O_i$ at~$t$ at the latest.
\end{proof}

\subsection{Correctness}
To show correctness, we begin with proving that Invariants~(1)-(5) hold.

The next observation shows that items are moved only from substreams with higher index to substreams with lower index. This is a crucial property to get a bound on the approximation factor.
%Using this lemma, we can later bound the approximation factor. 
\begin{observation}\label{obs:MoveOnlyLowerStream}
    Let $q\in A_i^t\cup R_i^t\cup O_i^t$ for some time~$t$.
    For all~$T>t$ and~$i < j \leq K$, it holds that~$q \notin A_j^T\cup R_j^T\cup O_j^T$.
\end{observation}
\begin{proof}
    Let $q\in A_i^t\cup R_i^t\cup O_i^t$ be such that $q\notin A_i^{t+1}\cup R_i^{t+1}\cup O_i^{t+1}$ and $q\in \bigcup_{\ell\in [\maxK]} A_\ell^{t+1}\cup R_\ell^{t+1}\cup O_\ell^{t+1}$. Then, $q$ is removed from $A_i\cup R_i\cup O_i$ in line~\ref{li:deleteFromRA} or \ref{li:deleteFromO} of \Cref{alg:updatek-center}. By lines~\ref{li:minj}, \ref{li:AddToRA}, and \ref{li:AddOtoA} of \Cref{alg:updatek-center}, there exist an index $j\leq i$ such that $q$ is added to $A_j\cup R_j\cup O_j$. Hence, the claim follows by induction. 
\end{proof}

\begin{observation}\label{obs:Invariants123}
    Invariants (1)-(3) hold during the algorithm.
\end{observation}
\begin{proof}
    Invariant~\ref{in:Ei_longest_time} and Invariant~\ref{in:orderEi} follow by \Cref{obs:MoveOnlyLowerStream} and lines~\ref{li:orderEi}, \ref{li:AddToRA}, and \ref{li:AddOtoA} of \Cref{alg:updatek-center}.
    To prove Invariant~\ref{in:distAR(a)} is to show that $d(a, R(a))\leq 2\gamma$ for all $a\in \bigcup_{i\in [\maxK]}A_i$. This holds since $R(a)$ is set in line~\ref{li:updateRepr} of \Cref{alg:updatek-center} or in $\AddAttractionPoint$.
\end{proof}

\begin{observation}\label{obs:distROtoA}
    Let $q\in R_i\cup O_i$. Then, for all $a\in \bigcup_{j\in [i-1]}A_j$, it holds that $d(q, a)>2\gamma$.
\end{observation}
\begin{proof}
    When~$q$ is first inserted into the stream, line~\ref{li:distToAttr} ensures that the claim holds. 
    When a new attraction point~$a$ is inserted at time~$t$ into~$A_j$ with~$j < i$, $\DiscardCoveredPoints$ is called.
    We have~$E(a) > e^t_j \geq e^t_i \geq E(q)$ by Invariant (1), so~$q$ is discarded if~$d(q,a)>2\gamma$ holds. 
    If~$q$ is moved to substream~$i$ from another substream~$\ell > i$, the claim already holds for all~$a \in \bigcup_{j \in [\ell-1]}A_j$ by induction.
\end{proof}

\begin{lemma}\label{lem:Invariants45}
    Invariants (4) and (5) hold during the algorithm. 
\end{lemma}
\begin{proof}
    We begin with proving that Invariant~\ref{in:distAA} holds, i.e., $d(a, a')>2\gamma$ for all distinct $a, a'\in \bigcup_{i\in [\maxK]}A_i$. 
    Let $a$ be an item that is added to $\bigcup_{i\in [\maxK]} A_i$ at time $t$. If it is added to $A_\ell$ in line~\ref{li:addToA} of \Cref{alg:updatek-center}, then by line~\ref{li:distToAttr} of \Cref{alg:updatek-center}, for all $a'\in \bigcup_{i\in [\ell]} A_i$, it holds that $d(a', a)>2\gamma$. 
    Otherwise, $a$ is added to $A_\ell$ in line~\ref{li:AddToRA} or \ref{li:AddOtoA} of \Cref{alg:updatek-center}. Hence, there exists an index $i>\ell$ such that $a\in R_i^{t-1} \cup O_i^{t-1}$. By \Cref{obs:distROtoA}, it holds that for all $a'\in \bigcup_{j\in [i-1]}A_j$ that $d(a, a')>2\gamma$.
    Further, in both cases $\DiscardCoveredPoints(a, \ell)$ is called when $a$ is inserted, which ensures that $d(a', a)>2\gamma$ holds for all $a'\in \bigcup_{i=\ell+1}^{\maxK}A_i$. 

    It remains to prove Invariant~\ref{in:distLongRO}, i.e., for all distinct $q_i\in R_i\cup O_i$ and $q_j\in R_j\cup O_j$ with $i\leq j$ and $E(q_i)\geq E(q_j)$, it holds that $d(q_i, q_j)>2\gamma$.
    Consider the time at which~$q_i$ is inserted into $R_i \cup O_i$. We show that $d(q_i,q_j) > 2\gamma$ holds for all~$q_j \in R_j \cup O_j$ with
    \begin{enumerate}[(i)]
        \item $j \leq i$ and $E(q_j) \geq E(q_i)$, or
        \item $j \geq i$ and $E(q_j) \leq E(q_i)$.
    \end{enumerate}
    When~$q_i$ is first inserted into the stream, then line~\ref{li:distToRepr} of \Cref{alg:updatek-center} ensures (i).
    If $q_i$ is moved from~$R_\ell \cup O_\ell$ for some~$\ell > i$ in line~\ref{li:AddToRA} or~\ref{li:AddOtoA}, then~(i) holds by induction.
    In both cases, $\DiscardCoveredPoints$ is called, which ensures~(ii).
\end{proof}

Next we show two properties that ensure that the solution to $\parK$-center has radius greater~$\gamma$.

\begin{lemma}\label{lem:sizeACertificate}
    Let $T$ be the current time. 
    If $|\bigcup_{i\in [\maxK]} A_i|>\parK$ or $\Abig>T$, then the solution for the metric $\parK$-center has radius greater than~$\gamma$.
\end{lemma}
\begin{proof}
    By Invariant~\ref{in:distAA} and \Cref{lem:Invariants45}, all distinct points in $\bigcup_{i\in [\maxK]} A_i$ have pairwise distance greater than~$2\gamma$. Therefore, no two points in $\bigcup_{i\in [\maxK]} A_i$ can be in the same cluster of a solution with radius $\gamma$. Hence, if $|\bigcup_{i\in [\maxK]} A_i|>\parK$, there cannot exist a solution to $\parK$-center with radius at most~$\gamma$.
    If~$\Abig$ is set to~$E(a_\text{old})$ in line~\ref{li:Abig} of $\AddAttractionPoint$ for some substream~$i \in [K]$, then we have~$|A_i|>\maxK\geq\parK$ and all points in~$A_i$ will remain active until~$E(a_\text{old})$.
    Hence, if~$\Abig=E(a_\text{old})>T$, it follows that there is no solution with radius~$\leq\gamma$.
\end{proof}

\begin{lemma}\label{lem:ek+1Certificate}
    Let $T$ be the current time and $\parK\in[\maxK]$. If $t_{\parK}> T$, then the solution to the metric $\parK$-center has radius greater than $\gamma$.
\end{lemma}
\begin{proof}
    Let $p$ be the item that caused $t_{\parK}$ to be set to~$E=E(p)$. 
    For all~$i \in [\parK]$, define $q_i\in R_i^{S(p)-1}\cup O_i^{S(p)-1}$ to be the point that expires at time $e_i^{S(p)-1}$, i.e., $e_i^{S(p)-1}=E(q_i)$. By Invariant~\ref{in:orderEi} and~\ref{in:distLongRO}, it holds that $d(q_i, q_j)>2\gamma$ for all $i\neq j\in [\parK]$. 
    Further, for all $i\in [\parK]$ it holds that $t_{\parK}=E(p)\leq e_i^{S(p)-1}=E(q_i)$ by line~\ref{li:E>ei} of \Cref{alg:updatek-center}. Therefore, $d(p, q_i)>2\gamma$ by line~\ref{li:distToRepr} of \Cref{alg:updatek-center}. 
    Hence, $q_1, \dots, q_k$ and $p$ are all still active at time~$T$ and have pairwise distance greater $2\gamma$. It follows that the solution to the metric $\parK$-center at time~$T$ has radius greater than~$\gamma$.
\end{proof}

\iffalse
\begin{lemma}
    Consider a point $a$ that is removed from~$A_i$ at some time~$t+1$, i.e., $a\in A_i^{t}$ and $a\notin A_i^{t+1}$.
    If~$R^t(a)$ is removed from~$R_i \cup O_i$ at some time~$t'<E(R^t(a))$, then one of the following holds:
    \begin{enumerate}[(i)]
        \item $|A_i^T|> k$ for~$t' \leq T < E(R^t(a))$.
        \item There is an~$\ell < i$ and a point $a_\ell\in A_\ell^{t'}$ such that $E(R^t(a)) \leq E(R^{t'}(a_\ell))$ and $d(a, a_\ell)\leq 6\gamma$.
    \end{enumerate}
\end{lemma}
\begin{proof}
    By our assumption, the point~$R^t(a)$ is removed from~$R_i \cup O_i$ before it expires.
    If the removal occurs in line~\ref{li:boundSizeO}, this implies Case~(i).
    Otherwise, it occurs in UpdateRepresentative.
    Then there exists a value $\ell<i$ and a point $a_\ell\in A_\ell^{t'}$ such that $E(R^{t}(a))\leq E(R^{t'}(a_\ell))$ and $d(R^{t}(a), R^{t'}(a_\ell))\leq 2\gamma$. By Invariant~\ref{in:distAR(a)}, it holds that
    \[d(a, a_\ell)\leq d(a, R^{t}(a))+d(R^{t}(a), R^{t'}(a_\ell))+ d(R^{t'}(a_\ell), a_\ell)\leq 6\gamma.\qedhere\]
\end{proof}
\fi

To handle the case that a solution of radius $\gamma$ exists, we prove some distance bounds between different points of the stream.

\begin{lemma}\label{lem:distToAStep}
    Let $R^{\infty}(a)$ be the last representative point of an attraction point $a$. Further, let $t$ be the time and~$i$ the index such that $R^{\infty}(a)\in R_i^{t-1}\cup O_i^{t-1}$ and $R^{\infty}(a)\notin\bigcup_{j\in[\maxK]} R^{t}_j\cup O^{t}_j$. Then one of the following holds:
    \begin{enumerate}[(i)]
        \item $t=E(R^{\infty}(a))$, 
        \item $\Abig^T\geq E(R^{\infty}(a))$ for all $t\leq T$, or
        \item there is an~$\ell < i$ and a point $p\in R_\ell^{t}\cup O_\ell^{t}$ such that $E(R^{\infty}(a)) \leq E(p)$ and $d(a, p)\leq 4\gamma$.
    \end{enumerate}
\end{lemma}
\begin{proof}
    There are the following four possibilities for removing $R^{\infty}(a)$ from $\bigcup_{j\in[\maxK]}R_j\cup O_j$.
    \begin{enumerate}[(a)]
        \item The point $R^{\infty}(a)$ expires. Then, $t=E(R^{\infty}(a))$.
        \item The point is removed in line~\ref{li:boundSizeO} of $\AddAttractionPoint$. Then it holds that $\Abig\geq E(R^{\infty}(a))$.
        \item The point is removed from $\bigcup_{j\in[\maxK]}A_j\cup R_j\cup O_j$ in $\DiscardCoveredPoints(p, \ell)$ for some $\ell<i$ and a point~$p$. Then, $E(R^\infty(a))<E(p)$ and $d(p, R^\infty(a))\leq 2\gamma$ or $d(p, a)\leq 2\gamma$. By Invariant~\ref{in:distAR(a)}, it holds that $d(a, R^\infty(a))\leq 2\gamma$ and hence $d(a, p)\leq 4\gamma$.
        Note that if \Cref{alg:updatek-center} calls $\DiscardCoveredPoints(p', i')$ at time~$T$ then if later at time~$T$ $\DiscardCoveredPoints(p'', i'')$ is called it holds that $i''>i'$. Hence, it follows that $p\in R_\ell^{t}\cup O_\ell^{t}$ and (iii) follows. 
        \item The point $R^\infty(a)$ is added as an attraction point to $A_j$ at some time $t'<t$ for an $j<i$ in line~\ref{li:AddToRA} or \ref{li:AddOtoA} and then at time~$t$, the representative of $R^\infty(a)$ is updated to $R(R^\infty(a))\neq R^\infty(a)$. Then, $d(a, R^t(R^\infty(a)))\leq d(a, R^\infty(a))+d(R^\infty(a), R^t(R^\infty(a)))\leq 4\gamma$ by Invariant~\ref{in:distAR(a)}. Hence, (iii) follows as $R^t(R^\infty(a))\in R_\ell^t$ for some $\ell\leq j$ by \Cref{obs:MoveOnlyLowerStream}.
        \qedhere
    \end{enumerate}
\end{proof}

\begin{lemma}\label{lem:ROinduction}
    Let $a$ be an attraction point and let $R^{\infty}(a)$ be its last representative point. Further, let $t$ be the time and $i$ the index such that $R^{\infty}(a)\in R_i^{t-1}\cup O_i^{t-1}$ and $R^{\infty}(a)\notin \bigcup_{j\in [\maxK]} R_j^{t}\cup O_j^t$.
    Then, for any time $t\leq T<E(R^{\infty}(a))$ such that $\Abig^T\leq T$, 
    %$|\bigcup_{j\in[\maxK]}A_j^T|\leq \maxK$, 
    there exists an $\ell \leq i$ and a point~$r\in R_\ell^T \cup O_\ell^T$ such that $d(a, r)\leq (i-\ell)6\gamma+2\gamma$.
\end{lemma}
\begin{proof}
    We show the claim by induction over~$i$. If $i=1$, then the claim follows by~\Cref{lem:distToAStep} and Invariant~\ref{in:distAR(a)}. Now let $i>1$. By~\Cref{lem:distToAStep}, there exists a point $p_m\in R_m^{t}\cup O_m^{t}$ with $m < i$, $E(R^{\infty}(a)) \leq E(p_m)$ and $d(p_m, a)\leq 4\gamma$. Let $a_m$ be the corresponding attraction point to $p_m$, i.e., at some time $t'$ it holds that $R^{t'}(a_m)=p_m$. Then, $d(a, a_m)\leq 6\gamma$ and $d(a, R^T(a_m))\leq 8\gamma$ by Invariant~\ref{in:distAR(a)}. Hence, the claim holds if $R^T(a_m)\in \bigcup_{j\in [\maxK]}R_j^T\cup O_j^T$. Otherwise by \cref{obs:MoveOnlyLowerStream}, it follows that there exists a time $t'<T$ and an index $i_m\leq m$ such that $R^{\infty}(a_m)\in R_{i_m}^{t'-1}\cup O_{i_m}^{t'-1}$ and $R^{\infty}(a_m)\notin \bigcup_{j\in[\maxK]} R_{j}^{t'}\cup O_{j}^{t'}$. It holds that $T<E(R^\infty(a))\leq E(p_m)\leq E(R^\infty(a_m))$.
    Hence, by induction there exists an $\ell\leq i_m$ and a point $r\in R_{\ell}^T\cup O_{\ell}^T$ such that $d(a_m, r)\leq (i_m-\ell)6\gamma+2\gamma$. Therefore, 
    \[d(a, r)\leq d(a, a_m)+d(a_m, r)\leq 6\gamma+(i_m-\ell)6\gamma+2\gamma\leq (i-\ell)6\gamma+2\gamma.\qedhere\]
\end{proof}

Next we obtain a distance bound for any active point~$p$ that expires later than $\max\{t_{\parK}, \Abig\}$.

\iffalse
\begin{lemma}\label{lem:distancePToA}
    Let $\parK\in[\maxK]$ and~$(p,S,E)$ be an item and~$T$ a time with~$S \leq T < E$ and~$t_{\parK}^T<T$.
    Then there exists a time~$t \leq S$ and a point $a\in \bigcup_{i\in[\parK]}A_i^t$ with $d(p,a)\leq 4\gamma$ and~$T < E(R^t(a))$.
\end{lemma}
\begin{proof}
    Define $\ell=\min\{i\in [K]\mid E> e_i^{S-1}\}\cup \{K+1\}$. 
    By \Cref{obs:relationET}, it holds that $e_{\parK+1}^{S-1}\leq t_{\parK}^{S-1}\leq t_{\parK}^T<T<E$. 
    Hence, $\ell\leq \parK+1$.
    %As $t_{\parK}^S\leq t_{\parK}^T<T<E$, 
    As $\parK\in[\maxK]$,
    one of the following holds when~$p$ is inserted:
    \begin{enumerate}
        \item There exists a point $q\in\bigcup_{i\in [\parK+1]} O_i^S\cup R_i^S$ with $E(q)\geq E>T$ and $d(p, q)\leq 2\gamma$. Let $t\leq S$ be the time at which~$q$ is inserted. Then, there exists a point $a\in\bigcup_{i\in [\parK+1]} A_i^t$ with $d(a,q)\leq 2\gamma$ and~$R^t(a)=q$. Hence, $d(p,a)\leq 4\gamma.$
        \label{case:d(a,p)leq4Gamma}
        %\item There exists an $q\in R_i^S$ with $E(q)\geq E$ and $d(p, q)\leq 2\gamma$. As $q\in R_i^S$ there exists a point $a\in A_i^S$ with $d(a, p)\leq d(a_i, q)+d(q, p)\leq 4\gamma$. 
        \item There exists a point $a\in \bigcup_{i\in [\parK+1]}A_i^S$ with $d(p, a)\leq 2\gamma$ and $E(R^S(a))\geq E>T$. \label{case:d(a,p)leq2Gamma}
    \end{enumerate} 
    If $a\in A_{k+1}^S$ then $\ell=k+1$ and $t_{\ell-1}^S=t_k^S$ is set to $E$ in line~\ref{li:setT}. Hence $E=t_{k}^S\leq t_{k}^T$, which contradicts the assumption that $t_k^T<T<E$ and the claim follows.
\end{proof}
\fi

\begin{lemma}\label{lem:distanceToRO}
     Let $T$ be the current time and $\parK\in [\maxK]$. If $\max\{t_{\parK}^T, \Abig^T\}<T$, then for any item~$p$ active at time~$T$, i.e.,  $S(p)\leq T<E(p)$, there exists a point $r\in \bigcup_{i\in[\parK]}R_i^T\cup O_i^T$ with $d(p,r)\leq 6\parK\gamma$.
\end{lemma}
\begin{proof} Define $S=S(p)$.
    It holds that $t_\parK^S\leq t_\parK^T<T$. Hence, by \Cref{alg:updatek-center} there exist
    \begin{enumerate}
        \item a point $R^S(a)\in \bigcup_{j\in [\parK]}R_j^S\cup O_j^S$ such that $T<E(p)\leq E(R^S(a))$ and $d(p, R^S(a))\leq 2\gamma$, or
        \item a point $a\in \bigcup_{j\in [\parK]}A_j^S$ such that $T<E(p)\leq E(R^S(a))$ and $d(p, a)\leq 2\gamma$. 
    \end{enumerate}
    Therefore, for any time $t\geq S$, it holds that $d(p, R^t(a))\leq 6\gamma$, since $d(a, R^S(a))\leq 2\gamma$ and $d(a, R^t(a))\leq 2\gamma$ by Invariant~\ref{in:distAR(a)}.
    Hence, in the case that $R^T(a)\in \bigcup_{i\in [\parK]}R_i^T\cup O_i^T$ the claim follows. 
    Otherwise by \Cref{obs:MoveOnlyLowerStream}, there exist a time $t\leq T$ and an index $i\leq \parK$ such that $R^\infty(a)\in R_i^{t-1}\cup O_i^{t-1}$ and $R^\infty(a)\notin \bigcup_{j\in [\maxK]} R_i^{t}\cup O_i^{t}$. Then by \Cref{lem:ROinduction}, there exists an $\ell\leq i\leq \parK$ and a point $r\in R_\ell^T\cup O_\ell^T$ such that $d(a, r)\leq (i-\ell)6\gamma+2\gamma$. Hence, 
    \[d(p, r)\leq d(p, a)+d(a, r)\leq 4\gamma+ (i-\ell)6\gamma+2\gamma\leq 6\parK \gamma.\qedhere\]
\end{proof}

Running the algorithm described in \Cref{sec:algorithm} for all values of $\gamma\in \{(1+\eps)^i \mid i=1, \dots, \lceil\log_{1+\eps}\Delta\rceil\}$ in parallel results in the following theorem.

\ThmKcenter* % restatable theorem 

% \begin{theorem}
% \label{thm:kcenter}
%     In every metric space, there exists an algorithm storing $\calO((\maxK^2/\eps)\log \Delta)$ words in the expiration streaming model that can return at any time and for every value $\parK\leq \maxK$ a $(6\parK+2)(1+\eps)$-approximate solution for $\parK$-center.
% \end{theorem}

\begin{proof}
    The space bound follows by \Cref{lem:spaceBoundARO}.
    Let~$\gamma'$ be the smallest estimate for which the algorithm returns a solution~$C$ with at most~$\parK$ elements.
    Then, it follows from \Cref{lem:sizeACertificate} and \Cref{lem:ek+1Certificate} that $\text{OPT} > \gamma'/(1+\eps)$.
    By \Cref{lem:distanceToRO}, it holds that for every active item $p$ there exists a point $r\in \bigcup_{i\in [\parK]}R_i\cup O_i$ such that $d(p, r)\leq 6\parK\gamma'$. Hence, there exists a point $c\in C$ with $d(p, c)\leq (6\parK+2)\gamma' \leq (6\parK+2)(1+\eps)\text{OPT}$.
\end{proof}
\begin{corollary}
    In every metric space, there exists an algorithm in the expiration streaming model that can return at any time and for every value $\parK\leq \maxK$ 
    \begin{enumerate}[(1)]
        \item a $(6\parK+2+\eps)$-approximate solution for $\parK$-center storing $\calO((\maxK^3/\eps)\log \Delta)$ words, \label{item:6k+2+eps}
        \item a $(6\parK+3+\eps)$-approximate solution for $\parK$-center storing $\calO((\maxK^2/\eps+\maxK^3)\log \Delta)$ words, and \label{item:6k+3+eps}
        \item a $(7\parK)$-approximate solution for $\parK$-center with $\parK\geq 3$ storing $\calO(\maxK^2\log \Delta)$ words.\label{item:7k}
    \end{enumerate}
\end{corollary}
\begin{proof}
    We use \Cref{thm:kcenter} and use different values for $\eps$ to get the other bounds. Setting $\eps={\eps'}/{(8\maxK)}$ results in (\ref{item:6k+2+eps}). Setting $\eps=\min\{{(6\maxK)^{-1}}, \eps'\}$ gives (\ref{item:6k+3+eps}).
    Setting $\eps=1/24$ yields~(\ref{item:7k}). 
\end{proof}

\section{Counting in Expiration Streams}
\label{sec:GE-count}

Here we consider the central building block problem of counting elements, under the expiration streaming model after inserting $n$ items.  In the standard insertion-only model, exact counting is  possible in $\calO(\log n)$ bits of space.  Then Datar \etal~\cite{DGIM02} famously introduced the sliding-window model with focus on counting, and showed that exact counting with window size $w$ requires $\Theta(w)$ bits, and approximating a count within a $1+\eps$ factor can be done in $\Theta((1/\eps) \log^2 w)$ bits.  
The expiration model is more challenging, and so for our new upper bounds we work under the simplified setting where we count ``words.''  This allows us to work in the more general "timestamp" model where times of insertions, expirations, and queries can have any real value (those representable by a ``word").  

In particular, the counting problem is more challenging in the expiration model because we need to keep track of when items expire, which we only gain information at the time they are inserted.  
In short, if item $i$ is inserted with an expiration time $E_i$, then at some future time $t$ the counting task is to maintain how many items have been inserted, but not yet expired.  

In this section we focus on maintaining a \emph{sketch} of the data.  This is a small space data structure, where one can ask queries, and we guarantee approximation factors on the answer provided.  The most central goal is to trade off the space required with the approximation factor.  In the most basic problem, the query is a count at some future time $t$; we will define both additive and relative error guarantees.

\subparagraph*{Warmup: Consistent Expiration Model.}
We maintain two counters: $C_I$ and $\hat C_E$.  The first $C_I$ is the exact count of all items inserted so far, and can be done storing $\calO(1)$ words.  Then let $C_E$ be the number of items that have expired so far.  It is not obvious how to maintain this exactly in small space, so we instead maintain the approximate count $\hat C_E$ which satisfies $C_E \leq \hat C_E \leq C_E + \eps n$.  
To maintain $\hat C_E$, we just need to store the at most $1/\eps$ expiration time points $t_1, t_2, \ldots$ so that there are $\eps n$ expirations between~$t_j$ and~$t_{j+1}$.  
This algorithm works inductively.  Consider we have set time point $t_j$, signifying when $j \eps n$ items have expired.   Then we keep a counter of how many new items are inserted until we have witnessed $\lfloor \eps n \rfloor$; %\andre{shouldn't this be $\eps n$?} \jeff{thanks, fixed}
at this point we mark that expiration time as~$t_{j+1}$.  

%By adding a new expiration check point $t_j$ when an additional $\eps C_I$ have expired, the space increases to $\calO((1/\eps)\log n)$ and we can reach the stronger relative error guarantee.   

The final estimate $\hat n_{\text{now}} = C_I - \hat C_E$ has properties described in the following theorem.  

\begin{theorem}
  In the consistent expiration model with $n$ items and error parameter $\eps \in (0,1/2]$, we can approximate the number of not yet expired items $n_{\text{now}}$ with $\hat n_{\text{now}}$ such that 
%    \begin{itemize}
%        \item 
        $\hat n_{\text{now}} \leq n_{\text{now}} \leq \hat n_{\text{now}} + \eps n$ storing $\calO(1/\eps)$ words.
%        \item $\hat n_{\text{now}} \leq n_{\text{now}} \leq \hat n_{\text{now}} + \eps n_{\text{now}}$ with  $\calO((1/\eps) \log n)$ space.
%    \end{itemize}
\end{theorem}

%\jeff{We can extend to relative error counters, but not sure I want to get into it in this warm up, since involved some updating tricks we develop below.}

\subparagraph*{Hardness of Exact Counting.}
Here we show that some approximation is necessary to count under the consistent expiration model with sublinear space. We reduce to the INDEXING problem that considers a set $n$ bits.  Alice holds the bits and Bob will ask a question about one bit $i \in [n]$.  Alice can communicate to Bob, but not the other direction.  
It is known that for Bob to be able to determine the value of the $i$th bit with constant probability of success, Alice needs to communicate $\Omega(n)$ bits~\cite{Kushilevitz_Nisan_1996}.  

\begin{lemma}
    \label{lem:hard-count}
    In a consistent expiration stream of size $n$, to be able to answer the count of active items at some unknown future time $t$ with at least a constant probability of success, the sketch requires $\Omega(n)$ space.  
\end{lemma}
\begin{proof}
    The construction inserts all $n$ points.  It structures the expirations so each item $i$ expires at either time $n + 2i$ or $n + 2i+1$.  We encode the indexing problem such that if Alice's $i$th bit is $0$, then the $i$th expiration is at time $n + 2i$ and if it is $1$, then the $i$th item expires at $n + 2i + 1$.  

    Now for Bob to exactly answer a query to item $i$, the sketch needs to return the count at time $t = n + 2i$.  If the answer is $n-i$, then the $i$th bit is $0$; if it is $n-i+1$, then the $i$th bit is $1$.  Thus, if we can answer the counting question exactly, we can solve INDEXING.  Thus it requires $\Omega(n)$ space.  
\end{proof}

Recall that the consistent expiration model is contained in the general expiration model.  Thus exact counting is hard under the general expiration model;  $\Omega(n)$ space is required to solve it exactly.  

% \jeff{TODO: formalize}
% Conjecture:  It required $\Omega(n)$ space to be able to report this quantity exactly.  Probably a reduction to INDEXING.  

% Sketch:  Insert all points.  Then delete all points.  
% Point $i$ deleted at time $2i$ or $2i+1$;  determine this by bit $b_i$.  To know number of points at time $2i$, you need know bit $b_i$.  Reduce to INDEXING.  

% Also applies to consistent expiration model.  

\subsection{Quantile Sketching}
The quantile problem considers a set of $n$ items in a stream~$X$, where the items come from a domain~$\mathcal{X}$ that can be totally ordered under a $\calO(1)$-time comparator operator.  
The goal is to store a small space sketch of small size so that for a query $v$ then one can approximately return a count $\hat C_v$ of the number of items which satisfy $x < v$, where the precise total count is $C_v = |\{x \in X \mid x < v\}|$.  
Additive error quantile sketches ensure the guarantee that $|\hat C_v - C_v| \leq \eps n$.  
Relative error quantile sketches ensure that $|\hat C_v - C_v| \leq \eps C_v$.  

After a long sequence of developments, the best space bound for additive error quantile sketch under the insertion only streaming model is the KLL sketch~\cite{karnin2016optimal}, which only requires $\calO((1/\eps) \log \log (1/\eps \delta))$ words of space;
it is randomized and fails with probability $1-\delta$.  This probabilistic result holds simultaneously for all quantile queries $v$; if one only requires it to hold for a single query, the $\eps$ can be removed from the $\log \log$ term, for a space bound of $\calO((1/\eps) \log \log (1/\delta))$.

Note that since these algorithms work for any domain with a comparator operator, to provide these general guarantees it must store elements from the set which demark the boundaries between queries $v$.  The stored elements have $\calO(\eps n)$ elements between them in the sorted order.  Moreover, the space bounds assume it takes $\Theta(1)$ space to store an element of $\mathcal{X}$.  Thus it implies that at least $\Omega(1/\eps)$ space is required.  

%Zhao \etal~\cite{zhao2021kll} studied an extension of the KLL sketch to work under insertion+deletion model.
%https://www.vldb.org/pvldb/vol14/p1215-zhao.pdf
%This is studied under the \emph{$\alpha$-bounded deletion model} where at any point, only a $(1-1/\alpha)$-fraction of all points inserted could have been deleted.  Then they show a variant of the KLL sketch, allowing deletions, can be maintained in $\calO((\alpha^{1.5}/\eps)\log \log(1/\eps \delta))$ space in words. 

For the relative error quantile sketch, Gribelyuk \etal~\cite{gribelyuk2025near,cormode2021relative} show how to maintain a sketch of size $\tilde{\calO}((1/\eps) \log (\eps n) \log^3\log(1/\delta))$ in words, it fails with probability at most $\delta$. Cormode \etal~\cite{cormode2021relative} also note that $\Omega((1/\eps) \log (\eps n))$ space is required.

\subsection{Expiration Counting via Quantiles}
\label{sec:GE-count-quantiles}

In this subsection we argue that quantile sketches can be used to solve the counting problem under the general expiration model.  We apply a quantile sketch on expiration times $E_i$, with a reversing of the comparison-order so the ``minimum" possible value is the latest possible time.  Then a quantile sketch for query at time $t$ returns $\hat C_t$ which estimates the number of items that have been inserted but at time $t$ have not yet expired.  

We state the following theorem for the implications of this employing the additive-error and relative-error quantile sketch bounds, respectively.  

\begin{theorem}
\label{thm:GE-counter}
    Under the expiration streaming model with $n$ inserted items, each with an expiration time, we can maintain an approximation of the count $n_{\text{now}}$ of items that were inserted, but not yet expired at the current time.  Our estimate $\hat n_{\text{now}}$ satisfies 
    \begin{itemize}
        \item \underline{$\eps$-additive error:} $|n_{\text{now}} - \hat n_{\text{now}}| \leq \eps n$ storing $\calO((1/\eps) \log \log (1/ \delta))$ words and
        \item \underline{$\eps$-relative error:} $|n_{\text{now}} - \hat n_{\text{now}}| \leq \eps n_{\text{now}}$ storing $\tilde{\calO}((1/\eps) \log (\eps n)\log^3 \log(1/\delta))$ words.
    \end{itemize}
Both cases are randomized algorithms which succeed with probability at least $1-\delta$ and hold for approximating $n_{\text{now}}$ for at any one query at any one point in the stream.
\end{theorem}

In fact, the expiration counting problem is intimately tied to the insertion-only quantiles problem.  Consider freezing all insertions, and define  $C_t$ as the count at a future time $t$ and its approximation as $\hat C_t$.  If we just wait until time $t$ (there are no additional insertions), then for correctness, the sketch must be able to satisfy all such future time queries.  The maintenance of the count at future times is precisely a distribution estimate as preserved by quantiles; for the additive error case, they both maintain the Kolmogorov-Smirnov distance between distributions up to $\eps$ error.  Thus the space lower bounds for insertion-only quantiles estimates apply to counting in the expiration streaming model. We can hence state the following theorem.

\begin{theorem}
\label{thm:GE-count-LB}
Consider the expiration streaming model where storing a time and a count each require $\Omega(1)$ space in words.  To maintain a sketch that at any future time $t$ can maintain an
\begin{itemize}
  \item \underline{$\eps$-additive error} estimate requires $\Omega(1/\eps)$ space in words; and 
  \item \underline{$\eps$-relative error} estimate requires $\Omega((1/\eps)\log (\eps n))$ space in words.  
\end{itemize}
\end{theorem}

As noted above, in the insertion-only, the counts can be maintained exactly in $\calO(\log n)$ bits of space, or $\calO(1)$ words of space.  Recall, in the sliding-window model~\cite{DGIM02}, $\eps$-relative error count requires $\Theta((1/\eps) \log^2 w)$ bits for size $w$ sliding-windows.

\subsection{Implications of Counting}

A counter is a classic tool in many different, especially linear, sketches.  As a demonstration, we now show that we can leverage this approximate counter inside a Count-Min sketch to solve the frequent items problem.  

These approach consider the frequency approximation problem where a stream $X = \langle x_1, x_2, \ldots \rangle$ observes a sequence of items each $x_i \in [U]$ for some very large domain $[U] = \{1, 2, \ldots, U\}$ (e.g., all IP addresses).
The goal it be able to answer the frequency $f_j = |\{x \in X \mid x = j\}|$ of any item $j \in [U]$.

\subparagraph*{Count-Min Sketch.}
The Count-Min Sketch~\cite{cormode2005improved} maintains a $k \times t$ table of counters $C$ for $k = 2/\eps$ and $t = \log_2 (1/\delta)$.
It then uses a set of $t$ $4$-way independent hash functions $h_r : [U] \to [k]$ (for $r \in [t]$). At an insertion of an object $x$ the sketch increments one counter in each row of the table $C$. In particular, for row $r$ it updates counter $c = h_r(x)$ as  $C_{r,c} = C_{r,c} + 1$.  Then on a query $q \in [U]$, the sketch returns the smallest associated counter
\[
S(q) = \min_{r \in [t]} C_{r,h_r(q)}.  
\]
It guarantees that $S(q) \geq f_q$, and with probability at least $1-\delta$ that 
\[
S(q) - f_q \leq \eps n.
\]

Note that as a \emph{linear sketch}, it can handle streams that provide instantaneous insertions or deletions.  We simply increment or decrement each associated counter.  However, in the expiration streaming model this approach does not carry over, since we would need to retain the expiration times in memory until the time it expires.  

We can augment this sketch to fit in the expiration model by replacing each exact counter $C_{r,c}$ with an additive error expiration counter $\hat C_{r,c}$, as described in Theorem \ref{thm:GE-counter}.  
We make the observation that our new estimate $\hat S(q)$ for any query $q \in [U]$ satisfies, with probability~$1-\delta/t$.  
\[
|S(q) - \hat S(q)| \leq \left| \min_{r \in [k]} ( C_{r,h_r(q)}  -  \hat C_{r,h_r(q)}) \right| \leq \eps n, 
\]
and as a result, we have $|f_q - \hat S(q)| \leq 2 \eps n$.  
Adjusting for constant factors in $k$ and $t$ we can obtain the following theorem.

\begin{theorem}
    \label{thm:GE-CMS}
    By using expiration approximate counters in a Count-Min Sketch $\hat S$, storing $\calO((1/\eps^2) \log(1/\delta)\log\log(1/\delta))$ words, we can handle expiration updates and guarantee with probability at least $1-\delta$ that for any query $q$ we can approximate the frequency 
    \[
|\hat S(q) - f_q | \leq \eps n.
    \]
\end{theorem}

We do not know how to obtain improved bounds using other frequency approximation sketches such as the Count Sketch~\cite{charikar2004finding} or the Misra-Gries Sketch~\cite{misra1982finding}.  In particular, the Count Sketch would be of interest since it could potentially apply to various matrix sketching models~\cite{TCS-060}.  We leave this extension, or showing a lower bound that it is not possible, to future work.

\section{Sampling in Expiration Streams}
\label{sec:GE-Sampling}

In this section we consider another fundamental building block used in many algorithms: a random sample of size $k$.
For insertion-only streams, the classic method to generate
a single sample ($k=1$) is reservoir sampling~\cite{vitter1985random}.
It extends easily to a sample of size $k$
by running $k$ independent copies of reservoir sampling,
which generates \emph{with-replacement} samples. 
% since items may be chosen more than once.
Alternatively, one can maintain a single reservoir of size $k$,
and this provides \emph{without-replacement} samples.
This method is a bit harder to analyze,
but may lead to smaller variance when used in some estimators.
These methods can be generalized to \emph{weighted sampling},
where each item $i$ is sampled with probablity proportional to its weight $w_i>0$.
For $k=1$, reservoir sampling directly extends to weighted sampling
by maintaining the total weight $W$
% and keeping each new item proportionally to $w_i/W$.
and storing each new item $i$ in the reservoir (instead of the current sample)
with probablity $w_i/W$.
This method easily extends to $k$ with-replacement samples,
by running $k$ independent copies. 
Producing $k$ without-replacement samples has been explored as well.
The bottom-$k$ sampling approach of Cohen and Kaplan~\cite{CK07}
(see also \cite{Cohen97,CK08})
assigns each item $i$ a value $u_i \sim \mathsf{Exp}(w_i)$
and maintains the $k$ items with smallest values.
Efraimidis and Spirakis~\cite{ES06} alternatively choose $\bar{u}_i \sim \mathsf{Unif}(0,1]$
and select the $k$ items with largest value $\bar{u}_i^{1/w_i}$.
These methods are in fact equivalent
because one way to generate $u_i \sim \mathsf{Exp}(w_i)$
is to draw $\bar{u}_i \sim \mathsf{Unif}(0,1]$
and set $u_i = -\ln(\bar{u}_i)/w_i$. % see \cite{Cohen97}.

In the sliding-window model, Babcock, Datar, and Motwani~\cite{BabcockMR02}
provided several algorithms for uniform sampling (i.e., unit weights).
Most are for sequence-based windows,
where the number of active items remains fixed
(because at each time step, one item enters the window and one leaves)
and these algorithms rely on this structure.
They also provide another algorithm for timestamp-based windows,
where the number of active items might vary significantly
(because multiple items can enter and leave the window in a non-proportional manner).
This algorithm assigns each item a random value $u_i \sim \mathsf{Unif}(0,1]$,
called \emph{priority}, 
and the active item with smallest priority $u_i$ is chosen as the sample.
Running $k$ copies of this algorithm in parallel
maintains $k$ samples with replacement,
and with high probability uses $\calO(k \log n)$ words of space.  

In expiration streams, we are able to replicate many of these methods, with space overhead $\calO(\log n)$ per sample, similar to sliding-windows.
Our algorithms produce a random sample from the currently active set, and can furthermore produce a sample for any future time $t$ among the items already inserted (i.e., assuming no additional insertions).
We present the case of uniformly sampling $k$ items without replacement in \cref{sec:GE-reservoir}, and that of weighted sampling of $k$ items with replacement in \cref{sec:weighted}.
Note that the former can easily be adapted to with-replacement sampling by simply running $k$ independent copies of the $k=1$ algorithm.
% In both cases, we consider a single sample ($k=1$), and one can run $k$ independent copies of this algorithm to produce $k$ with-replacement samples.

These results have a wide variety of implications,
including maintaining an expiration sketch for several problems.
Recall that a \emph{sketch} is a small-space data structure that provides approximate answers to queries,
and often trades off the space bound with the error guarantee.
An \emph{expiration sketch} handles items that are inserted with an expiration time, and can answer a query about the active items at any future time $t$.  
Within this section, we formalize the following results,
where an item can be a real number or a vector in $\RR^d$ (e.g., a point or a row in a matrix);
we count space bounds in words and allow constant probability of failure.  
To simplify the list below, we assume $d$ is fixed in results that consider points, but not in matrix settings.

\begin{enumerate} \compactify
\renewcommand{\labelenumi}{(\alph{enumi})}
\item
  A quantile sketch with $\eps$-additive error using $\calO((1/\eps^2)\log n)$ space; see \Cref{thm:GE-quantile}.
\item
  A sketch for approximate range counting of shapes in $\RR^d$ whose VC-dimension is $\nu$,
using $\calO((\nu /\eps^2)\log n)$ space; see \Cref{thm:GE-ARC}.
\item
  A sketch for PAC learning of Boolean function classes with VC-dimension $\nu$, 
  using space in $\calO((\nu/\eps^2)\log n)$; see \Cref{thm:GE-learning}.
\item
  A sketch for $\frac{1}{t}\|w\|_2^2$-regularized logistic regression in $\RR^d$ (and other classification loss functions),
  using $\calO((t^3/\eps^2)\log n)$ space; see \Cref{thm:GE-logistic-reg}. 
\item
  A sketch for $\eps$-approximate kernel density estimation in $\RR^d$
  using $\calO((1/\eps^2)\log n)$ space; see \Cref{thm:GE-KDE}.  
 \item
   A sketch of a matrix $A \in \RR^{n \times d}$ presented in row-order with expirations,
   for approximating the best projection onto rank-$k$ matrices with additive error $\eps \|A\|_F^2$,
   using space in $\calO((dk/\eps^2) \log \|A\|_F^2)$; see \cref{thm:GE-mat-proj}.
 \item
   Another sketch of a matrix $A \in \RR^{n \times d}$ presented in row-order with expirations,
   for approximating the covariance matrix $A^T A$ with additive error $\eps \|A\|_F^2$, 
   using $\calO((d^2/\eps^2) \log \|A\|_F^2)$ space; see \cref{thm:GE-mat-cov}. 
\end{enumerate}

\subsection{Uniform Sampling with Expirations}
\label{sec:GE-reservoir}

For insertion-only streaming, reservoir sampling~\cite{vitter1985random} (and its variants) is a powerful technique to maintain a uniformly random set of $k$ items, called a \emph{reservoir}, from those inserted to the stream.
In the special case $k=1$, when the $i$-th item arrives, it is kept with probability $1/i$, and if so it replaces the item in the reservoir.
For general $k$, we traditionally keep $k$ one-item reservoirs independently, although in practice it is often better to have $k$ items \emph{without replacement},
which can be implemented by generating for each $i$-th item a random $u_i \sim \mathsf{Unif}(0,1]$,
which is often called \emph{priority},%
\footnote{Note that priority sampling~\cite{duffield2007priority} use $1/u_i$,
  although for our unweighted case the technique appeared earlier~\cite{Luby86,BabcockMR02,BorodinNR02}.
}
and maintaining a single size-$k$ reservoir containing the $k$ items with the smallest $u_i$ values. 
Its size is related to the $n$-th Harmonic number $H_n=1+\tfrac12+\ldots+\tfrac1n \in (\ln n, \ln n +1)$.
As a warm-up, even before introducing expirations, we generalize 
the with-replacement result of~\cite{BabcockMR02} to the without-replacement setting. 

\begin{lemma}
    \label{lem:k-res-sample}
    In insertion-only streams, the total number of items ever stored (possibly temporarily) in a size-$k$ without-replacement reservoir
    is at most $k H_n$ in expectation, and is at most 
    $k (8H_n + \ln(1/\delta))$ 
    %$k (\ln n + 1) + \sqrt{k(\ln n + 1)/\delta}$
    with probability at least $1-\delta$.
\end{lemma}

\begin{proof}
Let $X_i$ be an indicator random variable for the event that item $i$ is stored in the reservoir when it arrives,
and let $S_i = \sum_{j=1}^i X_j$ be the number of items among $1,\ldots,i$ ever placed in the reservoir.

We argue that $\Pr[X_i = 1] = \min\{1, k/i\}$. 
Indeed, for $i \leq k$ we have deterministically $X_i=1$ and $S_i = i$.
For $i > k$, let us sort the values $u_1,\ldots,u_i$ seen so far, 
and then $X_i =1$ if $u_i$ is in the bottom $k$ values. 
Since each $u_{i'}$ is from the same distribution ($\mathsf{Unif}(0,1]$),
where $u_i$ falls in the sorted order of $u_1,\ldots,u_i$ is uniform over the $i$ possible positions,
and moreover it is independent of the relative order of the earlier $u_{i'}$.
Thus, the probability that $u_i$ is in the bottom $k$ is exactly $k/i$,
\emph{independently} of the same event for $i'<i$
(i.e., whether each earlier $u_{i'}$ was in the bottom $k$ when it arrived).
Note that we condition here only on the relative order, not the actual values of $u_{i'}$.  
Thus 
\[
\E[S_n] = \sum_{i=1}^n \E[X_i] = \sum_{i=1}^n \min\{1,k/i\}
\leq k \sum_{i=1}^n 1/i = k H_n .
\]
More precisely $\E[S_n] = k H_n - k H_k + k$, where the last two terms adjust for the determinism of the first $k$ items.   

To get the tail bound, we can use a result of Mulmuley~\cite[Theorem A.1.7, part 2]{Mul94} which states that given a set of $n$ independent generalized harmonic random variables, namely $\Pr[X_i=1] \leq \mu/i$,
letting $S_n = \sum_{i=1}^n X_i$,
we have $\Pr[S_n \geq c \mu H_n] \leq \exp(-H_n[1+c \ln(c/e)])$ for all $c>1$.
This holds for our setting with $\mu = k$,
and now setting $c = 8 + \frac{1}{H_n} \ln(1/\delta)$,
we have $\ln(c/e)>1$ and 
\begin{align*}
\Pr[S_n \geq 8 k H_n +  k \ln(1/\delta)] 
  & \leq
\exp\Big(-H_n \Big(1 + 8 + \frac{1}{H_n} \ln(1/\delta)\Big)\Big)
% = \exp(-(9 H_n + \ln(1/\delta)))) 
< \delta ,
\end{align*}
as claimed.
\end{proof}

\subparagraph*{Sample Thinning via Backwards Analysis.}
We cannot directly apply this insertion-only analysis in the expiration model.  Instead, we borrow ideas from a model where the size of the sliding-window can be chosen at query time~\cite{shi2021time}, and the original priority-based algorithm of Babcock \etal~\cite{BabcockMR02}.  
These ideas are present also in sliding-window algorithms for linear algebra~\cite{braverman2020near}, sampling Lewis weights~\cite{woodruff2025online}, and $k$-clustering~\cite{woodruff2023near}.
We adapt this here to show how to maintain a valid set of samples at any future time even with general expirations.

Denote the set of active items at time $t$ by $A(t)$,
and its size by $C(t)=|A(t)|$. 
At a given time $t_0$ (some place in the stream order), 
we can consider a future time $t\ge t_0$ (assuming no additional insertions).
Future times where $C(t) \leq k$, which we call the \emph{tail zone},
require us to store all their respective  active items deterministically.  

The expiration $k$-sample sketch works as follows, with three operations at each event:  insertion of new item, removal of expired items, and a thinning step to keep small the sketch size.  
On insertion of a new item $x_i$, it is assigned a priority value $u_i \sim \mathsf{Unif}(0,1]$ and stored in the sketch.  
%Upon insertion of item $x_i$, assign it a priority value $u_i \sim \mathsf{Unif}(0,1]$,
%store it in the sketch, remove from the sketch expired items,
%and in addition perform the following \emph{thinning step}:
The \emph{thinning step}
considers each item $x_j$ currently in the sketch,
and if at least $k$ items in the sketch expire no earlier than $E(x_j)$
and have priority value smaller than $u_j$,
then it removes this $x_j$ from the sketch.
Obviously, this thinning step is vacuous if the sketch currently stores less than $k$ items.
When queried about future time $t$, report the $k$ stored items with smallest $u_i$ values among those that expire after $t$.  
We do not analyze the processing time, but to improve it,
the sketch should store its items in sorted order of their expiration times,
or an even more sophisticated data structure.

\begin{theorem}
    \label{thm:GE-sample}
    The sample-thinned sketch can retrieve $k$ without-replacement uniformly random items, among the active set $A(t)$, for any future time $t$ before the tail zone.  In the tail zone, it retrieves all active items (which is at most $k$).  With probability at least $1-\delta$, the sketch never stores more than $k (9 \ln n + 8 + \ln(1/\delta))$ items.
\end{theorem}

\begin{proof}
There are 3 aspects to prove:  
(1) that the sketch can always produce the correct number of samples, $k$ (unless in the tail zone); 
(2) the size bound for the sketch; and
(3) the probabilistic guarantee for the full run of the sketch.  

We first show that an item $x_i$ cannot be removed at some time $t_0$
if it would be used as part of a sample at some future time $t\ge t_0$.
We show this in contrapositive direction.
Assume that $x_i$ is removed at time $t_0$,
hence the sketch must store (at least) $k$ items
that have both smaller priorities and no-earlier expiration.
It follows that at the future time $t$, these $k$ items would be active 
and should be part of the sample instead of $x_i$, which is a contradiction.

This argument shows that at any time $t$ before the tail zone,
the sketch stores the $k$ active items with largest priorities as desired.  
Furthermore, this argument covers also the tail zone,
i.e., items that would be active in a later tail zone are not thinned,
and will not be removed until they expire.

    %Items which expire among the last $k$ (the tail zone) are never removed.   
    %We only consider items for removal from the sketch if there are $k$ items which expire after it, thus for any future time $t$ (outside the tail zone) there are always $k$ active items maintained.  
    %Then $k$ items with smallest $u_i$ values represents a uniform without-replacement sample, so since this set of values is kept for each future $t$, the returned items satisfies the uniform sample requirements.  

    The space bounds follow from a ``backwards analysis'' of the without-replacement reservoir sampling captured in \cref{lem:k-res-sample}.  
    For a current time $t_0$, consider running the without-replacement insertion-only procedure on all active items $A(t_0)$, but processed in decreasing order of expirations, i.e., from largest to smallest value of $E(x_i)$, and retaining the $k$ items with smallest priorities (the $u_i$ value).
    By \cref{lem:k-res-sample} the total number of items ever stored by this procedure is at most $k(8 \ln n + 8 + \ln(1/\delta))$ with probability at least $1-\delta$.
    For every future time $t > t_0$ (before the tail zone), this procedure stores the $k$ items $x_{i'}$ with smallest priorities among those processed so far in that decreasing expiration order, i.e., among those satisfying $E(x_{i'}) \geq t$.
    Thus, an item $x_i$ processed at this time $t=E(x_i)$ is not stored
    only if there are already $k$ items (so their expiration times are at least as large) with smaller priorities.
    This is the exact same condition under which the item $x_{i'}$ would be thinned from the sketch.
    Hence, the space bound from \cref{lem:k-res-sample} applies also to our expiration sketch for active items $A(t_0)$.

    The space bound from \cref{lem:k-res-sample} holds for any single time $t_0$,
    and we need to invoke it separately for each possible $t_0$,
    because each invocation applies to a different set of active items $A(t_0)$,
    which changes with each insertion or expiration.
%    \rnote{Actually, it suffices to consider insertions here, as expirations cannot increase the storage, see below.}
    To obtain a space bound for the entire execution,
    we apply a union bound after each of $n$ insertions
    (notice that in our algorithm an expiration cannot increase the space),
    which requires us to scale $\delta$ to $\delta/n$ and increases the space bound to $k(8\ln n + 8 + \ln(n/\delta))$.    
\end{proof}

There are three main differences here from the algorithm and analysis for the sliding-window model~\cite{BabcockMR02,shi2021time,braverman2020near,woodruff2025online}. 
First, we analyze the without-replacement case.  
Second, we need to maintain and update the set of items which will be among the samples at some future time,
whereas previous approaches were similar to the consistent model and only needed to compare an inserted item to the most recently retained items.    
Third, we observe that the analysis linking to the Harmonic numbers can still be applied; this is by simulating the process in the reverse order of expirations even if this does not align with the order of insertions.

\subsection{Implications of Uniform Sampling}
\label{sec:ImplicationsUniform}

\begin{figure}
    \centering
    \includegraphics[width=0.8\linewidth]{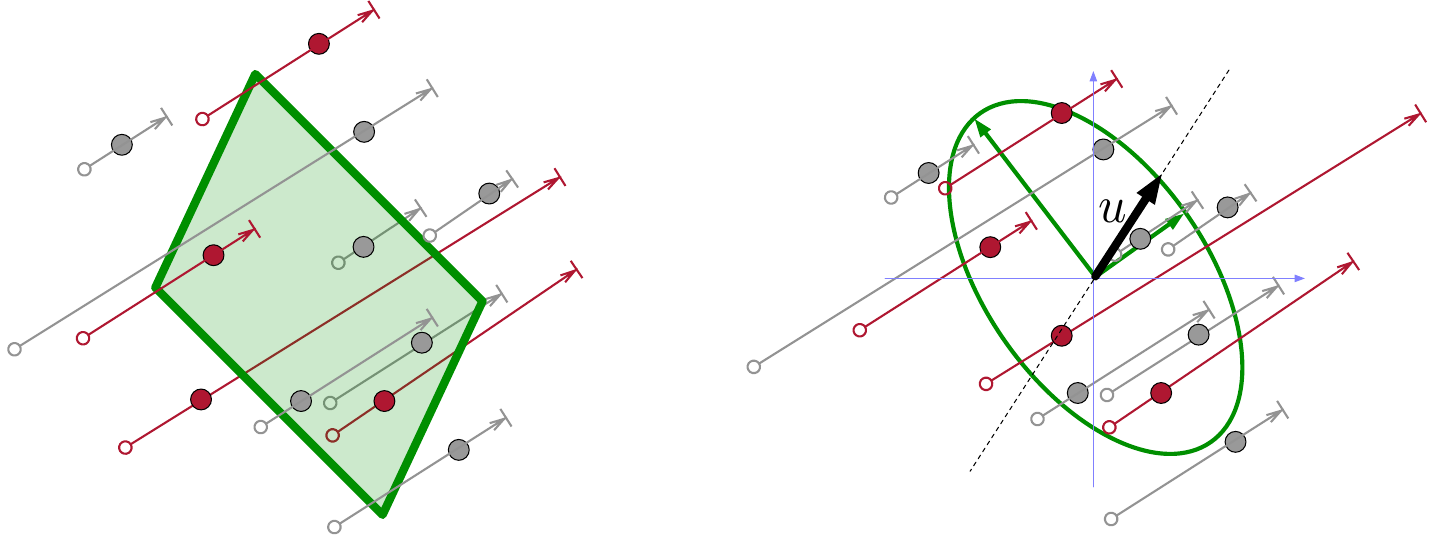}
    \caption{Illustration of geometric implications of (weighted) sampling of points $x \in \RR^2$ in expiration streams. 
    Lines (in a third dimension) through active points $x$ represent insertion time $S(x)$ (as $\circ$) and expiration time $E(x)$ (as $\dashv$).  
    Red points are selected by sample.  
    On the left, \emph{range counting}: 
    A green range admits approximate counting of the active points using the red sample.
    On $n$ points, the error is additive $\eps n$.    
    On the right, \emph{matrix covariance}:  
    A green ellipse is the best fit to the red sample among the active points,  
    which geometrically means that in every direction $u$, the squared distance (from the origin)
    to the ellipse approximates the average squared distance of all points $X$, 
    if they are projected onto $u$ (this is $\|X u\|^2/n$). 
    The error is additive $\eps \|X\|_F^2/n$,  
    where $\|X\|_F^2$ is the sum of squared distances of the points to the origin. 
    %\andre{This looks and sounds nice, but I don't really understand it.
      %Robi: I touched up the wording as I like this connection of our results to geometry. }
  }
    \label{fig:GE-geom}
\end{figure}

Many streaming results follow from maintaining one or a few uniformly random samples,
and we document a number of them in the rest of this subsection, also see Figure~\ref{fig:GE-geom}.

These results rely on an algorithm that produces, at any point in the stream,
$k$ with-replacement samples from $A(t)$ for any future time $t$,
and fails (e.g., exceeds its space bound) with probability at most $\delta>0$,
for any desired parameters $k$ and $\delta$.
We explain below how to design this sampling algorithm 
when $k > c\log(1/\delta)$ for sufficiently large but fixed $c > 1$,
which indeed holds in all our intended applications. 
We can simply run $k$ independent copies of \cref{thm:GE-sample},
each using parameter $k_0 = 1$ (producing $1$ item each). 
This increases the expected sketch size by factor $k$,
but we actually need to bound the sketch size with probability $1-\delta$.
One simple approach is to terminate any copy
that exceeds sketch size $\calO(\log (n/\delta'))$ for $\delta'=\delta/k$;
then each copy terminates (separately) with probability at most $\delta'$,
and by a union bound, the overall failure probability is $k\delta'=\delta$.
When $k\leq n$, the total sketch size is $O(k\log(n/\delta')) = O(k\log n)$,
and when $k>n$ it is even smaller because we can simply store all $n$ items. 
Another approach is to build on the proof of \cref{thm:GE-sample};
we can modify it to use a single Chernoff bound for all the $k$ copies together,
and show that the total sketch size is $\calO(k\log n)$
with probability at least $1-e^{-\Omega(k\log n)} > 1-\delta$.
Alternatively, we can run $3k$ independent copies of \cref{thm:GE-sample},
and terminate each copy that exceed sketch size $\calO(\log (n/\delta_0))$
for $\delta_0=1/3$;
by a Chernoff bound, with probability $1-e^{-\Omega(k)} \geq 1-\delta$,
at most $2k$ copies are terminated and we obtain $k$ samples as needed,
while the total sketch size is $\calO(k\log n)$. 
To streamline the exposition, we do not repeat these arguments for each application below.

\subparagraph*{Quantile Sketch with Expirations.}
Recall that the quantile sketch of a set of items $X \subset \RR$ allows one to approximate the fraction $F_v = |\{x \in X \mid x \leq v\}|/|X|$ of all items less than a query value $v \in \RR$ for all possible queries.
We say that it achieves \emph{$\eps$-additive error}
if its reported estimate $\hat F_v$
satisfies $|\hat F_v - F_v| \leq \eps$.
It is well known~\cite{agarwal2013mergeable,li2001improved,DKW56} that a uniformly random sample of size $k = \calO((1/\eps^2) \log (1/\delta))$ induces such a sketch,
% \rnote{don't we need factor $\log(1/\eps)$ to succeed on all queries (when $\delta$ is fixed)?}
% \jeff{Correct.  It follows from the refined VC-dimension arguments of LLS01~\cite{li2001improved}, but also can be seen via a classic Dvoretzky-Kiefer-Wolfowitz~\cite{DKW56} paper about Kolmogorov-Smirnov that predates VC notions.  I added these references.}
where the estimate $\hat F_v$ is just $1/k$ times the number of samples that are no larger than $v$.  
Thus, invoking \cref{thm:GE-sample} yields the following result.

\begin{theorem}
\label{thm:GE-quantile}
In an expiration stream consisting of $n$ items in $\RR$,
one can maintain a sketch for $\eps$-additive error quantiles,
that uses $\calO((1/\eps^2)\log n\log(1/\delta))$ space,
and with probability at least $1-\delta$,
it succeeds throughout the stream, on all queries, and all future times $t$.
\end{theorem}

To be clear, this bound considers all events parameterized by
a time $t_0$ during the stream, a future time $t\ge t_0$ given what has already arrived by time $t_0$, and a quantile query $v\in \RR$.
It guarantees success in \emph{all of these events together} with probability at least $1-\delta$.

\subparagraph*{Approximate Range Counting with Expirations.}
The \emph{range counting} problem considers a range space $(X,\Rcal)$ of a ground set $X$ and a family $\Rcal$ of subsets which can be queried.
Common examples include $X \subset \RR^d$ (or even $X = \RR^d$) and subsets that are defined geometrically,
say by halfspaces, axis-aligned rectangles, or balls,
meaning that we associate a subset $R \in \Rcal$ (called a range) with the shape defining it.  
For each query range $R \in \Rcal$ (which is often parameterized geometrically),
we count the number of points from $X$ in $R$, denoted $D(R) = |R \cap X|$.
A sketch $S$ for \emph{$\eps$-approximate range counting} allows the same queries,
but provides for every query $R \in \Rcal$ an estimate $S(R)$
that satisfies $|S(R) - D(R)| \leq \eps |X|$.  

A range space $(Y,\Rcal)$ is \emph{shattered} if for every subset $Z \subset Y$ there exists $R \in \Rcal$ such that $R \cap Y = Z$.
The \emph{VC-dimension} $\nu$ of $(X,\Rcal)$ is size of the largest subset $Y \subset X$, for which $(Y,\Rcal)$ is shattered.
It is classically known~\cite{vapnik1971uniform,li2001improved} that
if a range space $(X,\Rcal)$ has bounded VC-dimension $\nu$,
then with probability $1-\delta$,
a random sample $X' \subset X$ of size $k = \calO((1/\eps^2)(\nu + \log(1/\delta))$
forms a sketch $S$ for $\eps$-approximate range counting,
which on a query $R \in \Rcal$ reports $S(R) = |X|/k \cdot |R \cap X'|$.
For many common geometric shapes, the VC-dimension is linear in the dimension: 
for halfspaces $\nu=d+1$, 
for axis-aligned rectangles $\nu = 2d$, and
for balls $\nu = d+1$.  

To obtain this bound in an expiration stream, we can invoke \Cref{thm:GE-sample} to get a sample of size $k$, but we also need to multiply $|R \cap X'|$ by $|X|/k$ where in this setting $|X| = C_t$ is the count of active items at time $t$.
We cannot compute $C_t$ exactly in sublinear space by \cref{lem:hard-count},
however, by \cref{thm:GE-counter}, with probability at least $1-\delta$,
we can compute a relative $(1+\eps')$-approximation $\hat C_t$ of $C_t$ using $\tO((1/\eps) \log(\eps n) \log^3 \log(1/\delta))$ space.
If we choose $k = \calO((1/\eps)^2(\nu + \log(1/\delta))$, % to obtain $\eps' = \eps/2$ error,
then our approximate range counting sketch has
space bound that is dominated by the size of the sampling sketch,
and its estimate is $\hat S(R) = (\hat C_t / k)|R \cap X'| = (1 \pm \eps) S(R)$.  Then 
\[
|\hat S(R) - D(R)| \leq |S(R) - D(R)| + \eps S(R) \leq 2 \eps|X|, 
\]
which is our desired bound (up to scaling $\eps$ and $\delta$).
Ultimately, this implies the following result.

\begin{theorem}
\label{thm:GE-ARC}
Let $(X,\Rcal)$ be a range space with bounded VC-dimension $\nu$,
and suppose $X$ is presented as an expiration stream of $n$ points in $\RR^d$. % for $d=O(1)$.
Then one can maintain a sketch for $\eps$-approximate range counting,
that uses $\calO((\nu/\eps^2)d\log n \log(1/\delta))$ space, 
and with probability at least $1-\delta$
succeeds throughout the stream, on all queries, and for all future times $t$.  
\end{theorem}

\subparagraph*{PAC Learning with Expirations.}
An important case of the above construction is when the range space $(\X,\Rcal)$,
which has bounded VC-dimension, % \cite{vapnik1971uniform}
represents a function class for a machine learning problem.
Most classically, $\X = \RR^d$ where $\Rcal$ is the family of halfspaces,
and each observed $x_i \in \X$ has an associated binary label $y_i \in\{-1,+1\}$; this
encodes the linear classification problem.
Let $\{(x_1, y_1), (x_2, y_2), \ldots \}$ be the set of observed pairs,
denoted in short as $(X,y)$.
In this setting, one wants to find a range (a halfspace) $R^* \in \Rcal$
that does the best job of separating the labels, e.g., minimizing label mistakes.  That is, let 
\[
h_R(x_i) = \begin{cases}+1 & \text{ if } x_i \in R \\ -1 & \text{ if } x_i \notin R \end{cases}
\]
be the halfspace's prediction function,
and let $R_X^*$ be a range that minimizes the number of times $h_R(x_i) \neq y_i$ over all $(x_i,y_i) \in (X,y)$.  

Generalization bounds consider $(X,y)$ that is drawn from an unknown distribution $\mu$,
and use this $(X,y)$ is to choose a range $R$ that is likely to predict
the label of a ``test'' pair $(x_j, y_j)$ drawn from the same distribution $\mu$,
i.e., predict $y_j$ by $h_R(x_j)$.
A fundamental method, called empirical risk minimization (ERM),
is to choose a range that has the best fit to the observations, i.e.,
\[
  R_X^*
  = \arg \min_{R \in \Rcal} \sum_{(x_i, y_i) \in (X,y)} \mathbbm{1}_{ \{h_R(x_i) \neq y_i\} }
  = \arg \min_{R \in \Rcal} \Pr_{(x_i, y_i) \sim (X,y)} [h_R(x_i) \neq y_i] ,
\] 
where the last formulation views $(x_i, y_i)$ as drawn uniformly from the pairs in $(X,y)$.
%This is an equivalent way to define 
In comparison, a range that is optimal for the distribution $\mu$ is given by
\[
R_\mu^* = \arg \min_{R \in \Rcal} \Pr_{(x_j, y_j) \sim \mu} [y_j \neq h_R(x_j)].
\]
Classic learning theory results~\cite{vapnik1971uniform,li2001improved}
show that for sufficiently large
$|X| = \calO((1/\eps^2)(\nu + \log(1/\delta))$,
with probability $1-\delta$,  
\[
 \left| \Pr_{(x_j,y_j) \sim \mu}[h_{R_X^*}(x_j) \neq y_j]
- 
 \Pr_{(x_j,y_j) \sim \mu}[h_{R_\mu^*}(x_j) \neq y_j] \right|
 \leq \eps.
\]
That is, the range $R_X^*$ optimized over $(X,y)$ is $\eps$-close in probability of misclassification to the range $R_\mu^*$ optimized over the true distribution $\mu$.  
A sample $(X,y)$ that is large enough to satisfy the above condition,
will be called \emph{$(\eps,\delta)$-PAC} on the range space $(\X, \Rcal)$.

% We now consider sketches $S_{(X,y)}$ of
Now suppose that each observation $(x_i,y_i) \in (X,y)$ has an expiration time. 
An \emph{expiration sketch} $S_{(X,y)}$ allows a query about a future time $t$
and produces a prediction function $\hat h_t \leftarrow S_{(X,y)}(t)$,
that approximates the best range prediction function $h_t^*$ (for some range space $\Rcal$) for the active set $A(t)$ of $(X,y)$ at that future time $t$.
A randomized sketch satisfies that with probability at least $1-\delta$,
\[
 \left| \Pr_{(x_i,y_i) \sim (X,y)}[\hat h_t(x_i) \neq y_i]
- 
 \Pr_{(x_i,y_i) \sim (X,y)}[h_t^*(x_i) \neq y_i] \right|
 \leq \eps.
\]
An expiration sketch that satisfies this for all future times $t$
is called an \emph{$(\eps,\delta)$-PAC sketch} for the range space $(\X,\Rcal)$.  

This expiration sketch does not directly compare to $\mu$,
only to the set of observations $(X,y)$.
But if we assume that $A(t)$, the active set of the observations at time $t$,
is drawn from some distribution $\mu(t)$ and is $(\eps,\delta)$-PAC with respect to the function class $\Rcal$,
then by the triangle inequality this sketch is $(2\eps,2 \delta)$-PAC of this unknown $\mu(t)$.  
Again, invoking \Cref{thm:GE-sample},
one can maintain a sketch that achieves these learning bounds under the expiration model.  

\begin{theorem}
\label{thm:GE-learning}
Let $(X,\Rcal)$ be a range space with bounded VC-dimension $\nu$,
and suppose $(X,y)$ is presented as an expiration stream of $n$ points in $\RR^d$. % for $d=O(1)$.
Then one can maintain an $(\eps,\delta)$-PAC sketch
that uses space $\calO((\nu/\eps^2)d \log n\log(1/\delta))$,
and succeeds throughout the stream for all future times $t$.
\end{theorem}

%\rnote{I touched up the text from here till the end of the subsection.}\jeff{looks good -- made one minor English change}

Another popular way to pose the classification problem
is to define a convex loss function $f_X$,
%instead of measuring how often the classifier predicts a wrong label. 
where gradient descent can be used for empirical risk minimization, 
i.e., finding an optimal classifier for the observations. 
Consider linear classification, where 
$h(x_i) = \sgn(\langle w, x_i \rangle)$
using the halfspace's normal direction $w$ (and no offset, for simplicity). 
The \emph{logistic loss} at each labeled points $(x_i, y_i)$ is defined by
\[
\ell(x_i, y_i; w) = \ln(1 + \exp(y_i \langle w, x_i \rangle)).
\]
and the total loss over the entire data set $(X,y)$ is
$f_X(w) = \sum_{(x_i,y_i) \in X} \ell(x_i, y_i; w)$.
Now, an optimal classifier corresponds to $w \in \RR^d$ that minimizes $f_X(w)$,
and this task is called \emph{logistic regression}.  

One can augment the loss function with a regularization term, which accomplishes two things.
First, it biases the solution $w$ to be smaller and simpler (e.g., sparse).
Second, it enables us to create a coreset with just a uniform random sample.  
We formalize this with an $\ell_2$ regularizer; define the loss function
$\ell_\beta(x_i, y_i; w) = \ell(x_i, y_i; w) + \frac{1}{\beta}\|w\|_2^2$,
and then the total loss is
$f_{X,\beta}(w) = \sum_{(x_i,y_i) \in X} \ell_\beta(x_i, y_i; w)$.  
A uniform random sample $Q$ of $X$ of size $k = \calO((\beta^3/\eps^2)\log(1/\delta))$
will, with probability at least $1-\delta$, satisfy 
\begin{equation} \label{eq:CoresetRegression}
  \forall w \in \RR^d,
  \qquad
  |f_{X,\beta}(w) - f_{Q,\beta}(w)| \leq \eps f_{X,\beta}(w).
\end{equation}
In fact, this holds not only for logistic loss,
but also for sigmoid, SVM, and ReLu variants~\cite{alishahi2024no}.
We call such $Q\subset X$ satisfying \eqref{eq:CoresetRegression}
an \emph{$\eps$-coreset for $\beta$-regularized logistic regression}.  
As a result, invoking \cref{thm:GE-sample} provides the following result in expiration streams.

\begin{theorem}
\label{thm:GE-logistic-reg}
Consider a binary labeled data set $(X,y) \subset \RR^d \times \{-1,+1\}$, with $\|x\| \leq 1$ for each $x \in X$,
and suppose $(X,y)$ is presented as an expiration stream of size $n$.
Then one can maintain a sketch for $\eps$-approximate $\beta$-regularized logistic regression,
that uses $\calO((\beta^3/\eps^2)d \log n\log(1/\delta))$ space,
and with probability at least $1-\delta$
succeeds throughout the stream, on all queries $w$, and for all future times $t$.  
\end{theorem}

The functional approximation~\eqref{eq:CoresetRegression} implies that
the minimizer $w_Q^* = \arg\min_{w \in \RR^d} f_{Q,\beta}(w)$ of $f_{Q,\beta}$
achieves comparable value to the true optimum $w_X^* = \arg\min_{w \in \RR^d} f_{X,\beta}(w)$.
More precisely, $w_Q^*$ satisfies $f_{X,\beta}(w_Q^*) \geq f_{X,\beta}(w_X^*)$ and
\[
f_{X,\beta}(w_Q^*) \leq \frac{1}{1-\eps} f_{Q,\beta}(w_Q^*) \leq \frac{1}{1-\eps} f_{Q,\beta}(w_X^*) \leq \frac{1+\eps}{1-\eps} f_{X,\beta}(w_X^*),
\]
and thus optimizing $f_{Q,\beta}$ (even approximately) optimizes also the original $f_{X,\beta}$, up to small relative error.

\subparagraph*{Kernel Density Estimate Sketch with Expirations.}

%A kernel density estimate of a point set $X \subset \RR^d$ uses
Let $K : \RR^d \times \RR^d \to \RR$ be a kernel function,
which captures the similarity between two points,
A classical example is the Gaussian kernel, defined by $K(x,p) = \exp(-\|x-p\|^2)$.
One can normalize the kernel so it is a probability distribution over one of its arguments, or add a bandwidth parameter,
but this will be immaterial for us up to constants depending on these choices,
hence we normalize it so that $\max_{x} K(x,p) = K(p,p) = 1$. 
The \emph{kernel density estimate} of a point set $X\subset \RR^d$
is a function $\kde_X:\RR^d\to\RR$, 
whose value at a point~$p$ is the average of the kernel values at $p$ with all points in $X$, i.e., $\kde_X(p) = \frac{1}{|X|}\sum_{x \in X} K(p,x)$.  
An \emph{$\eps$-KDE coreset} of a point set $X$ is another point set $Q \subset \RR^d$ so that
%for every query $p$ that $|\kde_X(p) - \kde_Q(p)| \leq \eps$.
\[
  \forall p\in\RR^d,
  \qquad 
  |\kde_X(p) - \kde_Q(p)| \leq \eps, 
\]
and in this case we also say that it \emph{$\eps$-approximates $\kde_X$}.  
It is known (c.f., \cite{phillips2020near}) that a random sample $Q$ from $X$ of size $k = \calO((1/\eps^2)\log(1/\delta))$ provides a $\eps$-KDE coreset with probability at least $1-\delta$.
Storing this coreset requires $\calO(dk)$ space.
Then again, as a direct result of \cref{thm:GE-sample},
we can maintain this in an expiration stream.  

\begin{theorem}
\label{thm:GE-KDE}
Consider a point set $X \subset \RR^d$,
and suppose $X$ is presented as an expiration stream of size $n$.
Then one can maintain a sketch for $\eps$-approximating $\kde_X$,
that uses space in $\calO((d/\eps^2)\log n\log(1/\delta))$,
and with probability at least $1-\delta$
succeeds throughout the stream, on all queries, and for all future times $t$. 
\end{theorem}

\subsection{Weighted Sampling}
\label{sec:weighted}

Our next goal is to extend \cref{thm:GE-sample} to the weighted case,
i.e., items have positive weights and the probability to sample each item is proportional to its weight,
however we restrict attention to producing one sample, i.e., $k=1$.

We start with extending \cref{lem:k-res-sample},
which handles insertions-only streams, to the weighted case (but only for $k=1$). 
The algorithm is similar except for how we choose the priorities $u_i$,
which follows~\cite{Cohen97,CK07,CK08}.
More precisely, when item $i$ with weight $w_i>0$ arrives, 
generate for it a random priority $u_i \sim \mathsf{Exp}(w_i)$,
and maintain the item $i$ with smallest priority $u_i$ seen so far
(i.e., a reservoir containing a single item).
Recall that we denote the total weight of all items by $W = \sum_{i=1}^n w_i$.

We first verify that the sample maintained in the reservoir has the correct distribution.
Consider the reservoir after $n$ items have arrived 
(the argument is similar for every $n'\leq n$). 
The sample stored in the reservoir is item $i\in[n]$
if $u_i$ is smaller than $\min_{j\neq i} u_j$.
We now use two well-known and easy-to-verify properties of the exponential distribution.
First, it is min-stable, which implies that
$\min_{j\neq i} u_j$ is distributed as $\mathsf{Exp}(\sum_{j\neq i} w_j)$,
and second $\Pr[ \mathsf{Exp}(a) < \mathsf{Exp}(b)] = \tfrac{a}{a+b}$.
We obtain that the sample stored in the reservoir is item $i$ with probability
\begin{equation} \label{eq:weighted-1-sample1}
  \Pr[u_i < \min_{j\neq i} u_j]
  = \frac{w_i}{ w_i + \sum_{j\neq i} w_j}
  = \frac{w_i}{W} .
\end{equation}

\begin{lemma}
\label{lem:weighted-1-sample}
In insertion-only streams,
where each item $i$ has weight $w_i\geq 1$ and $W\ge 2$,
the above weighted-sampling algorithm maintains a sample ($k=1$)
from the items seen so far, with probability proportional to their weight. 
Moreover, the total number of items ever stored (possibly temporarily)
in the reservoir is $\calO(\log W)$ in expectation,
and it is $\calO(\log (W/\delta))$ with probability at least $1-\delta$. 
\end{lemma}

\begin{proof}
Let $X_i$ be an indicator for the event that item $i$ is stored in the reservoir when it arrives,
and thus $S_n = \sum_{i=1}^n X_i$ is the total number of items ever store the reservoir.
Denote for convenience $W_i := \sum_{j\leq i} w_j$, hence $W_n=W$. 
Observe that $X_1=1$ deterministically,
and for $i\ge 2$, by our calculation in~\eqref{eq:weighted-1-sample1}, 
\[
  \Pr[X_i = 1]
  = \tfrac{w_i}{W_i}
  = \int_{W_{i-1}}^{W_i} \frac{1}{W_i} dt
  \le \int_{W_{i-1}}^{W_i} \frac{1}{t} dt .
\]
Now the claimed expectation bound follows by
\[
  \E[S_n]
  = \sum_{i=1}^n \E[X_i]
  % = 1 + \sum_{i=2}^n \int_{W_{i-1}}^{W_i} \frac{1}{t} dt
  = 1+ \int_{W_1}^{W_n} \frac{1}{t} dt
  \leq 1 + \ln \tfrac{W_n}{W_1}
  \leq 1 + \ln W . 
%  \leq \calO(\log W). 
\]

To prove the high-probability bound, it suffices to show that for every $i\ge2$, 
\begin{equation} \label{eq:weighted-1-sample2}
  \Pr[X_i=1 \mid X_1,\ldots,X_{i-1}] = \tfrac{w_i}{W_i}. 
\end{equation}
Indeed, this implies that $S_n=\sum_{i=1}^n X_i$ is a Martingale,
and we can still apply a Chernoff bound for sufficiently large fixed $c>1$
and every $\delta>0$, 
to obtain 
\[
  \Pr[S_n \geq c\log (W/\delta)]
  \leq e^{-\Omega(c\log (W/\delta))}
  \leq \delta . 
\]

We actually prove a stronger statement than~\eqref{eq:weighted-1-sample2},
by essentially adding more random variables to the conditioning on the LHS.
Let $A$ be the permutation of $[i-1]$ that represents the relative ordering of $u_1,\ldots,u_{i-1}$,
i.e., $u_{A(1)} < \cdots < u_{A(i-1)}$,
assuming throughout that the priorities are distinct, which occurs with probability $1$.
Observe that the random value of $A$ completely determines
the value of $X_1,\ldots,X_{i-1}$, and therefore proving 
\begin{equation} \label{eq:weighted-1-sample3}
  \forall a,
  \qquad 
  \Pr[X_i=1 \mid A=a ] = \tfrac{w_i}{W_i} 
\end{equation}
would imply~\eqref{eq:weighted-1-sample2} by the law of total probability.
To prove this, we rewrite the LHS using Bayes' rule, as
\begin{equation} \label{eq:weighted-1-sample4}
  \Pr[X_i=1 \mid A=a ]
  = \frac { \Pr[X_i=1] \cdot \Pr[A=a \mid X_i=1] }{ \Pr[ A=a ] }.
\end{equation}
Now consider $\Pr[A=a \mid X_i=1, u_i]$,
where we condition on $X_i=1$, which is the event that $u_i < \min_{j<i} u_j$,
and also on $u_i$, which is needed for technical reasons.
%which we saw occurs with probability $w_i/W_i$.
By the memoryless property of the exponential distribution,
each $u_j-u_i$, for $j<i$, has distribution $\mathsf{Exp}(w_j)$,
and notice that these $i-1$ random variables are mutually independent.
The permutation $A$ represents the relative ordering of $u_1,\ldots,u_{i-1}$,
which does not change if we translate these values by $-u_i$,
so if we define $\bar{u}_j = u_j-u_i$ and $\bar A$ as the permutation of $[i-1]$
that represents the relative ordering of $\bar{u}_1,\ldots,\bar{u}_{i-1}$,
then clearly $\bar{A} = A$.
We thus have 
\[
  \Pr[A=a \mid X_i=1, u_i]
  = \Pr[\bar{A}=a \mid X_i=1, u_i]
  = \Pr[\bar{A}=a]
  = \Pr[A=a] ,
\]
%where the final equality is because $\bar{u}_j$ have the same distribution law as $u_j$.
and by the law of total probability also 
$\Pr[A=a \mid X_i=1] = \Pr[A=a]$.
Plugging this into~\eqref{eq:weighted-1-sample4}
proves~\eqref{eq:weighted-1-sample3},
and this completes the proof of \cref{lem:weighted-1-sample}.
\end{proof}

The extension to expiration streams follows the same arguments
as detailed in \cref{thm:GE-sample} for uniform sampling. 
That is, at each time $t_0$ during the stream,
the size of the sketch needed is the same as the total number of items ever kept
when processing an insertion-only stream consisting of the items active at time $t_0$ in reverse order of expirations.
For the weighted case, this is upper bounded by applying \cref{lem:weighted-1-sample},
namely, with probability $1-\delta/n$ the size of the sketch at time $t_0$
is $\calO(\log (Wn/\delta)) = \calO(\log (W/\delta)$,
and we can then apply a union bound over the $n$ insertions in the stream.
We can further extend this to maintaining $k$ with-replacement samples,
by simply running $k$ independent copies, 
as described in \cref{sec:ImplicationsUniform}.
The next theorem states this more formally.

\begin{theorem}
\label{thm:GE-W-sample}
In an expiration stream,
where each item $i$ has weight $w_i\geq 1$ and $W\ge 2$,
running $k$ independent copies of a sample-thinned sketch for weighted sampling,
can retrieve $k$ with-replacement samples among the active set $A(t)$,
for any future time $t$.
With probability at least $1-\delta$,
the sketch stores during the execution at most $\calO(k \log (W /\delta))$ items. 
\end{theorem}

\subsection{Implications of Weighted Random Sampling}
\label{sec:wGE-implications}

\cref{thm:GE-W-sample} allows to extend our results thatuse uniform random sampling
to variants that require sampling items proportionally to their weight.  
For instance, it is natural to consider approximate range counting
where each point has a weight, and a query asks for the total weight of a range.
One can also consider weighted versions of quantile approximation.
% \rnote{I don't see frequency in our list of results at beginning of the section.}
% \jeff{It follows directly from the quantile bound.  Arbitrarily order elements.  Consider two quantile queries just before and just after an item $j$.  Each have additive error $\leq \eps n$, so the difference has error $\leq 2 \eps n$, adjust to $\eps \to \eps/2$, and you get the frequency bound.  The MG and Count-Min improve this by factor $1/\eps$ space, so typically not mentioned.  But gives the same bound in expiration setting, since counting requires extra $\calO(1/\eps)$.  I wanted to have a direct application of counting, so I did not go into this.}\jeff{We could just remove this to avoid confusion if you prefer.}
%\rnote{I removed the mention of frequency to avoid potential confusion, although it's a nice trick, I wasn't aware of it!}
This also leads to new results that critically leverage weighted sampling,
as we demonstrate next in the context of matrix sketching.

\subparagraph*{Matrix Sketching with Expirations.}

There are many types of matrix sketching,
where one gets as input a ``large'' matrix $A \in \RR^{n \times d}$
and aims to approximate it with a ``small'' matrix $B \in \RR^{\ell \times d}$,
namely, with $\ell$ that is bounded.
In the expiration row-streaming setting,
the rows $a_i$ of $A$ arrive one at a time, and each is associated with an expiration time,
and the goal is to maintain a small matrix $B$ that approximates
the active rows of $A$ (those seen so far and have not yet expired),
rather than all of $A$.

Regardless of the streaming model,
there are several variants of the approximation guarantees one could consider between $B$ and $A$ (or its active part) \cite{MAL-035,TCS-060}.
One powerful variant is relative error,
where $\|B x\| \in (1\pm\eps) \|A x\|$ for all $x \in \RR^d$,
however we do not know how to maintain it in the expiration model. 
%\rnote{I don't understand the last sentence; it's too vague to have value as an open question.}
%\jeff{is this better now?}
%\rnote{thanks. I reworded/shortened it.}
%
We present below two results that approximate (the active part of) $A$
in two other ways: its low-rank structure and its covariance structure. 
% that ensure the best rank-$k$ subspace is approximately preserved and that the covariance structure is approximately preserved.
In both cases, the sketch produces $B \in \RR^{\ell \times d}$ with bounded $\ell$;
that is, we represent $B$ using $\Theta(d \ell)$ space
and do not seek a more compressed representation. 

Let $\pi_B(A)$ be the projection of the matrix $A$ onto the row space of $B$,
meaning that we project each row $a_i$ of $A$ onto the row space of $B$
(which still lives in $\RR^d$).
Let $A_k$ denote the ``best'' rank-$k$ approximation of $A$,
defined by $A_k := \arg\max_{B: \mathsf{rank}(B) =k} \|A - \pi_B(A)\|_F$,
and recall that it can be found using SVD. 
We say that $B$ achieves \emph{$\eps$-additive projection error} if 
\[
\|A - \pi_B(A)\|_F^2 \leq \|A - A_k\|_F^2 + \eps \|A\|_F^2.
\]

Drineas \etal~\cite{Drineas06} showed that sampling $\ell = \calO(k/\eps^2)$ rows from $A$,
with probabilities proportional to their squared norms $\|a_i\|^2$,
which is the classic norm-sampling approach~\cite{Frieze04},
produces a matrix $B$ that achieves $\eps$-additive projection error with constant probability. 
By invoking \cref{thm:GE-W-sample} on the rows $a_i$ with $w_i = \|a_i\|^2$,
we obtain the following result.  

\begin{theorem}
\label{thm:GE-mat-proj}
Consider a matrix $A \in \RR^{n \times d}$ presented as an expiration row stream, and $\eps > 0$. 
Suppose each row $a_i$ in $A$ has $\|a_i\| \geq 1$,
and let $A(t)$ be the matrix formed by the rows active at time $t$. 
One can maintain a randomized sketch of size $\calO((dk/\eps^2)\log \|A\|_F^2)$
that, with constant probability, succeeds throughout the stream in the following:
given as query a future time $t$, it produces a matrix $B(t)$
achieving $\eps$-additive projection error with respect to $A(t)$.
\end{theorem}

Next, we consider an approximation of the covariance structure of a matrix $A$.
We say that $B$ achieves \emph{$\eps$-covariance error} if 
\[
  \forall x \in \RR^d, \|x\|=1,
  \qquad
  \Big| \|B x\|^2 - \|A x\|^2 \Big|  \leq  \eps \|A\|_F^2 ,
\]
which informally means that $A$ and $B$ are close in all directions $x\in\RR^d$.
It is also equivalent to a bound on the associated covariance matrices
\[
\|A^T A - B^T B\|_2 \leq \eps \|A\|_F^2.  
\]
Such a guarantee was shown, for instance, for the FrequentDirections sketch~\cite{FD16}. 

Desai \etal~\cite{Desai16} showed %\jeff{note that the proof is only in the arXiv version, the arXiv URL is in the reference -- the journal version had a page limit.  The proof is not that hard...} 
that sampling $\ell = \calO(d/\eps^2)$ rows of $A$
with probabilities proportional to their squared norms $\|a_i\|^2$
and rescaling them to each have squared norm $\|A\|_F^2/\ell$,
produces a matrix $B$ that achieves $\eps$-covariance error with constant probability.
Again invoking \cref{thm:GE-W-sample} on the rows $a_i$ with $w_i = \|a_i\|^2$,
we obtain the following result.  

\begin{theorem}
\label{thm:GE-mat-cov}
Consider a matrix $A \in \RR^{n \times d}$ presented as an expiration row stream, and $\eps > 0$. 
Suppose each row $a_i$ in $A$ has $\|a_i\| \geq 1$,
and let $A(t)$ be the matrix formed by the rows active at time $t$. 
One can maintain a sketch of size $\calO((d^2/\eps^2)\log \|A\|_F^2)$
that, with constant probability, succeeds throughout the stream in the following:
given as query a future time $t$, it produces a matrix $B(t)$
given query $t$ produces a matrix $B(t)$
achieving $\eps$-covariance error with respect to $A(t)$.
\end{theorem}

%%%%%%%%%%% bibliography
\bibliography{main}

\appendix
\section{Omitted Proofs}
\begin{lemma}
\label{lem:AlmostSmooth}
The diameter of points in a general metric space is $2$-almost-smooth.
\end{lemma}

\begin{proof}
Following \cite{KR22}, observe that the diameter is non-negative and monotonically increasing,
hence we need to prove that for $d=2$ and every disjoint segments $A,B,C$ of a stream, 
\begin{equation} \label{eq:AlmostSmooth}
  \diam(B) \leq d\cdot \frac{\diam(AB)\ \diam(BC)} {\diam(ABC)}
\end{equation}
whenever $\diam(AB),\diam(ABC)\neq 0$.
%We can view $A,B,C$ also as sets of points,
%because the diameter does not depend on the order of points in the stream.

Pick $x,y\in ABC$ that have the largest distance, i.e., $d(x,y) = \diam(ABC)$.
If $x,y\in BC$,
then $\diam(ABC) = d(x,y) = \diam(BC)$,
while by monotonicity $\diam(B) \leq \diam(AB)$,
and~\eqref{eq:AlmostSmooth} follows (even with $d=1$). 
If $x,y\in AB$,
then $\diam(ABC) = d(x,y) = \diam(AB)$,
and again by monotonicity, \eqref{eq:AlmostSmooth} follows (even with $d=1$). 
In the remaining case, $\{x,y\}$ intersects both $A$ and $C$,
so without loss of generality we assume $x\in A$ and $y\in C$.
Pick arbitrary $b\in B$, then by the triangle inequality 
\begin{align*}
  \diam(ABC)
  = d(x,y)
  \leq d(x,b) + d(b,y)
  &\leq 2 \max\{ d(x,b) , d(b,y) \}\\
  &\leq 2 \max\{ \diam(AB) , \diam(BC) \}. 
\end{align*}
Now regardless of which term is the maximizer on the RHS,
again \eqref{eq:AlmostSmooth} follows by monotonicity of the diameter. 
\end{proof}

\end{document}